\documentclass{article}
\pdfoutput=1
\usepackage{fullpage} % letter paper and margins
\usepackage{xspace,xcolor,graphicx,tabularx} %basic enhancements
\usepackage{mathtools,amsthm,amssymb} %basic maths packages
\usepackage{enumitem} %better lists
\setenumerate[0]{label={\normalfont (\roman*)}}

\usepackage[T1]{fontenc}
\usepackage[UKenglish]{babel}
\usepackage[scaled=0.86]{helvet}
\usepackage{mathptmx}
\DeclareMathAlphabet{\mathcal}{OMS}{ntxm}{m}{n}
\usepackage{dsfont} %double stroke font
\let\mathbb\relax
\let\mathbb\mathds
\usepackage{stmaryrd} %extra symbols
\makeatletter
\renewcommand{\paragraph}{%
  \@startsection{paragraph}{4}%
  {\z@}{2.25ex \@plus 1ex \@minus .2ex}{-1em}%
  {\normalfont\normalsize\bfseries}%
}
\makeatother
\usepackage{authblk}

\definecolor{linkblue}{HTML}{001487}
\usepackage[colorlinks=true,allcolors=linkblue]{hyperref}
\usepackage{url}
\usepackage[nameinlink,capitalize,noabbrev]{cleveref}
\newtheorem{theorem}{Theorem}[section]
\newtheorem*{theorem*}{Theorem}

\newtheorem{lemma}[theorem]{Lemma}

\newtheorem{corollary}[theorem]{Corollary}
\theoremstyle{remark}
\newtheorem{remark}[theorem]{Remark}
\theoremstyle{definition}
\newtheorem{definition}[theorem]{Definition}

\newtheorem{protocol}{Protocol}
\newtheorem{process}{Process}
\numberwithin{equation}{section}

\newcommand{\setft}[1]{\textnormal{#1}}
\newcommand{\eps}{\epsilon}
\newcommand{\1}{\mathds{1}}
\newcommand{\id}{\setft{id}}
\newcommand{\C}{\ensuremath{\mathds{C}}}

\newcommand{\Z}{\ensuremath{\mathds{Z}}}
\newcommand{\bits}{\ensuremath{\{0, 1\}}}

\newcommand{\ot}{\ensuremath{\otimes}}
\newcommand{\deq}{\coloneqq}

\newcommand{\norm}[1]{\left\lVert#1\right\rVert}
\DeclareMathOperator{\pos}{Pos}
\DeclareMathOperator{\poly}{poly}
\DeclareMathOperator{\negl}{negl}

\DeclareMathOperator*{\E}{\mathds{E}}

\newcommand{\sth}{{\setft{~s.t.~}}}

\newcommand{\ket}[1]{|#1\rangle}
\newcommand{\bra}[1]{\langle#1|}

\DeclarePairedDelimiterX\braket[2]{\langle}{\rangle}{#1 \delimsize\vert #2}

\newcommand{\cA}{\ensuremath{\mathcal{A}}}
\newcommand{\cB}{\ensuremath{\mathcal{B}}}
\newcommand{\cC}{\ensuremath{\mathcal{C}}}

\newcommand{\cF}{\ensuremath{\mathcal{F}}}
\newcommand{\cG}{\ensuremath{\mathcal{G}}}
\newcommand{\cH}{\ensuremath{\mathcal{H}}}
\newcommand{\cI}{\ensuremath{\mathcal{I}}}

\newcommand{\cK}{\ensuremath{\mathcal{K}}}

\newcommand{\cQ}{\ensuremath{\mathcal{Q}}}
\newcommand{\cR}{\ensuremath{\mathcal{R}}}

\newcommand{\cX}{\ensuremath{\mathcal{X}}}

\DeclareSymbolFont{greekletters}{OML}{ntxmi}{m}{it}
\DeclareMathSymbol{\alpha}{\mathord}{greekletters}{"0B}
\DeclareMathSymbol{\beta}{\mathord}{greekletters}{"0C}
\DeclareMathSymbol{\gamma}{\mathord}{greekletters}{"0D}
\DeclareMathSymbol{\delta}{\mathord}{greekletters}{"0E}
\DeclareMathSymbol{\epsilon}{\mathord}{greekletters}{"0F}
\DeclareMathSymbol{\zeta}{\mathord}{greekletters}{"10}
\DeclareMathSymbol{\eta}{\mathord}{greekletters}{"11}
\DeclareMathSymbol{\theta}{\mathord}{greekletters}{"12}
\DeclareMathSymbol{\iota}{\mathord}{greekletters}{"13}
\DeclareMathSymbol{\kappa}{\mathord}{greekletters}{"14}
\DeclareMathSymbol{\lambda}{\mathord}{greekletters}{"15}
\DeclareMathSymbol{\mu}{\mathord}{greekletters}{"16}
\DeclareMathSymbol{\nu}{\mathord}{greekletters}{"17}
\DeclareMathSymbol{\xi}{\mathord}{greekletters}{"18}
\DeclareMathSymbol{\pi}{\mathord}{greekletters}{"19}
\DeclareMathSymbol{\rho}{\mathord}{greekletters}{"1A}
\DeclareMathSymbol{\sigma}{\mathord}{greekletters}{"1B}
\DeclareMathSymbol{\tau}{\mathord}{greekletters}{"1C}
\DeclareMathSymbol{\upsilon}{\mathord}{greekletters}{"1D}
\DeclareMathSymbol{\phi}{\mathord}{greekletters}{"1E}
\DeclareMathSymbol{\chi}{\mathord}{greekletters}{"1F}
\DeclareMathSymbol{\psi}{\mathord}{greekletters}{"20}
\DeclareMathSymbol{\omega}{\mathord}{greekletters}{"21}
\DeclareMathSymbol{\varepsilon}{\mathord}{greekletters}{"22}
\DeclareMathSymbol{\vartheta}{\mathord}{greekletters}{"23}
\DeclareMathSymbol{\varpi}{\mathord}{greekletters}{"24}
\DeclareMathSymbol{\varrho}{\mathord}{greekletters}{"25}
\DeclareMathSymbol{\varsigma}{\mathord}{greekletters}{"26}
\DeclareMathSymbol{\varphi}{\mathord}{greekletters}{"27}

\newcommand{\comm}{\mathsf{com}}
\newcommand{\tele}{\mathsf{tele}}

\DeclareMathOperator{\tr}{tr}
\DeclareMathOperator{\Tr}{Tr}
\DeclareMathOperator{\diag}{diag}

\usepackage{longfbox}

\newcommand{\capprox}{\overset{c}{\approx}}

\usepackage{mathrsfs}

\newcommand{\Gen}{\mathsf{Gen}}
\def\Eval{\mathsf{Eval}}

\newcommand{\secp}{\lambda}

\newcommand{\sk}{sk}
\newcommand{\clef}{\E_{sk} \leftarrow \Gen(1^\lambda)}

\newcommand{\pST}{\; \middle| \;}

\newcommand{\Nat}{\mathbb{N}} %natural numbers

\newcommand{\aux}{\mathsf{aux}}

\newcommand{\PPT}{\mathsf{PPT}}
\newcommand{\QPT}{\mathsf{QPT}}
\newcommand{\Enc}{\mathsf{Enc}}
\newcommand{\Dec}{\mathsf{Dec}}

\newcommand{\ct}{\mathsf{ct}}

\newcommand{\QHE}{\mathsf{QHE}}

\newcommand{\LQHE}{\mathsf{QHE}}

\newcommand{\QMA}{\mathsf{QMA}}

\newif\ifnotes\notestrue
\ifnotes
\newcommand{\znote}[1]{\textcolor{blue}{(Tina: #1)}}
\newcommand{\anote}[1]{\textcolor{red}{(Anand: #1)}}
\newcommand{\tm}[1]{\textcolor{teal}{(Tony: #1)}}
\else
\newcommand{\znote}[1]{}
\newcommand{\anote}[1]{}
\newcommand{\tm}[1]{}
\fi

\newcommand{\setr}[1]{|_{\{#1\}}}
\newcommand{\setrw}{\setr{w=W}}

\begin{document}

\title{Succinct arguments for QMA from standard assumptions \\ via compiled nonlocal games}

\date{}
\author[1]{Tony Metger\footnote{Email: tmetger@ethz.ch}}
\author[2]{Anand Natarajan\footnote{Email: anandn@mit.edu}}
\author[2]{Tina Zhang\footnote{Email: tinaz@mit.edu}}
\affil[1]{ETH Zurich}
\affil[2]{MIT}

\maketitle

\begin{abstract}
% Since the dawn of time humankind has searched for succinctness. It is an old dictum that brevity is the soul of wit. Paninian grammarians were reputed to `rejoice over the saving of half a syllable in a grammatical statement as over the birth of of a son.' Given this, it would not do for this abstract to be too prolix.
We construct a succinct classical argument system for $\QMA$, the quantum analogue of $\mathsf{NP}$, from generic and standard cryptographic assumptions. Previously, building on the prior work of Mahadev (FOCS '18), Bartusek et al.~(CRYPTO '22) also constructed a succinct classical argument system for $\QMA$.
However, their construction relied on  post-quantumly secure indistinguishability obfuscation, a very strong primitive which is not known from standard cryptographic assumptions. 
In contrast, the primitives we use (namely, collapsing hash functions and a mild version of quantum homomorphic encryption) are much weaker and are implied by standard assumptions such as LWE.
Our protocol is constructed using a general transformation which was designed by Kalai et al.~(STOC '23) as a candidate method to compile any quantum nonlocal game into an argument system. Our main technical contribution is to analyze the soundness of this transformation when it is applied to a \emph{succinct} self-test for Pauli measurements on maximally entangled states, the latter of which is a key component in the proof of $\mathsf{MIP}^* = \mathsf{RE}$ in quantum complexity.
\end{abstract}

\newpage 
{
\hypersetup{linkcolor=black}
\setcounter{tocdepth}{2}
\tableofcontents
}

\newpage

\section{Introduction}
Succinct verification of computation is a notion that has been extensively studied in the classical setting. A weak classical client may delegate a classical computation to a powerful server, and may then wish to check whether the server performed the computation correctly without having to compute the answer for itself. In this case, the client can ask the server to execute a \emph{succinct interactive argument}, in which the server (efficiently) convinces the client beyond reasonable doubt that the computation was performed correctly, and the client only has to do work scaling with $\poly \log T$ in order to be convinced, where $T$ is the time that it took to do the computation itself. The messages in this succinct interactive protocol should also be $\poly \log T$ in length.

Not long after it came to light that quantum algorithms could outperform the best known classical algorithms in certain computational tasks, the question was posed of whether a quantum prover could convince a classical verifier of the answer to a problem in $\mathsf{BQP}$ without requiring the classical verifier to simulate the computation itself. For certain problems, like factoring, a classical verifier can check correctness by exploiting the fact that the problem lies in $\mathsf{NP}$; however, $\mathsf{NP}$ is not known to contain $\mathsf{BQP}$, and for some problems this may be infeasible. This line of inquiry was initiated by Gottesman in 2004~\cite{aaronson-blogpost}, and has led to a long line of work on the problem now known as \emph{quantum verification}.

\paragraph{Succinct quantum verification with a single cryptographically bounded prover.}

Mahadev's work in 2018~\cite{mahadev2018classical} showed that it is indeed possible for an efficient quantum prover to convince a classical verifier of the answer to any problem in $\mathsf{BQP}$, given that the quantum prover is subject to certain (post-quantum) cryptographic assumptions. (In fact, her work also showed that it is possible for an efficient quantum prover to convince a classical verifier of the answer to any problem in $\mathsf{QMA}$, assuming the prover is given polynomially many copies of the witness state for the $\mathsf{QMA}$ problem.) Mahadev's quantum verification protocol inspired a slew of followup work in which her techniques were used to design other cryptographic quantum verification protocols with desirable additional properties, e.g.~the property of being non-interactive~\cite{BKVV20} or composable~\cite{GV19} or linear-time~\cite{zhang2022classical}. In 2022, Bartusek et al.~\cite{BKLMMVVY22} showed, \emph{assuming post-quantum $\mathsf{iO}$}, that some version of Mahadev's protocol can be made \emph{succinct}, in the same sense that we described in the opening paragraph: the classical verifier only needs to read messages that are $\poly \log n$ bits long, where $n$ is the size of the instance, and do work scaling with $\poly \log T + \tilde O(n)$, where $T$ is the time required to execute the verification circuit.

$\mathsf{iO}$, or indistinguishability obfuscation, is an immensely powerful and subtle primitive that has recently been constructed from a combination of several standard assumptions~\cite{JLS21}. However, some of these assumptions are \emph{not} post-quantum, and there is currently no construction of post-quantum $\mathsf{iO}$ from standard assumptions. Post-quantum $\mathsf{iO}$ is known to imply other elusive cryptographic objects, e.g.~public-key quantum money~\cite{zhandry2021quantum}, and constructing it from standard post-quantum assumptions remains a difficult and important open problem.

The essential difficulty, and the reason for the use of $\mathsf{iO}$ in~\cite{BKLMMVVY22}, is that Mahadev's approach to verification is in some sense a qubit-by-qubit approach, and requires $\Omega(\lambda)$ bits of communication (where $\lambda$ is the security parameter) \emph{for every qubit} in the prover's witness state, because the verifier needs to send the prover as many $\Omega(\lambda)$-sized public keys as the witness has qubits. As a result, Mahadev's approach is difficult to make succinct, since setting up the keys already requires at least $n \cdot \Omega(\lambda)$ bits of communication, where $n$ is the number of qubits in the prover's witness. The authors of~\cite{BKLMMVVY22} use $\mathsf{iO}$ in a clever way to compress the keys and thus reduce the amount of required communication to $\poly(\lambda) \cdot \poly\log n$; this is the bulk of their work. More specifically, the authors begin by constructing a \emph{question-succinct} (short questions, long answers) protocol for verifying $\mathsf{QMA}$ using $\mathsf{iO}$.
Then they present a general compiler which uses a recent post-quantum analysis of Killian's succinct arguments of knowledge~\cite{CMSZ,LMS} to turn any question-succinct protocol that satisfies certain properties into a fully succinct protocol.

\paragraph{Our result.}

Our main contribution in this paper is to construct succinct classical-verifier arguments for $\mathsf{QMA}$ from standard (and even relatively general) assumptions, without relying on post-quantum $\mathsf{iO}$. More specifically, we prove the following theorem:

\begin{theorem}[Informal version of~\Cref{thm:main-succinct-technical}]
\label{thm:main-informal}
Assume that a quantum levelled homomorphic encryption scheme exists which specialises to a classical encryption scheme when it is used on classical plaintexts.\footnote{In fact, we do not need all the properties of a typical QHE scheme: for example, we do not use the standard notion of \emph{compactness}, which says that decryption time cannot depend on the size of the circuit being evaluated. Instead, we only need a weak notion of compactness which says that classical ciphertexts encrypted under the QHE scheme should be classically decryptable (in any polynomial time, even if the decryption time depends on the evaluated circuit). We also expect that the weaker primitive of classical-client quantum blind delegation (in which interaction is allowed) would likely suffice. These weaker primitives could plausibly be instantiated from weaker assumptions than LWE, since they do not imply classical FHE, which is only known assuming LWE to date. For recent progress towards this, see~\cite{GV24}.} Assume also that post-quantum succinct arguments of classical knowledge exist.\footnote{These can be constructed from any collapsing hash function.} Then a constant-round classical-verifier argument system for any promise problem in $\mathsf{QMA}$ exists, in which:
\begin{enumerate}
\item the honest quantum prover runs in quantum polynomial time, given polynomially many copies of an accepting $\QMA$ witness state,
\item the completeness-soundness gap is a constant, and
\item the total communication required is of length $\poly\log n \cdot \poly \lambda$, where $n$ is the instance size and $\lambda$ is the security parameter. The verifier runs in time $\poly(\log T, \lambda) + \tilde O(n)$, where $T$ is the size of the $\mathsf{QMA}$ verification circuit. 
\end{enumerate}
\end{theorem}

The main advantage of our protocol compared with Bartusek et al.'s protocol~\cite{BKLMMVVY22} is that our protocol does not use post-quantum $\mathsf{iO}$, which at this time cannot be instantiated from standard assumptions. Even setting aside the issue of post-quantum $\mathsf{iO}$, however, we remark that the non-$\mathsf{iO}$ assumptions that our approach relies on are more generic than the Learning With Errors (LWE)--based assumptions which Bartusek et al.~use. For example, our approach avoids using the delicate `adaptive hardcore bit' property of LWE-based trapdoor claw-free functions (TCFs), which was introduced in~\cite{BCMVV18} and used in Mahadev's original verification protocol (as well as Bartusek et al.'s protocol).
The main primitive we rely on, quantum homomorphic encryption (QHE), can be constructed in its usual form from LWE without the adaptive hardcore bit assumption~\cite{mahadev2020classical}. Moreover, we do not in fact need all the properties of standard QHE: for example, we do not use the standard notion of \emph{compactness}, which says that decryption time cannot depend on the size of the circuit being evaluated. Instead, we only need a weak notion of compactness which says that classical ciphertexts encrypted under the QHE scheme should be classically decryptable (in any polynomial time, even if the decryption time depends on the evaluated circuit). This more general notion of \emph{non-compact QHE with classical decryption for classical ciphertexts} plausibly exists from assumptions other than LWE: for instance, \cite{GV24} represents recent progress in this direction. As such, our approach shows that the important primitive of quantum verification---and even succinct verification---may exist from a wider range of assumptions than LWE only. (In contrast, a large number of post-quantum primitives that use techniques from Mahadev's original verification protocol can only, as far as we can see, be constructed from LWE.)

We achieve \Cref{thm:main-informal} by combining powerful information-theoretic tools which originate in the study of nonlocal games (e.g.~those found in~\cite{delaSalle2022spectral}) with tools that cryptography offers (in particular, cryptographic succinct arguments of knowledge turn out to be very useful for us). The resulting protocol is (compared with the protocol designed by Bartusek et al.) a remarkably clean object which has a natural intuitive interpretation. The tool that allows us to combine self-testing techniques with cryptographic techniques is a \emph{compilation procedure} introduced by Kalai, Lombardi, Vaikuntanathan, and Yang~\cite{KLVY21}, which Natarajan and Zhang~\cite{NZ23} recently exploited in order to achieve classical-verifier quantum verification using a different approach from Mahadev's original approach.

\paragraph{A different approach to verification based on nonlocal games.}

Since Bell's historical observation~\cite{bell1964einstein} that there are certain \emph{nonlocal games} which quantum entangled players can win with higher probability than classical players, the \emph{entangled two-prover} model of computation has been a model of great interest in quantum complexity theory and quantum foundations~\cite{scarani2013device}. A nonlocal game is a game played between a single efficient classical referee (or verifier) and two or more unbounded players (or provers) who cannot communicate with each other but are allowed to share entanglement. The study of the computational power of nonlocal games (i.e.~what can the verifier compute efficiently with the help of the provers, if the verifier doesn't trust the provers?) has led to a fruitful line of work which, in particular, has shown that the verifier in this setting can decide any problem in $\mathsf{RE}$~\cite{JNVWY20}. In addition, it is known~\cite{Grilo17} that, even if the honest provers are required to be efficient, the verifier can still decide any problem in $\mathsf{BQP}$ (or $\mathsf{QMA}$, if one of the provers gets access to polynomially many copies of a witness). Put another way, quantum verification in the \emph{entangled two-prover} setting is known to exist.

In 2023 Natarajan and Zhang~\cite{NZ23} presented a reproof of Mahadev's result which took a different approach to her original approach, building on previous work on quantum blind delegation~\cite{mahadev2020classical} and the work of Kalai, Lombardi, Vaikuntanathan, and Yang~\cite{KLVY21}. Kalai, Lombardi, Vaikuntanathan, and Yang used quantum blind delegation (in particular, quantum homomorphic encryption) in order to design a \emph{compilation} scheme which maps any entangled two-prover proof system to a \emph{single-prover argument}, using cryptography to enforce the no-communication assumption between the provers. Kalai et al.~showed that their compilation scheme preserves quantum completeness and classical soundness, and Natarajan and Zhang showed that it also preserves quantum soundness for a certain restricted class of two-prover nonlocal games, which was sufficient to compile a two-prover quantum verification protocol into a single-prover cryptographic protocol and thus recover Mahadev's result.

From the point of view of designing succinct arguments, this approach is more attractive than Mahadev's original approach as a starting point, because the verifier only needs to send the prover a \emph{single} public key of length $\poly(\lambda)$ in order to allow it to do homomorphic evaluations. One might then hope to construct a succinct cryptographic verification protocol for $\mathsf{QMA}$ in the following way: start with a succinct \emph{two-prover} quantum verification protocol, pass it through the KLVY compiler, and prove soundness using similar techniques to those which Natarajan and Zhang used in~\cite{NZ23}. This approach avoids using $\mathsf{iO}$ entirely, because the KLVY compiler is `naturally' succinct when applied to a succinct protocol.

\paragraph{Succinct quantum verification in the entangled two-prover setting.}

It is therefore natural to ask whether \emph{succinct} quantum verification in the entangled two-prover setting is known. The answer to this question is---unfortunately---no, but for surprisingly complicated reasons. Below is a list of the partial results in this area which are known:
\begin{enumerate}
\item If the honest provers are allowed to be inefficient, and if the (classical) verifier is allowed to take $\poly n$ time, then there is a protocol with $\poly \log n$ total communication in the entangled two-prover setting to decide $\mathsf{QMA}$ (in fact, to decide all of $\mathsf{RE}$). This was shown by~\cite{NZ23free}.

Unfortunately, this result is not useful to us if our goal is to compile a succinct two-prover proof system into a succinct one-prover quantum verification protocol, since we want the honest prover to be efficient.

\item In a setting where the verifier interacts with \emph{seven} provers instead of two,~\cite{NV18} claimed to show that efficient-prover quantum verification of $\mathsf{QMA}$ is possible. However, the proof of this result had two substantial errors in it. One of these errors has been resolved by~\cite{ji2022tensorcodes}. The other one remains unresolved: see this erratum notice with an explanation of the error~\cite{NN24}.

Even assuming the errors in~\cite{NV18} can be fixed, a seven-prover protocol is not useful to us because the techniques from~\cite{KLVY21,NZ23} were only designed for nonlocal games with two provers.
It seems difficult to extend these techniques to a larger number of provers, which would be necessary to compile the seven-prover protocol from~\cite{NV18}.
\item Examining the proof of $\mathsf{MIP}^* = \mathsf{RE}$ from~\cite{JNVWY20} shows that it relies on two so-called \emph{compression theorems}: a \emph{question reduction theorem} which takes a two-prover nonlocal game with long questions (messages from the verifier to the provers) and maps it to a nonlocal game with exponentially smaller questions while preserving most other properties of the game, and an \emph{answer reduction theorem} which takes a two-prover nonlocal game with long answers (messages from the provers to the verifier) and maps it to a nonlocal game with exponentially smaller answers.

One would think that these theorems would make proving succinctness in the nonlocal setting easy. Unfortunately, these compression theorems come with caveats: in particular, the answer reduction theorem can only be applied to so-called \emph{oracularisable} protocols, and no one has come up with a two-prover verification protocol for $\mathsf{QMA}$ with efficient honest provers which satisfies this property. Moreover, even supposing that we had a protocol to which we could apply answer reduction, the answer reduction procedure itself happens to be so complicated and delicate that there is no clear way to analyse its soundness in the compiled setting, even given the techniques from~\cite{NZ23} and the additional techniques for compiling nonlocal games which have been developed since then~\cite{CMMNPSWZ24}.

Question reduction is both simpler and more lenient, however: while it has never been published, \emph{question-succinct} quantum verification for $\mathsf{QMA}$ in the two-prover setting can be elegantly obtained from known results~\cite{delaSalle2022spectral,Grilo17}. 
\end{enumerate}

\paragraph{The best of both worlds.}

The essential reason that two-prover succinct verification remains an open problem is that nonlocal \emph{answer reduction} is hard. The only known way to make the answers in a nonlocal game shorter is to use an `entanglement-sound' classical PCPP, and constructing this object is arguably the most technical and delicate part of the proof that $\mathsf{MIP}^* = \mathsf{RE}$. On the other hand, one can make the \emph{questions} in certain (useful) classes of nonlocal games shorter using only the elegant machinery of de la Salle~\cite{delaSalle2022spectral}, who simplified the question reduction theorems of~\cite{JNVWY20} by rephrasing them in terms of sampling from $\eps$-biased sets. Therefore, in the nonlocal world, \emph{question reduction} is now considered to be relatively easy, and \emph{answer reduction} remains hard.

In Bartusek et al.'s approach to succinct verification, meanwhile, the situation was just the opposite: shortening the questions in the Mahadev protocol using only cryptography was a significant challenge, and shortening the answers could be done using known techniques in a relatively black-box manner. Given that this is the case, one might hope to combine the Bartusek et al.~approach with the compilation approach in order that the strengths of each might cancel out the weaknesses of the other.

This is precisely what we do in this work. We construct a succinct verification protocol for $\mathsf{QMA}$ by firstly compiling, using the KLVY compiler, a \emph{question-succinct} two-prover protocol for $\mathsf{QMA}$, and then compressing the answers in a generic way using Bartusek et al.'s Killian-based compiler.

The success of this approach makes a case for using the KLVY compiler as a general way to translate techniques that are well-understood in the entangled two-prover world into the single-prover cryptographic world. Once this has been done, they can be combined with `natively' cryptographic techniques in order to marry the desirable properties of both. It seems plausible that many of the existing results in the sphere of classical-client quantum delegation and verification could have been obtained in a more unified way and from milder or more generic assumptions if the KLVY compiler had been known at the time of their genesis, because many tasks that appear difficult in the cryptographic single-prover setting are well-studied already in the nonlocal setting (and vice versa).

\subsection{Technical overview}

We focus here on how we obtain question-succinct quantum verification in the single-prover cryptographic setting, since the Killian-based answer compression protocol and its analysis were already presented in~\cite[Section 9]{BKLMMVVY22}, and we include an exposition-oriented sketch of these results in \Cref{sec:killian} only for completeness.

\subsubsection{The basic template from~\cite{NZ23}}

Like~\cite{NZ23}, our starting point is a basic framework for $\mathsf{QMA}$ verification in the two-prover setting due to  Grilo~\cite{Grilo17}. The verifier and the two provers (who we will call Alice and Bob) receive as input an instance of the $\mathsf{QMA}$-complete promise problem \emph{2-local XZ Hamiltonian}~\cite{biamonte2008realizable}. In other words, the problem that the verifier is trying to decide is whether a certain Hamiltonian $H$ on $n$ qubits, expressed as a sum of polynomially many 2-local X/Z Pauli terms (where each term is a tensor product of $n$ operators, each of which is chosen from $\{\1, \sigma_X, \sigma_Z\}$, such that all but 2 factors in the tensor product are $\1$), has lowest eigenvalue $\leq \alpha$ or $\geq \beta$ for two real numbers $(\alpha, \beta)$, where we are promised that $\beta - \alpha \geq \frac{1}{\poly(n)}$.

Honest Alice and Bob start out by sharing $n$ EPR pairs. The two-prover protocol underlying~\cite{NZ23} for deciding whether $H$ has lowest eigenvalue $\leq \alpha$ or $\geq \beta$ consists of two subtests, the \emph{Pauli braiding test} and the \emph{energy test}:
\begin{protocol}[informal]
\label{prot:nz-informal}
\end{protocol}
\begin{enumerate}
\item \textbf{Pauli braiding.} Alice and Bob execute a version of the Pauli braiding protocol from~\cite{natarajan2017quantum}, in which they play interleaved copies of CHSH (or another similar game, like Magic Square) and a simple game known as the `commutation test'. This protocol is a \emph{robust self-test} for the $n$-qubit Pauli group\footnote{More precisely the Heisenberg-Weyl group, the group consisting of tensor products of $\1, \sigma_X, \sigma_Z$ with $\pm 1$ signs, but we ignore this distinction in this introduction.}, in the sense that entangled players who win with high probability in this game must both be playing with measurement operators that are close (up to local isometries) to actual Pauli measurements. 
In other words, the Pauli braiding test allows the verifier to `force' entangled provers to perform Pauli measurements when requested to do so, even without trusting the provers. 
The most modular analysis of this protocol proceeds via a theorem from approximate representation theory that was first proven by Gowers and Hatami~\cite{GH15}.
\item \textbf{Energy testing via teleportation.} Alice is asked to \emph{teleport} the $n$-qubit witness state to Bob using their $n$ shared EPR pairs. She then reports the teleportation corrections to the verifier. Bob is asked to measure certain Pauli operators and report the outcomes. The verifier corrects Bob's reported outcomes using Alice's reported teleportation corrections, and interprets the result as a measurement of a term from $H$. It accepts or rejects depending on whether this measurement indicates that the state which Alice was meant to teleport to Bob is low-energy or high-energy.
\end{enumerate}
The intuition for the soundness of \Cref{prot:nz-informal} is as follows: the Pauli braiding test guarantees in some sense, through the use of the Gowers-Hatami theorem~\cite{GH15}, that all successful Bobs are in fact equivalent to honest Bob; and the energy test is straightforward to analyse if Bob is honest. In order to translate the intuition into reality, we have to make sure that Bob uses the \emph{same} strategy in both subtests so that the guarantee on Bob in the Pauli braiding test also applies to Bob in the energy test. That is, we must make sure he cannot play honestly only in the Pauli braiding test and then deviate however he likes in the energy test.

Suppose for the moment that the two subtests can be made \emph{perfectly indistinguishable} to Bob: that is, suppose that Bob's questions in both subtests are drawn from the same distribution. This would ensure that he does the same measurements in both subtests, since he does not know which subtest is being performed. The Pauli braiding subtest then guarantees that these measurements are `close' to honest measurements, and the soundness of subtest (ii) follows directly from the soundness of subtest (ii) with an honest Bob.

In~\cite{NZ23}, following a template laid out by Vidick in~\cite{Vid20-course}, the two subtests were indistinguishable because Bob's questions are very simple: in both subtests, Bob only ever receives one of two questions, each with $\frac{1}{2}$ probability. One of these two questions is an instruction to measure all of his qubits in the $Z$ basis (and report all $n$ outcomes), and the other is an instruction to measure all his qubits in the $X$ basis. Slightly more formally, honest Bob will in one case apply the projective measurement $\{ \ket{z} \bra{z} : z \in \{0,1\}^n \}$, and in the other case he will apply the projective measurement $\{ H^{\otimes n} \ket{x} \bra{x} H^{\otimes n} : x \in \{0,1\}^n \}$. 

Measurements of this form, as it turns out, are particularly `compatible' with the Gowers-Hatami-based analysis of the Pauli braiding test, in a sense that we will make somewhat more precise later (when we explain our `mixed-vs-pure basis test' later in this overview). It would therefore be convenient if this question structure was also sufficient for the energy test. Fortunately, this happens to be the case in the non-succinct setting: it turns out that 2-local X/Z Hamiltonian with inverse polynomial gap is complete for $\mathsf{QMA}$ even if we restrict the 2-local terms to $XX$ and $ZZ$ terms, i.e., terms where the two non-identity components of the $n$-fold tensor product are always of the same type ($\sigma_X$ or $\sigma_Z$). Note that the verifier can reconstruct a measurement of any $XX$-type term from the outcomes of an all-$X$ measurement performed by Bob, and any $ZZ$-type term from the outcomes of an all-$Z$ measurement performed by Bob. This means that in~\cite{NZ23}, it was sufficient in both subtests to ask Bob the same two questions (all-$X$ and all-$Z$), each with $\frac{1}{2}$ probability. Perfect indistinguishability of the two subtests in \Cref{prot:nz-informal} then follows.

\subsubsection{Obtaining succinctness}

In designing a (question-)succinct protocol with two entangled provers (which we will later compile into a crypographically secure single-prover protocol), we are faced with two new challenges compared with~\cite{NZ23}:
\begin{enumerate}
\item The Pauli braiding test (subtest (i) of \Cref{prot:nz-informal}) does not have succinct questions. In particular, while \emph{Bob's} questions can easily be made succinct (as we just described, it suffices to have only two Bob questions), \emph{Alice's} questions are more complicated.
\item In the non-succinct setting, $\frac{1}{\poly(n)}$ completeness-soundness gap is generally tolerated because it is assumed that $\poly(n)$ many rounds of sequential repetition can be performed in order to boost the gap. In the succinct setting, this is not feasible, since repeating a succinct protocol $\poly(n)$ times results in $\poly\log(n) \cdot \poly(n)$ communication; therefore, in the succinct setting, we must design a protocol which has \emph{constant} soundness gap even without any repetition. This means that we cannot start with a 2-local XX/ZZ Hamiltonian, since it is not known whether this problem is $\mathsf{QMA}$-complete with a constant promise gap. If we are to take the same approach of starting from some Hamiltonian problem, then it has to be a Hamiltonian problem with \emph{constant} promise gap such that the terms can be grouped into a small number of subsets (at most $2^{\poly \log n}$ subsets), each of which contains only terms that commute. If this is the case, then Bob can measure all the terms in a single subset simultaneously and report all the outcomes together, and the verifier only needs to use $\poly\log n$ bits to tell Bob which subset to measure. If this is not the case, then the energy testing template from subtest (ii) of \Cref{prot:nz-informal} will not work, since the verifier will not be able to tell Bob which terms he should measure in a succinct way.
\end{enumerate}

\paragraph{Subsampling Hamiltonians.}
We take a similar approach to Bartusek et al.~\cite{BKLMMVVY22} in order to deal with the second issue. We use na\"ive $\mathsf{QMA}$ parallel amplification (first written down in~\cite{KSV02}; the procedure simply repeats the $\mathsf{QMA}$ verifier in parallel a polynomial number of times) in order to boost the promise gap to a constant; this results in a Hamiltonian that is a sum of exponentially many terms, each of which can be efficiently measured by measuring each of the $n$ qubits of the witness in either the $X$ or the $Z$ Pauli basis (with potentially different basis choices for different qubits). We then use a generic \emph{PRG with soundness against adversaries with quantum advice} in order to `subsample' these terms and emerge with a Hamiltonian that is a sum of $2^{\poly \log n}$ terms, each of which can be efficiently measured by measuring each of the $n$ qubits of the witness in either the $X$ or the $Z$ Pauli basis. This part of the work is done in \Cref{sec:hamiltonian-subsamp}. 

\paragraph{The mixed-vs-pure basis test and a new self-testing-oriented proof of Gowers-Hatami.}
At this point we have created a new problem: the terms of the Hamiltonian we want to use in the energy subtest can no longer be measured by a Bob who only ever measures every qubit of his state in either the $X$ basis or the $Z$ basis. This is because the amplified Hamiltonian contains tensor products of arbitrary combinations of $XX$ and $ZZ$ terms from the original Hamiltonian, and not only tensor products of terms in the same basis. 
These mixed terms \emph{can} be measured by a Bob who does what we call \emph{mixed basis} measurements (measurements that involve measuring each of $n$ qubits in either the $X$ or the $Z$ basis, with potentially different basis choices for different qubits). However, if the verifier picks the mixed bases depending on the distribution induced by the constant-gap Hamiltonian, the resulting distribution over Bob questions is not necessarily `compatible' with even the regular Pauli braiding test. Moreover, it becomes even more difficult to use anything other than the all-$X$ and all-$Z$ measurements when we consider the \emph{succinct} version of Pauli braiding, for reasons that we will elaborate on shortly (in the section `Succinct Pauli braiding').

The natural solution is to use the all-$X$ and all-$Z$ measurements when we play Pauli braiding, use the mixed basis measurements when we do the energy test, and introduce some sort of \emph{consistency test} to ensure that the operators that Bob uses in the energy test are in some sense the same ones as the ones he uses in Pauli braiding. (We call this test the mixed-vs-pure basis test; the protocol is described in \Cref{prot:mixed-vs-pure}.) Such tests have been analysed in the nonlocal setting before~\cite{NW19}, but we are the first to attempt to analyse such a test in the compiled setting, and the compilation introduces unforeseen difficulties (see `Difficulties in the analysis of the mixed-vs-pure basis test' below).

The easiest solution to the difficulties that we were able to come up with involves reproving the Gowers-Hatami theorem (or, rather, the parts of the theorem relevant for self-testing) in a way that supports arbitrary non-uniform expectations. The (informal) theorem statement for our version of Gowers-Hatami is as follows:
\begin{theorem}[informal]
\label{thm:nonuniform-gh-informal}
    Let $f: G \to U(\cH)$ be a function from a finite group $G$ to the set of unitaries on some Hilbert space $\cH$. Then there exists a finite-dimensional Hilbert space $\cH'$, an isometry $V: \cH \to \cH'$, and a unitary representation $\pi: G \to U(\cH')$ of $G$ such that for all measures $\mu$ over $G$,
    \[ \E_{g \sim \mu, h \sim W_n} \| f(h) f(g) - f(hg) \|^2 \leq \eps \implies \E_{g \sim \mu} \| f(g) - V^\dagger \pi(g) V \|^2 \leq \eps, \]
    where we are being purposefully vague about the norm.
\end{theorem}
The difference between this theorem and the more typical formulation is that the typical formulation has uniform expectations over the group everywhere. A version of Gowers-Hatami similar to \Cref{thm:nonuniform-gh-informal} is often needed in the self-testing setting when $\mu$ is in particular the uniform distribution over $\{ \sigma_Z(a) : a \in \{0,1\}^n \}$ or $\{ \sigma_X(b) : b \in \{0,1\}^n \}$, and it is plausible that \Cref{thm:nonuniform-gh-informal} could also be proven by modifying in some way Gowers and Hatami's original proof of their theorem. 
Nonetheless, the proof that we present under \Cref{lem:non-uniform-gh} (the formal version of \Cref{thm:nonuniform-gh-informal}) is an entirely different proof that only uses basic tools from quantum information, namely Stinespring dilation (instead of matrix Fourier analysis on non-Abelian groups~\cite{GH15}). We emphasise that our proof is \emph{not} a reproof of the full Gowers-Hatami theorem, because the original theorem gets bounds on the dimension of the `post-rounding' Hilbert space $\cH'$ (which one typically does not need in self-testing-related applications of Gowers-Hatami). However, our proof has the advantage that it is completely elementary and self-contained. We believe this proof may be of independent interest, because the fact that it is simple and self-contained makes it easier to modify the statement when necessary to incorporate additional desirable properties (such as, for example, the tolerance for non-uniform expectations that we needed for this work). Together with a `distribution-switching' trick presented in \Cref{lem:pauli_rounding_all_dist}, we are able to use this version of Gowers-Hatami to work out an analysis of the mixed-vs-pure basis test. We give more details about how we did this at the end of the following section.

\paragraph{Difficulties in the analysis of the mixed-vs-pure basis test.}
Now we elaborate more thoroughly on the nature of the difficulties that we encountered in analysing the mixed-vs-pure basis test, which we finally resolved by using \Cref{lem:non-uniform-gh} and \Cref{lem:pauli_rounding_all_dist}. We firstly justify the sense in which the all-$X$ and all-$Z$ measurements are particularly `compatible' with Pauli braiding, in order to clarify why the consistency test is necessary in the first place. \\

\noindent \emph{Why the mixed-vs-pure basis test is necessary. \:} The reason why the all-$X$ and all-$Z$ measurements are particularly suitable for use in the Pauli braiding protocol is that the all-$Z$ question can be interpreted as a simultaneous measurement of the $2^n$ binary observables $\{ \sigma_Z(a) : a \in \{0,1\}^n \}$, where $\sigma_Z(a)$ is the binary observable that is the tensor product of $\sigma_Z$ on all the qubits $i$ where $a_i = 1$ and identity otherwise; and, similarly, the all-$X$ question can be interpreted as a simultaneous measurement of the $2^n$ binary observables $\{ \sigma_X(b) : b \in \{0,1\}^n \}$. Another (more precise) way to say this is that, given the (potentially cheating) projective measurement $\{P^Z_{u} : u \in \{0,1\}^n\}$ that Bob applies when he receives the instruction to measure everything in the $Z$ basis, we can construct a set of $2^n$ binary observables $\{Z(a) : a \in \{0,1\}^n\}$ which are \emph{exactly linear}, in the sense that
\begin{align}
Z(a) Z(a') = Z(a + a'), \label{eq:linZ}
\end{align}
even if Bob is dishonest: simply take
\begin{align*}
Z(a) \coloneqq \sum_{u \in \{0,1\}^n} (-1)^{u \cdot a} P^Z_u.
\end{align*}
A similar statement holds true for the all-$X$ measurement: we can define a set of $2^n$ binary observables $\{X(b) : b \in \{0,1\}^n\}$ such that
\begin{align}
X(b) X(b') = X(b + b'). \label{eq:linX}
\end{align}
We can use the CHSH game and the commutation test in order to certify that these $2 \cdot 2^n$ binary observables $\{Z(a), X(b) : a, b \in \{0,1\}^n\}$ satisfy the commutation relations that would hold if they were genuine Paulis, i.e.
\begin{align}
\| Z(a) X(b) - (-1)^{a \cdot b} X(b) Z(a) \|_2 \leq O(\eps). \label{eq:com}
\end{align}
Taking the linearity (\Cref{eq:linZ} and \Cref{eq:linX}) and commutation (\Cref{eq:com}) relations together, we can prove that $\{Z(a) : a \in \{0,1\}^n\}$ and $\{X(b) : b \in \{0,1\}^n\}$ approximately satisfy the relations satisfied by the corresponding elements of the Pauli group. Moreover, by taking products, we can extend $Z(a)$ and $X(b)$ to a matrix-valued function $f(s, a,b) = (-1)^s Z(a)X(b)$ that approximately obeys the multiplication law of the Pauli group. The Gowers-Hatami theorem then implies that there is a \emph{rounding} of $f$ which \emph{exactly} satisfies the Pauli group relations (up to isometry). That is, there exists a representation $\rho$ of the Pauli group such that, on average over $s,a,b$, $f(s,a,b)$ is close to $\rho(s,a,b)$ conjugated by the isometry. 

Zooming back out to the level of designing Bob's questions, note that the all-$Z$ and all-$X$ questions were particularly nice for the Pauli braiding test because (1) the sets $\{ \sigma_Z(a) : a \in \{0,1\}^n \}$ and $\{ \sigma_X(b) : b \in \{0,1\}^n \}$ taken together generate the entire $n$-qubit Pauli group, and (2) the trick of constructing many binary observables from a single projective measurement gave us \emph{exact linearity} on the $Z$ side and the $X$ side individually almost for free: that is, $\{ Z(a) : a \in \{0,1\}^n \}$ is automatically an exact representation of $\mathbb{Z}_2^n$, and the same is true of $\{ X(b) : b \in \{0,1\}^n \}$.

There is no guarantee that these nice properties hold if we consider (instead of the all-$X$ and all-$Z$ questions) the set of mixed-basis questions induced by the energy test for our constant-gap Hamiltonian. In particular, there is no guarantee that the binary observables which can be constructed from Bob's set of mixed basis measurements will generate the whole Pauli group, in the way that $\{ \sigma_Z(a) : a \in \{0,1\}^n \}$ and $\{ \sigma_X(b) : b \in \{0,1\}^n \}$ generate the whole Pauli group. It becomes even more important to use the all-$X$ and all-$Z$ questions if we want to eventually make the Pauli braiding test question-succinct: we give some intuition as to why this is the case in the section `Succinct Pauli braiding'.

The easiest solution seems to be to introduce a consistency test between Bob's mixed basis measurements (that we would like Bob to use when he plays the energy test) and Bob's pure basis measurements (that we would like Bob to use when he plays the Pauli braiding test). More specifically---following the standard template for designing tests of this form---we will introduce two new questions into Alice's question set that are identical to Bob's pure basis questions (i.e.~`measure all in $X$' and `measure all in $Z$'); we will ask Bob to play his pure basis operators against Alice's pure basis operators, in order to check that Bob's pure basis operators are consistent with Alice's pure basis operators; and then we will ask Bob to play his mixed basis operators against Alice's pure basis operators, and check that they agree whenever the bases align, which (since we checked that Bob's and Alice's pure basis operators agree) is essentially equivalent to checking that Bob's mixed basis operators are consistent with Bob's pure basis operators. We might hope that this test, combined with the usual analysis, will be sufficient to allow us to `round' Bob's mixed operators in the same way that we can round Bob's all-$X$ and all-$Z$ measurements by using the usual Gowers-Hatami analysis. \\

\noindent \emph{Difficulties in the analysis. \:} Unfortunately, instantiating this intuition proves to be nontrivial in the compiled setting, even though the analysis is fairly routine in the nonlocal setting. The tensor product structure of the provers' Hilbert space in the nonlocal setting is useful because it supports a large range of convenient operations that are loosely grouped together under the name of `prover-switching'. The ordinary nonlocal analysis of a consistency test like this one would proceed primarily through prover-switching calculations. While we did find it necessary to prove some lemmas which capture certain applications of prover-switching in the compiled setting (see \Cref{sec:prover-switching}), we found that these lemmas were insufficient in order to analyse the mixed-vs-pure basis test.

More specifically, the main difficulty we encountered was the following. The statement we would like to show, in order to make the energy test work in the presence of mixed terms, is of the following form. Let $w \in \{\1,X,Z\}^n$ be a string indicating which Pauli bases to measure $n$ qubits in. We want to show that, if Alice and Bob win in our protocol with high probability, then there exists an isometry $V$ such that, for the distribution $D$ on Pauli basis choices induced by the constant-gap Hamiltonian,
\begin{align*}
\E_{w \sim D} \E_{a \in \bits^n} \| O^w(a) - V^\dagger (\sigma_w(a) \otimes \1_{\rm aux}) V \|^2 \leq \mathrm{small},
\end{align*}
where $\sigma_w(a)$ is the honest Pauli observable that corresponds to the tensor product
\begin{align*}
\sigma_w(a) = \bigotimes_i \sigma_{w_i}^{a_i} \,.
\end{align*}
and $O^w(a)$ is Bob's potentially cheating version of $\sigma_w(a)$.

Normal pure-basis Gowers-Hatami tells us that, if Alice and Bob win with high probability in Pauli braiding, then for any $W \in \{X,Z\}$ it is the case that
\begin{align}
\label{eq:pure-gh}
\E_{a \in \bits^n} \| W(a) - V^\dagger (\sigma_W(a) \otimes \1_{\rm aux}) V \|^2 \leq \mathrm{small},
\end{align}
for some fixed isometry $V$. One idea for proceeding with the analysis might be to show that $O^w(a) \approx Z(c) X(d)$ using the mixed-vs-pure basis test (with $c$ being the string such that $c_i = 1$ iff $a_i = 1$ and $w_i = W$, and similarly for $d$), and then to `round' $Z(c)$ and $X(d)$ separately using \Cref{eq:pure-gh}. Unfortunately, rounding something of the form $Z(c) X(d)$ na\"ively using \Cref{eq:pure-gh} produces something of the form
\begin{align*}
V^\dagger (\sigma_Z(c) \otimes \1_{\rm aux}) V V^\dagger (\sigma_X(d) \otimes \1_{\rm aux}) V.
\end{align*}
Since $V$ is an isometry and not a unitary, $V V^\dagger$ is not necessarily $\1$, and it is unclear how to get rid of it: we call this the `$V V^\dagger$ problem'. There are ways to bypass this problem in the nonlocal setting using tensor product structure, but we were not able to replicate these techniques in the compiled setting.

Instead, we bypass the problem by `directly' proving a form of Gowers-Hatami that, perhaps surprisingly, allows us to round in expectation over \emph{any} distribution over the Pauli group, even though the Pauli braiding test is only played with the uniform distribution. 
More specifically, we prove our version of Gowers-Hatami (\Cref{lem:non-uniform-gh}, stated earlier informally as \Cref{thm:nonuniform-gh-informal}), which can be used to round \emph{arbitrary} distributions over the underlying group, provided with the right hypothesis; and then we prove, using a `distribution-switching' trick (\Cref{lem:pauli_rounding_all_dist}), that the hypothesis of \Cref{lem:non-uniform-gh} can be obtained for \emph{any} distribution $\mu$ even if we only start with commutation relations that hold on uniform average over pure-basis elements (and a few other conditions, such as exact linearity), which is what we have access to through the pure-basis Pauli braiding test.

\paragraph{Succinct Pauli braiding.}
Finally, armed with the mixed-vs-pure basis test, we can focus on making the Pauli braiding test succinct (where, by `Pauli braiding test', we mean the version in which Bob always gets asked either the all-$X$ or the all-$Z$ question). Our starting point for this mission is de la Salle's elegant simplification~\cite{delaSalle2022spectral} of `question reduction' from~\cite{JNVWY20}, in which he introduces a version of Pauli braiding where Alice's questions are sampled from $\eps$-biased sets. The normal Pauli braiding game proceeds as follows:
\begin{itemize}
\item The verifier chooses two strings $a, b \in \{0,1\}^n$ uniformly at random.
\item The verifier decides what to do next based on the parity of $a \cdot b$:
\begin{itemize}
\item If $a \cdot b = 0$, the verifier referees a \emph{commutation game} (in which honest Alice plays with $\sigma_Z(a)$ and $\sigma_X(b)$).
\item If $a \cdot b = 1$, the verifier referees an \emph{anticommutation game} (in which, again, honest Alice plays by embedding $\sigma_Z(a)$ and $\sigma_X(b)$ into her strategy).
\end{itemize}
\end{itemize}
The commutation game is designed to test that two particular operators commute, and the anticommutation game (based on CHSH or Magic Square) is designed to test that two particular operators anticommute.

Note that the verifier has to send $a,b$ to Alice for this protocol to work. The protocol was made succinct by de la Salle simply by choosing $a, b$ from $\eps$-biased sets instead of from all of $\{0,1\}^n$. This is a natural idea, but it is at first surprising that it works at all: after all, the sets of Paulis $\{ \sigma_Z(a) : a \in S \}$ and $\{ \sigma_X(b) : b \in S \}$ for some $\eps$-biased $S$, where $|S| = \poly(n)$, only cover an exponentially small fraction of the Pauli group! All that the protocol directly certifies is commutation and anticommutation relations among pairs of operators in these sets. Na\"{i}vely, to deduce relations about representations of \emph{general} group elements, one would need to write these elements as $\poly(n)$-length words in the group elements from the $\eps$-biased sets, and apply the relations on the $\eps$-biased sets $\poly(n)$-many times. This would seemingly rule out a test with constant soundness.

Miraculously, however, everything still works as before, and the reason is that we do probe the entire group through Bob, who still measures the all-$X$ and all-$Z$ mesaurements. In particular, we have `for free' (or by construction) that Bob's $X(b)$ operators, taken as a set, form an \emph{exact} representation of $\mathbb{Z}_2^n$, and the same for his $Z(a)$ operators. Meanwhile, all elements of the Pauli group can be written as words of constant length in the operators $\{ \sigma_Z(a) : a \in \{0,1\}^n \}$ and $\{ \sigma_X(b) : b \in \{0,1\}^n \}$. In some sense, de la Salle's test works because probing the commutation relations between two exact representations of $\mathbb{Z}_2^n$ on only an $\eps$-biased set is sufficient to establish the commutation relations everywhere, because the function of $\eps$-biased sets is precisely to `fool' exactly linear functions. In fact,  de la Salle's test and its analysis are analogous to the ``derandomized BLR test'' for linear functions and the Fourier-based analysis of it presented in Section~6.4 of~\cite{odonnellbook}.

In order to use the succinct version of Pauli braiding in our protocol, we have to come up with a version of the analysis that works in the compiled setting. Unfortunately, de la Salle's original proof in the nonlocal case is written in the `synchronous' setting, in which the provers (even malicious provers) are assumed to start out by sharing EPR pairs. This assumption simplifies the calculations because it allows us to move (`prover switch') measurements freely from one prover to the other. The synchronicity assumption is without loss of generality in the nonlocal setting by~\cite{vidick2022almost}, but no compiled version of this result exists. Therefore, we have to redo the proof in our setting using the state-dependent distance, and come up with ways to use the cryptography to simulate the ``prover switching" steps in de la Salle's analysis. (At the time of~\cite{NZ23} it was not known whether the cryptography could in fact simulate these properties.) In the process, we pare down de la Salle's proof to the parts that are essential for analysing succinct Pauli braiding and state it in more computer-science-like language, which may be useful for future readers with a computer science background. Our version of de la Salle's analysis is presented as \Cref{lem:dls-commutator-close}.

\paragraph{Related work.}
Simultaneously, a succinct argument system for $\QMA$ based only on the post-quantum security of LWE (a standard assumption) was achieved by \cite{GKNV24}. Both of these works use tools from~\cite{BKLMMVVY22}, in particular the technique of ``subsampling'' a Hamiltonian using a PRG, and the technique of transforming a semi-succinct protocol into a fully succinct one by using succinct arguments of knowledge. However, the methods they use to solve the key technical challenge of succinctly delegating many-qubit Pauli measurements are essentially disjoint. In particular, for us, the ``heavy lifting" to achieve question-succinctness is performed \emph{information theoretically}, in our question-succinct two-prover self-test for EPR pairs, whereas for them, succinctness is achieved by using specific technical features of a cryptographic construction using LWE.

\paragraph{Acknowledgements.}
We are grateful to Fermi Ma for allowing us to rewind him until we could extract an understanding from the succinct arguments in Section 9 of~\cite{BKLMMVVY22}. We also thank Yael Kalai, Alex Lombardi, and Thomas Vidick for useful discussions.
TM acknowledges support from the ETH Zurich Quantum Center and an ETH Doc.Mobility Fellowship.

\section{Preliminaries}

\subsection{Notation}
We assume basic familiarity with quantum states and measurements.
We consider finite dimensional Hilbert spaces, which we commonly denote by $\cH$.
The set of linear operator $A: \cH \to \cH$ is $L(\cH)$, the set of positive semidefinite $A: \cH \to \cH$ is $\pos(\cH)$, and the set of unitaries $U: \cH \to \cH$ is $U(\cH)$.
For $A \in L(\cH)$, we use $\norm{A}_p \deq \Tr[(A^\dagger A)^{p/2}]^{1/p}$ to denote its Schatten $p$-norm. 
For any (subnormalised) pure state $\ket{\psi}$, we denote the (subnormalised) density matrix by $\psi$.

For a probability distribution $D$ over some set $\cX$, we write $x \sim D$ to denote a sample drawn according to $D$.
We write $x \sim \cX$ to denote a sample drawn uniformly from $\cX$.

For $a, b \in \C$ and $\delta > 0$, we write $a \approx_\delta b$ if $|a - b| \leq \delta$.

\subsection{Nonlocal games}

In a non-local game, a verifier (or referee) sends questions to multiple non-communicating provers, receives an answer from each prover, and evaluates the question-answer combination using a verification predicate to decide whether the provers ``win'' or ``lose'' the game.
Formally, a two-prover non-local game is specified by the following.

\begin{definition}
    A nonlocal game $G$ is given by natural numbers $n_1, n_2, m_1, m_2$, a distribution $\cQ$ over pairs $(x,y) \in \{0,1\}^{n_1} \times \{0,1\}^{n_2}$, and a polynomial-time verification predicate $V( x, y, a, b) \in \{0,1\}$, where $a \in \{0,1\}^{m_1}$ and $b \in \{0,1\}^{m_2}$.
\end{definition}

\subsection{Pauli matrices and Heisenberg-Weyl group}
We use the usual Pauli matrices $\sigma_X, \sigma_Y, \sigma_Z$. We will also find it convenient to set $\sigma_{\1} = \1$.
For $w \in \{\1, X,Z\}^n$ and $a \in \bits^n$, we define
\begin{align*}
\sigma_w(a) = \bigotimes_i \sigma_{w_i}^{a_i} \,.
\end{align*}
We also write $\sigma_w = \sigma_w(\vec 1) = \otimes_i \sigma_{w_i}$.

In addition, we define the Pauli \emph{projections} $\pi^w_a$ as
\begin{align}
\label{eq:pauli-proj}
\pi^w_u = \bigotimes_i \left( \frac{\1 + (-1)^{u_i} \sigma_{w_i}}{2} \right) = \E_{a \in \bits^n} (-1)^{u \cdot a} \sigma_w(a)\,.
\end{align}

The Heisenberg-Weyl group on $n$ qubits, denoted $W_n$, is the group generated by $\{\sigma_X(a), \sigma_Z(b)\}_{a, b \in \bits^n}$, where the group operation is induced by matrix multiplication.
More formally, to differentiate between the group in the abstract and its concrete representation in terms of Pauli matrices, we view $W_n$ as a group with generators $X_1, \dots, X_n, Z_1, \dots, Z_n$, with the multiplication rule given by associating $X_i$ with $\sigma_X(e_i)$ and  $Z_j$ with $\sigma_X(e_j)$ (where $e_i \in \bits^n$ is 0 everywhere except on position $i$, where it is 1). 
We record the representations of $W_n$ in the following lemma.
\begin{lemma}\label{lem:hw-representations}
    The irreducible representations of the Heisenberg-Weyl group $W_n$ consist of
    \begin{itemize}
        \item The one-dimensional representations, given by choosing an assignment of 1 or $-1$ for each of the generators
        \begin{align*}
            X_1, Z_1, \dots, X_n, Z_n
        \end{align*}
        and filling out the rest of the multiplication table accordingly.
        Note that for any one-dimension irrep $\pi$, $\pi(-\id) = \pi(X_1 Z_1 X_1 Z_1) = 1$, where the first equality uses that $X_1 Z_1 X_1 Z_1 = -\id$ according to the group relations, and the second equality uses the representation property and the fact that multiplication of scalars is commutative.
        \item The `fundamental representation' given by the $2^n$-dimensional Pauli matrices: that is,
        \begin{align*}
            \forall j, \quad &X_j \mapsto \sigma_X(e_j), \\
            &Z_j \mapsto \sigma_Z(e_j).
        \end{align*}
    \end{itemize}
\end{lemma}

\subsection{Reduced measurements}

We will frequently consider measurements that return tuples as outcomes.
For those measurements, it is convenient to define reduced (or marginalised) measurements, where we only care about some of the elements of the outcome tuple and ignore the others.
For this, we define the following notation.

\begin{definition}
    For any alphabet $\Sigma$, let $u \in \Sigma^n$, and let $S = \{i_1, i_2, \dots, i_k\} \subseteq [n]$. Then $u|_{S} \deq (u_{i_1}, \dots, u_{i_k}) \in \Sigma^{k}$ is the restriction of $u$ to $S$.
\end{definition}
\begin{definition}
    For any alphabet $\Sigma$, $w \in \Sigma^n$, and $W \in \Sigma$, the set $\{w = W\}$ is defined as $\{i \in [n]: w_i = W\}$.
\end{definition}
\begin{definition}[Reduced measurements] \label{def:reduced_meas}
    For any alphabet $\Sigma$, measurement $\{Q_u\}_{u \in \Sigma^n}$, and set $S \subseteq [n]$, we define the reduced measurement
    \[ (Q|_S)_{v \in \Sigma^{|S|}} = \sum_{u: u|_S = v} Q_u. \]
\end{definition}
We will often apply this definition to measurements that measure many qubits, each in one of the Pauli bases, to select the outcomes corresponding to a single basis. Specifically, for $\Sigma = \bits$, and for any Pauli string $w \in \{X,Z\}^n$, basis choice $W \in \{X,Z\}$, and string $v \in \bits^{|\{w = W\}|}$, we have
\begin{align*}
(Q\setr{w = W})_v = \sum_{u: u|_{\{w=W\}} = v} Q_u \,,
\end{align*}
In other words $(Q|_{\{w = W\}})_v$ is the marginalisation over all outcomes on indices where $w_i \neq W$.

\subsection{State dependent norm}
\begin{definition}[State-dependent inner product and norm]\label{def:state_dep_inner_product}
 Let $\cH$ be a finite-dimensional Hilbert space and $A, B \in L(\cH)$ be linear operators on $\cH$. Let $\psi \in \pos(\cH)$. We define the state-dependent (semi) inner product of $A$ and $B$ w.r.t $\psi$ as 
\begin{equation*}
\langle A, B \rangle_\psi = \Tr[ A^\dagger B \psi ] \,.
\end{equation*} 
This induces the state-dependent (semi) norm 
\begin{equation*}
\norm{A}_\psi^2 = \langle A, A \rangle_\psi = \Tr[ A^\dagger A \psi ] \,.
\end{equation*}
\end{definition}

\begin{remark} \label{rem:schatten_norm}
The state dependent (semi) norm can also be expressed as a Schatten 2-norm (also called the Hilbert-Schmidt norm): 
\begin{equation*}
\norm{A}_\psi = \norm{A \psi^{1/2}}_2 \,.
\end{equation*}
\end{remark}

We collect a number of basic properties of the state dependent norm.
These are standard properties and easy to prove from the definitions.
\begin{lemma}[Basic properties of the state dependent norm] \label{lem:state_dep_norm_props}
For all (not necessarily normalised) states $\psi, \psi' \in \pos(\cH)$ and linear operators $A, B \in L(\cH)$ on some finite-dimensional Hilbert space $\cH$:
\begin{enumerate}
\item $\norm{A}_{B \psi B^\dagger} = \norm{AB}_{\psi}$. \label{item:conj_mult}
\item $\norm{A B}_{\psi} \leq \norm{A}_{\infty} \norm{B}_{\psi}$. \label{item:submult_infty}
\item For any unitary $U$, $\norm{UA}_\psi = \norm{A}_\psi$. \label{item:left_unitary_inv}
\item Linearity of the squared norm in the state: $\norm{A}_{\psi + \psi'}^2 = \norm{A}_{\psi}^2 + \norm{A}_{\psi'}^2$. \label{item:linearity}
\item Triangle inequality for the squared norm: $\norm{A+B}_\psi^2 \leq 2 \norm{A}_\psi^2 + 2 \norm{B}_\psi^2$. \label{item:triangle_ineq_sq}
\end{enumerate}
\end{lemma}

\begin{lemma} \label{lem:state_replace}
For an observable $A$ on $\cH$ and two states $\psi, \psi$ on $\cH$ with $\norm{\psi - \psi'}_1 \leq \eps$, we have that 
\begin{align*}
\norm{A}_{\psi}^2 \approx_{\norm{A}_{\infty}^2 \eps} \norm{A}_{\psi'}^2 \,.
\end{align*}
\end{lemma}
\begin{proof}
This follows immediately from the definition of the state dependent norm and H\"older's inequality.
\end{proof}

\begin{lemma} \label{lem:replace_on_state}
For all $i \in \cI$ (for some index set $\cI$) let $A_i\in L(\cH)$ and $\psi_i \in \pos(\cH)$ such that $\sum_i \Tr[\psi_i] \leq 1$. Then 
\begin{align*}
\sum_i \norm{A_i \psi_i}_1 \leq \sqrt{\sum_i \norm{A_i}_{\psi_i}^2} \,.
\end{align*}
\end{lemma}
\begin{proof}
We first split $\psi_i =  \psi_i^{1/2} \psi_i^{1/2}$ and apply H\"older's inequality: 
\begin{align*}
\sum_i \norm{A_i \psi_i}_1 
&\leq \sum_i \norm{A_i \psi_i^{1/2}}_2 \norm{\psi_i^{1/2}}_2 \,. \\
\intertext{Now applying Cauchy Schwarz:}
&\leq \sqrt{\sum_i \norm{A_i \psi_i^{1/2}}_2^2} \cdot \sqrt{\sum_i \norm{\psi_i^{1/2}}_2^2} 
\leq \sqrt{\sum_i \norm{A_i}_{\psi_i}^2} \,,
\end{align*}
where in the last line we rewrote the first factor as a state-dependent norm and for the second factor] observed that $\sum_i \norm{\psi_i^{1/2}}_{2}^2 = \sum_i \Tr[\psi_i] \leq 1$.
\end{proof}

\subsection{Efficient observables and computational indistinguishability}

\begin{definition}[Computational indistinguishability] \label{def:comp_indist}

Two families of (subnormalized) states $\{\psi_1(\lambda)\}_\lambda$ and $\{\psi_2(\lambda)\}_\lambda$, indexed by a security parameter $\lambda$, are \emph{compuationally indistinguishable}
if for any family of computationally eficient two-outcome  POVMs $\{M(\lambda), 1 - M(\lambda)\}_\lambda$ indexed by $\lambda$, it holds that
\[ \Tr[ M(\lambda) (\psi_1(\lambda) - \psi_2(\lambda)] \leq \negl(\lambda). \]
We typically indicate this by
\[ \psi_1 \capprox \psi_2,\]
suppressing the dependence on $\lambda$ in the notation.
\end{definition}

\begin{lemma} \label{lem:computational_state_switching}
Let $\{U^1_a\}_{a \in \cA}, \dots, \{U^k_a\}_{a \in \cA}$ and $\{V^1_b\}_{b \in \cB}, \dots, \{V^\ell_b\}_{b \in \cB}$ be families of efficient unitaries on some Hilbert space $\cH$, $s: \cA^k \times \cB^\ell \to \bits$ an efficiently computable function, and $\mu$ an efficiently sampleable distribution over $\cA^{\times k} \times \cB^{\times \ell}$.

Consider two states $\psi \capprox_\delta \psi'$.
Then 
\begin{multline*}
\E_{(a_1, \dots, a_k, b_1, \dots, b_\ell) \sim \mu} \norm{U^1_{a_1} \cdots U^k_{a_k} - (-1)^{s(a_1, \dots, a_k, b_1, \dots, b_\ell)} V^1_{b_1} \cdots V^\ell_{b_\ell}}_\psi^2 \\
\approx_\delta 
\E_{(a_1, \dots, a_k, b_1, \dots, b_\ell) \sim \mu} \norm{U^1_{a_1} \cdots U^k_{a_k} - (-1)^{s(a_1, \dots, a_k, b_1, \dots, b_\ell)} V^1_{b_1} \cdots V^\ell_{b_\ell}}_{\psi'}^2 \,.
\end{multline*}
\end{lemma}

\begin{proof}
Note that the product of efficient unitaries is another efficient unitary.
Then this follows directly from the definition of computational indistinguishability (\cref{def:comp_indist}), linearity of expectation, and~\cite[Lemma 2.6]{metger2021self}. 
\end{proof}

\subsection{Quantum homomorphic encryption}

The following definitions are taken with some modifications from \cite{KLVY21}. Note that this definition of quantum homomorphic encryption requires, among other specialised requirements, that the scheme specialises to a classical encryption scheme when applied to classical messages: we need this property in our constructions.
\begin{definition}[Quantum Homomorphic Encryption (QHE)]\label{def:QHE-aux}
A quantum homomorphic encryption scheme $\QHE=(\Gen,\Enc,\Eval,\Dec)$ for a class of quantum circuits $\cC$ is a tuple of algorithms with the following syntax:
\begin{itemize}
    \item $\Gen$ is a $\PPT$ algorithm that takes as input the security parameter $1^\secp$ and outputs a (classical) secret key $\sk$ of $\poly(\secp)$ bits;
    \item $\Enc$ is a $\PPT$ algorithm that takes as input a secret key $\sk$ and a classical input $x$, and outputs a classical ciphertext $\ct$;
    \item $\Eval$ is a $\QPT$ algorithm that takes as input a tuple $(C,\ket{\Psi},\ct_{\mathrm{in}})$, where $C:\cH\times(\mathbb{C}^2)^{\otimes n}\rightarrow (\mathbb{C}^2)^{\otimes m}$ is a quantum circuit, $\ket{\Psi}\in\cH$ is a quantum state, and $\ct_{\mathrm{in}}$ is a ciphertext corresponding to an $n$-bit classical plaintext. 
    $\Eval$ computes a quantum circuit
    $\Eval_C(\ket{\Psi}\otimes \ket{0}^{\otimes \poly(\secp, n)},\ct_{\mathrm{in}})$ which outputs a ciphertext $\ct_{\mathrm{out}}$. If $C$ has classical output, we require that $\Eval_C$ also has classical output.
    \item $\Dec$ is a $\sf{QPT}$ algorithm that takes as input a secret key $\sk$ and ciphertext $\ct$, and outputs a state $\ket{\phi}$. Additionally, if $\ct$ is a classical ciphertext, the decryption algorithm outputs a classical string $y$.
\end{itemize}

We require the following three properties from $(\Gen,\Enc,\Eval,\Dec)$:
\begin{itemize}
    \item \textbf{Correctness with auxiliary input:} For every security parameter $\secp\in\Nat$, any quantum circuit $C:\cH_{\mathsf{A}} \times (\mathbb{C}^2)^{\otimes n} \to \{0,1\}^*$ (with classical output), any quantum state $\ket{\Psi}_{\mathsf{A} \mathsf{B}} \in\cH_{\mathsf{A}} \otimes \cH_{\mathsf{B}}$, any message $x\in \{0,1\}^n$, any secret key $\sk \gets \Gen(1^\secp)$ and any ciphertext $\ct \gets \Enc(\sk,x)$, the following states have negligible trace distance:
    \begin{description}
        \item \textit{Game $1$.} Start with $(x, \ket{\Psi}_{\mathsf{A} \mathsf{B}})$. Evaluate $C$ on $x$ and register $\mathsf{A}$, obtaining classical string $y$. Output $y$ and the contents of register $\mathsf{B}$.
        \item \textit{Game $2$.} Start with $\ct \gets \Enc(\sk, x)$ and $\ket{\Psi}_{\mathsf{A} \mathsf{B}}$. Compute  $\ct' \gets\Eval_C(\cdot \otimes \ket{0}^{\poly(\secp, n)} ,\ct)$ on register $\mathsf{A}$. Compute $y'= \Dec(\sk,\ct')$. Output $y'$ and the contents of register $\mathsf{B}$.
    \end{description}
    
In words, ``correctness with auxiliary input'' requires that if QHE evaluation is applied to a register $\mathsf{A}$ that is a part of a joint (entangled) state in $\cH_{\mathsf{A}}\otimes \cH_{\mathsf{B}}$, the entanglement between the QHE evaluated output and $\mathsf{B}$ is preserved.
    
    \item \textbf{IND-CPA security against quantum distinguishers:} For any two messages $x_0, x_1$ and any $\sf{QPT}$ adversary ~$\cA$:
    \begin{gather*}\left|\Pr\left[\cA^{\Enc_\sk(\cdot)}(\ct_0) = 1 \pST
        \begin{array}{l}
        \sk\gets\Gen(1^\secp)\\
        \ct_0\gets\Enc(\sk,x_0)\\
        \end{array}
        \right]  \right.\\
    -\left. \Pr\left[\cA^{\Enc_\sk(\cdot)}(\ct_1) = 1 \pST
        \begin{array}{l}
        \sk\gets\Gen(1^\secp)\\
        \ct_1\gets\Enc(\sk,x_1)\\
        \end{array}
        \right]\right| \\
    \le \sf{negl}(\secp)\enspace.
    \end{gather*}
\item \textbf{Relaxed compactness:} We do \emph{not} need the typical notion of homomorphic encryption compactness, which states that decryption must not depend on the evaluated circuit. However, we do require the following (much weaker) properties:
\begin{itemize}
\item \emph{Classical-quantum (CQ) compactness.} There is a PPT algorithm $\mathsf{DecClassical}$ which, for any entirely classical ciphertext $\ct'$ originating from the experiment `Game 2' described in `Correctness with auxiliary input', satisfies the definition of correctness with auxiliary input when substituted for $\Dec$.
\item \emph{Compactness of encryption.} The running time of $\Enc$ on input single-bit messages should not be larger than $\poly(L) \cdot \poly\log(S) \cdot \poly(\lambda)$, where $L$ is a depth upper bound on the circuits that need to be evaluated and $S$ is a size upper bound. (Note that this restriction does not apply to $\Eval$.)
\end{itemize}
\end{itemize}

Compactness of encryption, which we require in order to achieve the question-succinctness of \Cref{prot:main}, is a property satisfied by most levelled homomorphic encryption schemes. This property suffices for us because all of the `Alice' computations in \Cref{prot:main-2-prover} can be done by log depth circuits (log in the instance size $n$), and for question succinctness we just need the verifier's questions to be $\poly\log n$ in length (and computable in the same time).

\end{definition}

\section{Approximate representation theory}

\subsection{Gowers-Hatami theorem with non-uniform measures}

The Gowers-Hatami theorem~\cite{GH15} states that if a function $f: G \to U(\cH)$ from a finite group to unitary matrices approximately behaves like a representation, then it can be ``rounded'' to an exact representation.
The formulation of this statement in~\cite{GH15} required that the ``approximate representation behaviour'' hold on average over uniformly sampled pairs of group elements.
As explained in the introduction, we require a similar statement, but for arbitrary measures over the group.
Formally, we show the following.

\begin{theorem} \label{lem:non-uniform-gh}
    Let $f: G \to U(\cH)$. Then there exists a finite-dimensional Hilbert space $\cH'$, an isometry $V: \cH \to \cH'$, and a unitary representation $\pi: G \to U(\cH')$ of $G$ such that for all measures $\mu$ over $G$ and all normalised states $\psi$,
    \[ \E_{g \sim \mu, h \sim W_n} \| f(h) f(g) - f(hg) \|_\psi^2 \leq \eps \implies \E_{g \sim \mu} \| f(g) - V^\dagger \pi(g) V \|_\psi^2 \leq \eps. \]
\end{theorem}
\begin{proof}
    Let $\cH_G$ be the complex space of formal sums of group elements and define a representation $\pi: G \to L(\cH_G)$ by 
    \begin{align*}
        \pi(g) = \sum_{h \sim W_n} \ket{h}\!\bra{hg} \,.
    \end{align*}
    It is clear that $\pi$ is a unitary representation of $G$ (in fact, $\pi$ is just the regular representation, written in Dirac notation).

    We define the convolution of $f$ with itself as 
    \begin{align*}
        f^*(g) = \E_{h \sim W_n} f(h)^\dagger f(hg) \,.
   \end{align*}
   With this, we can define a superoperator $\phi: L(\cH_G) \to L(\cH)$ by taking the linear extension of
   \begin{align*}
    \phi(\ket{g}\!\bra{h}) = \frac{1}{|G|} f^*(g^{-1} h) \,.
   \end{align*}
   Since $f(g)$ is always unitary, it follows immediately that $\phi$ is unital, i.e.~$\phi(\1) = \1$.
   We further claim that $\phi$ is completely positive.
   To show this, it suffices to show that the Choi operator 
\begin{equation*}
C = \sum_{g, h} \phi( \ket{g} \! \bra{h} ) \ot \ket{g} \! \bra{h} \in L(\cH \ot \cH_G)
\end{equation*}
is positive. For this, take any $\ket{x} \in \cH \ot \cH_G$ and decompose it as 
\begin{equation*}
\ket{x} = \sum_{g} \ket{x_g} \ket{g} \,.
\end{equation*}
Here, $\ket{x_g}$ are not necessarily orthogonal or normalised.
Then, 
\begin{align*}
\bra{x} C \ket{x} 
&= \sum_{g, h} \bra{x_g} \phi( \ket{g} \! \bra{h} ) \ket{x_h} \\
&= \frac{1}{|G|} \sum_{g, h} \bra{x_g} f^*(g^{-1} h) \ket{x_h} \\
&= \frac{1}{|G|^2} \sum_{g, h, l} \bra{x_g} f(l)^\dagger f(l g^{-1} h) \ket{x_h} \\
&= \frac{1}{|G|^2} \sum_{g, h, l} \bra{x_g} f(l g)^\dagger f(l h) \ket{x_h} \\
&= \frac{1}{|G|^2} \sum_{l} \langle \xi_l | \xi_l\rangle \geq 0 \,,
\end{align*}
with $\ket{\xi_l} = \sum_{g} f(l g) \ket{x_g}$, where we relabelled $l \mapsto l g$ in the fourth line.

Having shown that $\phi$ is unital and completely positive, we can apply Stinespring's dilation theorem to deduce that there exists some auxiliary space $\cK$ and an isometry $V: \cH \to \cH_G \otimes \cK$ such that for all $A \in L(\cH_G)$,
\begin{align*}
\phi(A) = V^\dagger(A \otimes \1_\cK)V \,.
\end{align*}

To relate this back to the original question of rounding the approximate representation $f$, we observe that 
\begin{align*}
f^*(g) = \phi(\pi(g)) = V^\dagger(\pi(g) \otimes \1_\cK)V \,.
\end{align*}
Taking the isometry in the theorem statement to be $V$ and the exact representation $\pi(g) \otimes \1_\cK$ (which, in a slight abuse of notation, we again just call $\pi$), we can therefore bound
\begin{align*}
\E_{g \sim \mu} \| f(g) - V^\dagger \pi(g) V \|_\psi^2
&= \E_{g \sim \mu} \| f(g) - f^*(g) \|_\psi^2 \\
&= \E_{g \sim \mu} \| \E_{h \sim W_n} (f(g) -  f(h)^\dagger f(hg)) \|_\psi^2 \\
&\leq \E_{g \sim \mu, h \sim W_n} \| f(g) - f(h)^\dagger f(hg) \|_\psi^2 \\
&= \E_{g \sim \mu, h \sim W_n} \| f(h) f(g) - f(hg) \|_\psi^2 \leq \eps \,,
\end{align*}
where the inequality follows from Cauchy-Schwarz and the last line uses left-unitary invariance of the state dependent norm (\cref{item:left_unitary_inv}) and the assumption that $f$ is an approximate representation.
\end{proof}

\begin{remark}
    Observe that if $f$ is \emph{exactly} left-multiplicative over some $\mu$, i.e.~if $f(g)f(h) = f(gh)$ for all $g \in \mathrm{supp}(\mu)$ and for all $h \in G$, then $f(g) = V^\dagger \pi(g) V$ for all $g \in \mathrm{supp}(\mu)$. 
\end{remark}

\subsection{Switching distributions for approximate Heisenberg-Weyl group representations}

We will be particularly concerned with approximate representations of the Heisenberg-Weyl group.
The following lemma shows that if a collection of operators $\{Z(a), X(b)\}_{a, b \in \bits^n}$ satisfies certain properties for \emph{uniform} expectations over the group, these operators can be rounded to Pauli operators for \emph{arbitrary} expectations.
We note that this switch from uniform to arbitrary measures does not work in general and crucially relies on each of $\{Z(a)\}$ and $\{X(b)\}$ being an \emph{exact} representation of $\Z_2^n$. 
\begin{lemma} \label{lem:pauli_rounding_all_dist}
Suppose we have collections of efficient binary observables $\{Z(a)\}_{a \in \bits^n}$ and $\{X(b)\}_{b \in \bits^n}$ on a Hilbert space $\cH$ that each form exact representations of $\Z_2^n$, i.e.~for all $a, b \in \bits^n$, $Z(a+b) = Z(a) Z(b)$ and $X(a+b) = X(a) X(b)$.
Consider a state $\psi$ and suppose that the following conditions hold: 
Suppose that
\begin{align}
\E_a Z(a) \psi Z(a) \capprox \psi \,, \\
\E_a X(a) \psi X(a) \capprox \psi \,, \\
\E_{a, b \in \bits^n} \norm{Z(a) X(b) - (-1)^{a \cdot b} X(b) Z(a)}_{\psi}^2 \leq \eps \,. \label{eqn:unif_comm_assumption}
\end{align}
Then there exists an isometry $V: \cH \to \C^{2^n} \otimes \cH_{\rm aux}$ such that for \emph{all distributions} $\mu$ on $\bits^n \times \bits^n$, 
\begin{align*}
\E_{(a,b) \sim \mu} \norm{Z(a)X(b) - V^\dagger (\sigma_Z(a) \sigma_X(b) \ot \1_{\cH_{\rm aux}})V}_{\psi}^2 \leq O(\eps) \,.
\end{align*}
\end{lemma}

\begin{proof}
Define the following function from the Heisenberg-Weyl group to $U(\cH)$: 
\begin{align*}
f(\pm \sigma_Z(a) \sigma_X(b)) = \pm Z(a) X(b) \,.
\end{align*}
Our proof strategy is as follows: first show that $f$ is an approximate representation over arbitrary measures $\mu$, i.e.~that it satisfies the hypothesis of \cref{lem:non-uniform-gh}; then use \cref{lem:non-uniform-gh} and the structure of the irreps of the Heisenberg-Weyl group to round $Z(a)$ and $X(b)$ to the corresponding Pauli operators.

To show that $f$ is an approximate representation, we first bound the following quantity for an arbitrary measure $\mu$ on $\bits^n \times \bits^n$.
\begin{align*}
&\E_{a, b \in \bits^n} \E_{(c, d) \sim \mu} \norm{X(a) Z(b) X(c) Z(d) - (-1)^{b \cdot c} X(a + c) Z(b+d)}_{\psi}^2 \\
\intertext{By \cref{lem:computational_state_switching} and \Cref{lem:state_dep_norm_props}~\Cref{item:linearity}:}
&\approx_{\negl} \E_{u \in \bits^n} \E_{a, b \in \bits^n} \E_{(c, d) \sim \mu} \norm{X(a) Z(b) X(c) Z(d) - (-1)^{b \cdot c} X(a + c) Z(b+d)}_{Z(u) \psi Z(u)}^2 \\
\intertext{By \cref{lem:state_dep_norm_props} \cref{item:conj_mult}:}
&=\E_{u \in \bits^n} \E_{a, b \in \bits^n} \E_{(c, d) \sim \mu} \norm{X(a) Z(b) X(c) Z(d + u) - (-1)^{b \cdot c} X(a + c) Z(b+d+u)}_{\psi}^2 \\
\intertext{Again applying \cref{lem:computational_state_switching} and \cref{lem:state_dep_norm_props} \cref{item:conj_mult} and \cref{item:linearity} to $\E_v X(v) \psi X(v) \capprox \psi$:}
&\approx_{\negl}\E_{u,v \in \bits^n} \E_{a, b \in \bits^n} \E_{(c, d) \sim \mu} \norm{X(a) Z(b) X(c) Z(d + u) X(v) - (-1)^{b \cdot c} X(a + c) Z(b+d+u) X(v)}_{\psi}^2 \\
\intertext{By \cref{eqn:unif_comm_assumption} (noting that the marginal distributions of $d+u$ and $b + d + u$ are both uniform and independent from $v$) and \cref{lem:state_dep_norm_props} \cref{item:submult_infty}:}
&=\E_{u,v \in \bits^n} \E_{a, b \in \bits^n} \E_{(c, d) \sim \mu} \norm{(-1)^{v \cdot (d + u)} X(a) Z(b) X(c+v) Z(d + u) - (-1)^{b \cdot c + v \cdot (b + d + u)} X(a + c + v) Z(b+d+u)}_{\psi}^2 + O(\eps) \\
\intertext{Repeating the same step (now noting that $c+v$ and $d+u$ are independent and uniform, and likewise for $a+c+v$ and $b+d+u$) and cancelling phases:}
&=\E_{u,v \in \bits^n} \E_{a, b \in \bits^n} \E_{(c, d) \sim \mu} \norm{(-1)^{c\cdot (d+u)} X(a) Z(b+d+u) X(c+v) - (-1)^{b \cdot c + v \cdot (b + d + u)} X(a + c + v) Z(b+d+u)}_{\psi}^2 + O(\eps) \\
\intertext{With one more repetition of the same step, we finally get:}
&=\E_{u,v \in \bits^n} \E_{a, b \in \bits^n} \E_{(c, d) \sim \mu} \norm{X(a+c+v) Z(b+d+u) - X(a + c + v) Z(b+d+u)}_{\psi}^2 + O(\eps) = O(\eps)\,.
\end{align*}

From the above calculation and the fact that $f$ is exactly multiplicative under $\pm 1$ we get that $f$ is an approximate representation of the Heisenberg-Weyl group $W_n$ for any distribution $\mu$ over the group.
Expressed in equations, we have shown that for any distribution $\mu$ on $W_n$, 
\begin{align*}
\E_{g \sim \mu, h \sim W_n} \| f(h) f(g) - f(hg) \|_\psi^2 \leq O(\eps) \,.
\end{align*}

Furthermore from the definition of $f$ we have that $f$ is exactly multiplicative under $\pm \1$, which we can also write (slightly cumbersomely) as
\begin{align*}
\E_{g \sim \mu_-, h \sim W_n} \norm{f(h) f(g) - f(h g)}^2_f = 0 \,,
\end{align*}
where $\mu_-$ is the point distribution that has all its weight on the $-\1$-element of $W_n$.

We can therefore apply \cref{lem:non-uniform-gh} and, noting that in that theorem the representation $\pi$ and the isometry $V$ are independent of the distribution $\mu$, get that there exists an isometry $V: \cH \to \cH'$ and a unitary representation $\pi$ of $W_n$ on $\cH'$ such that 
\begin{align}
\E_{g \sim \mu} \norm{f(g) - V^\dagger \pi(g) V}_\psi^2 \leq O(\eps) \,, \label{eqn:close_some_rep}\\
\E_{g \sim \mu_-} \norm{f(g) - V^\dagger \pi(g) V}_\psi^2 = \norm{\1 + V^\dagger \pi(-\1) V}_\psi^2 = 0 \label{eqn:close_minus1}\,.
\end{align}
For \cref{eqn:close_minus1}, we used that $f(-\1) = -\1$.

\cref{eqn:close_some_rep} shows that $f$ can be rounded to some representation of the Heisenberg-Weyl group.
However, we want a stronger statement: we want to show that we can round $f$ to the actual Pauli matrices (tensored with identity), whereas an arbitrary representation might contain other representations of $W_n$.

For this, we write $\pi(g) = \pi_{+}(g) \oplus \pi_{-}(g)$ where $\pi_{\pm}(\cdot)$ is a representation satisfying $\pi_{\pm}(-g) = \pm \pi_{\pm}(g)$ for all $g$. By \Cref{lem:hw-representations}, any representation of the Pauli group can be written this way, by decomposing it into a direct sum of irreducible representations, and grouping together all the one-dimensional irreps into $\pi_{+}$, and all the copies of the fundamental representation into $\pi_{-}$. Moreover, by padding with copies of the trivial representation (and adding all-$0$ rows to $V$ appropriately), we can ensure that $\dim(\pi_{+}) = k \dim(\pi_{-})$. As such, we assume without loss of generality that $\dim(\pi_{+}) = k \dim(\pi_{-})$.

The intuition for the remainder of the proof is as follows: we will simply replace the $\pi_+$-part of the representation with $k$ copies of the $\pi_-$-part; call this modified representation $\pi'$.
Of course $\pi$ and $\pi'$ now differ a lot as representations, but we need to show that after applying the isometry $V$ and in the state-dependent norm, this difference does not matter.
This will follow from \cref{eqn:close_minus1}.

More formally, we define 
\begin{align*}
    \pi'(g) &:= \begin{pmatrix} \pi_{-}(g) \otimes I_k  & \\ & \pi_{-}(g) \end{pmatrix} \\
    h'(g) &:= V^\dagger \pi'(g) V\,.
\end{align*}
We will also use the notational shorthand
\[ h(g) = V^\dagger \pi(g) V.\]

Note that because $\pi_-$ only consists of copies of the fundamental representation (the Pauli matrices), so does $\pi'$.
This means that $\pi'$ is of the form we claimed in the theorem, i.e.~it is a tensor product of the Pauli matrices with identities.

It therefore suffices to show that $\norm{h'(g) - h(g)}_\psi = 0$.
To see that this is the case, note that since $\pi_-(g) \leq \1$,
\begin{align*}
    \| h'(g) - h(g) \|_\psi^2 &= \| V^\dagger \begin{pmatrix} \pi_{-}(g) \otimes I_k - \pi_{+}(g) & \\ & 0 \end{pmatrix} V \|_\psi^2 \\
    &\leq \norm{ V^\dagger \begin{pmatrix} 2 \cdot \1 & \\ & 0 \end{pmatrix} V}_\psi^2 \\
    &= \norm{ V^\dagger (\pi(\1) + \pi(-\1)) V}_\psi^2 \\
    &= \norm{\1 + V^\dagger \pi(-\1) V}_\psi^2 = 0 \,.
\end{align*}
The last line uses \cref{eqn:close_minus1}. This completes the proof.
\end{proof}

\subsection{Lifting (anti-)commutation from small-bias sets}
In~\cite{delaSalle2022spectral}, de la Salle gave a question-succinct version of the Pauli braiding test of~\cite{natarajan2017quantum}.
The key step in his proof is to show that if (anti-)commutation statements for a certain set of observables hold on average over a $\lambda$-biased set, they also hold on average over all observables (up to small corrections).
We need a similar statement, but unfortunately we cannot use the result from~\cite{delaSalle2022spectral} directly since it is proven in the Frobenius norm, whereas we need it in the state-dependent norm.
Additionally, since we want to use this step for \emph{compiled} games, we cannot use the full suite of techniques from non-local games (in particular, we cannot use prover switching).
We remedy this situation by proving the following variant of de la Salle's result.
This also gives a more elementary proof of de la Salle's original result.

\begin{definition}
A set $S \subseteq \bits^n$ is $\lambda$-biased if for all $b \in \bits^n$ such that $b \neq 0^n$, 
\begin{align*}
|\E_{a \in S} (-1)^{a \cdot b}| \leq \lambda \,.
\end{align*}
\end{definition}
It is known how to efficiently construct a $\lambda$-biased set $S$ of size $O(n/\poly(\lambda))$ using error correcting codes~\cite{NaorNaor}. In our applications, $\lambda$ will always be a universal constant.

\begin{lemma}
\label{lem:dls-commutator-close}
Let $M$ and $\{W(a)\}_{a \in \bits^n}$ be efficient binary observables acting on a Hilbert space $\cH$ satisfying $W(a) W(b) = W(a+b)$ and $\rho$ a state on $\cH$ such that 
\begin{align*}
\rho \capprox_\delta \E_a W(a) \rho W(a) \,.
\end{align*}
Further let $S \in \bits^n$ be a $\lambda$-biased set.
Then 
\begin{align*}
\E_{a \in \bits^n} \norm{W(a) M - M W(a)}_\rho^2 \leq \frac{1}{1-\lambda} \E_{a \in T} \norm{W(a) M - M W(a)}_\rho^2 + \frac{2 \delta}{1 - \lambda} \,.
\end{align*}
\end{lemma}
\begin{proof}
Let $T \subseteq \bits^n$ be an arbitrary set (later we will set $T = S$ or $T = \bits^n$).
Define $\tilde \rho = \E_a W(a) \rho W(a)$.
By \cref{lem:computational_state_switching}, for any $a \in \bits^n$ it holds that 
\begin{align*}
\norm{W(a) M - M W(a)}_\rho^2 \geq \norm{W(a) M - M W(a)}_{\tilde \rho}^2 - \delta \,.
\end{align*}
Consequently, this also holds on average over $a \in T$, i.e.~
\begin{align}
\E_{a \in T} \norm{W(a) M - M W(a)}_\rho^2 \geq \E_{a \in T} \norm{W(a) M - M W(a)}_{\tilde \rho}^2 - \delta \,. \label{eq:wmmw-rho-to-rhoprime}
\end{align}
Using that $M$ and $Z(a)$ are binary observables, 
\begin{align*}
\E_{a \in T} \norm{W(a) M - M W(a)}_{\tilde \rho}^2 = 2 - 2 \Re \E_{a \in T} \Tr \left[ M W(a) M W(a) \tilde \rho \right] \,.
\end{align*}
Our goal is therefore to relate the trace expression for $T = \Z_2^n$ and $T = S$.
For this, observe that $W(a)$ forms a unitary representation of $\Z_2^n$, which we can decompose into a direct sum of irreducible representations: 
\begin{align*}
W(a) = \tilde U D(a) \tilde U^\dagger \,, \qquad D(a) = \oplus_i \left( (-1)^{\gamma_i \cdot a} \ot \1_{m_i}  \right)
\end{align*}
for $U \in U(\cH)$ and distinct strings $\gamma_i \in \Z_2^n$ with multiplicities $m_i$ such that $\sum_i m_i = \dim \cH$.
In this basis, $\tilde \rho$ has a block-diagonal structure made up of blocks of size $m_i$, i.e.~
\begin{align*}
\tilde U^\dagger \tilde \rho \tilde U = \oplus_i \rho_i
\end{align*}
for positive semi-definite blocks $\rho_i$ of dimension $m_i$.
For each $i$,  let $W_i \in U(m_i)$ be a unitary that diagonalises $\rho_i$ and define 
\begin{align*}
U = \tilde U (\oplus_i W_i) \,.
\end{align*}
This is a unitary and by construction, 
\begin{align*}
\sigma \deq U^\dagger \tilde \rho U = \diag(\sigma_1, \dots, \sigma_{\dim \cH})
\end{align*}
is diagonal.
Furthermore, $(\oplus_i W_i)$ commutes with all matrices $D(a)$, so that we also have 
\begin{align*}
W(a) = U D(a) U^\dagger \,, \qquad D(a) = \oplus_i^{\dim \cH} \left( (-1)^{c_i \cdot a}  \right) \,,
\end{align*}
where $c_i \in \bits^n$ are the same as the $\gamma_i$ above, except that for the rest of the proof we find it more convenient to repeat strings rather than keeping track of their multiplicities by tensoring with identity.

For convenience, we define 
\begin{align*}
\Gamma = U^\dagger M U \,.
\end{align*}
Then, we can rewrite 
\begin{align*}
\E_{a \in T} \Tr \left[ M W(a) M W(a) \tilde \rho \right] 
&= \E_{a \in T} \Tr \left[ \Gamma D(a) \Gamma D(a) \sigma \right] \\
&= \E_{a \in T} \sum_{i,j,k} \Gamma_{ij} (-1)^{c_j \cdot a} \Gamma_{jk} (-1)^{c_k \cdot a} \sigma_{ki} \\
&= \sum_{i,j,k} \left( \E_{a \in T} (-1)^{(c_j + c_k) \cdot a} \right) \Gamma_{ij} \Gamma_{jk} \sigma_{ki} \\
\intertext{Since $\sigma = \diag(\sigma_i)$:}
&= \sum_{i,j} \left( \E_{a \in T} (-1)^{(c_j + c_i) \cdot a} \right) \Gamma_{ij} \Gamma_{ji} \sigma_{i} \\
\intertext{Since $M$ is Hermitian, so is $\Gamma$, so that $\Gamma_{ji} = \Gamma_{ij}^*$:}
&= \sum_{i,j} \left( \E_{a \in T} (-1)^{(c_j + c_i) \cdot a} \right) |\Gamma_{ij}|^2 \sigma_{i} \,.
\end{align*}
We now divide up the indices $i,j$ over which we sum into two sets: 
\begin{align*}
C = \{(i,j) \;|\; c_i \neq c_j \} \,, \qquad \bar C = [\dim \cH]^2 \setminus C \,.
\end{align*}
If $(i,j) \in \bar C$, then 
\begin{align*}
\E_{a \in \Z_2^n} (-1)^{(c_j + c_i) \cdot a} = \E_{a \in S} (-1)^{(c_j + c_i) \cdot a} = 1 \,.
\end{align*}
On the other hand, if $(i,j) \in C$, then since $S$ is $\lambda$-biased we have that 
\begin{align*}
\E_{a \in \Z_2^n} (-1)^{(c_j + c_i) \cdot a} &= 0 \,,\qquad 
|\E_{a \in S} (-1)^{(c_j + c_i) \cdot a}| \leq \lambda \,.
\end{align*}

We therefore get 
\begin{align*}
\E_{a \in \bits^n} \norm{W(a) M - M W(a)}_{\tilde \rho}^2 = 2 - 2 \sum_{i,j \in \bar C} |\Gamma_{ij}|^2 \sigma_{i}
\end{align*}
and, since $|\Gamma_{ij}|^2 \sigma_i \geq 0$,
\begin{align*}
\E_{a \in S} \norm{W(a) M - M W(a)}_{\tilde \rho}^2 \geq 2 - 2 \sum_{i,j \in \bar C} |\Gamma_{ij}|^2 \sigma_{i} - 2 \lambda \sum_{i,j \in C} |\Gamma_{ij}|^2 \sigma_{i} \,.
\end{align*}
We can simplify this by observing that since $\Gamma$ is a binary observable, 
\begin{align*}
\sum_{ij} |\Gamma_{ij}|^2 \sigma_{i} = \Tr[\Gamma^\dagger \Gamma \sigma] = 1 \,.
\end{align*}
As a result, 
\begin{align*}
\E_{a \in S} \norm{W(a) M - M W(a)}_{\tilde \rho}^2 &\geq 2 - 2 \sum_{i,j \in \bar C} |\Gamma_{ij}|^2 \sigma_{i} - 2 \lambda \left( 1 - \sum_{i,j \in \bar C} |\Gamma_{ij}|^2 \sigma_{i} \right)\\
&= (1- \lambda) \E_{a \in \Z_2^n} \norm{W(a) M - M W(a)}_{\tilde \rho}^2 \,.
\end{align*}
The lemma now follows by first switching from $\rho$ to $\tilde \rho$ on both sides using \Cref{eq:wmmw-rho-to-rhoprime} (incurring an error $\delta + \frac{\delta}{1-\lambda} \leq \frac{2\delta}{1 - \lambda})$ and then using the above bound.
\end{proof}

\cref{lem:dls-commutator-close} implies the following corollary on anti-commutation relations.
\begin{corollary}
    \label{lem:dls-anticom}
    Let $\{X(a)\}_{a \in \{0,1\}^n}$ and $\{Z(a)\}_{a \in \{0,1\}^n}$ be families of efficient binary observables acting on a Hilbert space $\mathcal{H}$ satisfying $W(a) W(b) = W(a+b)$ for all $W \in \{X,Z\}$ and $a,b  \in \{0,1\}^n$. Let $\rho$ be a state on $\mathcal{H}$ such that
   \[ \rho \capprox_\delta \E_a W(a) \rho W(a) \,.\]
   for all $W \in \{X,Z\}$. Further, let $S \in \{0,1\}^n$ be a $\lambda$-biased set. Then
   \[ \E_{a,b \in \{0,1\}^n} \| Z(a)X(b) - (-1)^{a \cdot b} X(b) Z(a)\|_\rho^2 \leq \frac{1}{(1 - \lambda)^2} \E_{a,b \in S} \| Z(a)X(b) - (-1)^{a \cdot b} X(b) Z(a)\|_\rho^2 +\frac{2\delta(2 - \lambda)}{(1-\lambda)^2}. \]
\end{corollary}
\begin{proof}
    Let $\rho'_{AB} = \rho_{A} \otimes \frac{1}{2^n} I_B$. For $W \in \{X, Z\}$, define
    \[ \tilde{W}(a)_{AB} = W(a)_A \otimes \sigma_W(a)_B.\]
    Observe that
    \[ [\tilde{Z}(a), \tilde{X}(b) ] = (Z(a)X(b) - (-1)^{a,b} X(b)Z(a)) \otimes \sigma_{Z}(b) \sigma_{X}(a). \]
    Moreover, for any operators $M$ on the $A$ system and $N$ on the $B$ system, it holds that
    \begin{align*}
        \| M \otimes N\|_{\rho'}^2 &= \Tr[M^\dagger M \otimes N^\dagger N \rho'] \\
        &= \frac{1}{2^n} \Tr[M^\dagger M \rho] \cdot \Tr[N^\dagger N] \\
        &= \| M \|_\rho^2 \cdot \frac{1}{2^n} \Tr[N^\dagger N].
    \end{align*}
    Thus, we have
    \begin{align*}
        \E_{a,b \in \{0,1\}^n} &\| Z(a)X(b) - (-1)^{a \cdot b} X(b) Z(a) \|_\rho^2 \nonumber \\
        &= \E_{a,b \in \{0,1\}^n} \| [ \tilde{Z}(a), \tilde{X}(b) ]\|_{\rho'}^2 \\
        &\leq \frac{1}{1 - \lambda} \E_{a \in S} \E_{b \in \{0,1\}^n} \| [ \tilde{Z}(a), \tilde{X}(b) ]\|_{\rho'}^2 + \frac{2\delta}{1 - \lambda} \\
        &\leq \frac{1}{(1 - \lambda)^2} \E_{a, b \in S} \| [ \tilde{Z}(a), \tilde{X}(b)] \|_{\rho'}^2 + 2\delta \left( \frac{1}{1 - \lambda} + \frac{1}{(1 - \lambda)^2}\right) \\
        &= \frac{1}{(1 - \lambda)^2} \E_{a, b \in S} \| Z(a)X(b) - (-1)^{a \cdot b} X(b)Z(a) \|_{\rho}^2 + 2\delta \left( \frac{1}{1 - \lambda} + \frac{1}{(1 - \lambda)^2}\right) \,.
    \end{align*}
\end{proof}

\section{Description of the question-succinct protocol}

\subsection{Compiling nonlocal games using cryptography: the KLVY transformation}
\label{sec:compiler}

Kalai, Lombardi, Vaikuntanathan and Yang give a transformation that maps a $k$-player $1$-round nonlocal game into a $2k$-message ($k$-round) interactive protocol between a single prover and verifier. For simplicity, we will only present their transformation as it is applied to two-player nonlocal games, because this is the only context in which we need to use it. The general transformation, applicable to $k$-player nonlocal games for arbitrary $k$, is described in~\cite[Section 3.2]{KLVY21}. The following presentation is taken with some modifications from~\cite[Section 3.1]{KLVY21}.

\cite{KLVY21} presents a $\PPT$-computable transformation that converts any $2$-prover non-local game $G$ with question set $\cQ$ and verification predicate $V$ into a single-prover protocol (associated with security parameter $\secp$), defined as follows.
\label{def:compiled-game}
Fix a quantum homomorphic encryption scheme $\LQHE=(\Gen,\Enc,\Eval,\Dec)$.
\begin{enumerate}
    \item The verifier samples $(x, y)\gets\cQ$, $\sk\gets\Gen(1^\secp)$, and $c\gets\Enc(\sk,x)$. The verifier then sends $c$ to the prover as its first message. 
    \item The prover replies with a message $\alpha$.
    \item The verifier sends $y$ to the prover in the clear.
    \item The prover replies with a message $b$.
    \item Define $a \deq \Dec(\sk,\alpha)$. The verifier accepts if and only if $V(x,y,a,b)=1$.
\end{enumerate}

\subsection{Description of question-succinct protocol for $\mathsf{QMA}$}
In the following, we will denote the syntax of the protocol in regular font, and we will denote the actions of the honest prover(s) in italics.

\begin{longfbox}[breakable=false, padding=1em, padding-right=1.8em, padding-top=1.2em, margin-top=1em, margin-bottom=1em]
\begin{protocol} Question-succinct argument system for $\mathsf{QMA}$ \label{prot:main} \end{protocol}
\textbf{Inputs:} An instance $x \in \{0,1\}^*$, an algorithm $C$ which is the verifier for a $\QMA$ promise problem $A = (A_\mathrm{yes}, A_\mathrm{no})$ such that $A \in \mathsf{QMA}$, and a security parameter $\lambda \in \mathbb{N}$.
The protocol is given by applying the KLVY compiler (described in \Cref{sec:compiler}) with security parameter $\lambda$ to \Cref{prot:main-2-prover}.
\end{longfbox}

\begin{longfbox}[breakable=false, padding=1em, padding-right=1.8em, padding-top=1.2em, margin-top=1em, margin-bottom=1em]
\begin{protocol} Question-succinct (two-prover) proof system for $\mathsf{QMA}$ \label{prot:main-2-prover} \end{protocol}
\textbf{Inputs:} An instance $x \in \{0,1\}^*$ and an algorithm $C$ which is the verifier for a $\QMA$ promise problem $A = (A_\mathrm{yes}, A_\mathrm{no})$ such that $A \in \mathsf{QMA}$.

On input $x, C$, execute the reduction given in \Cref{thm:qma-to-hamiltonian} to produce a Hamiltonian problem $(H, \alpha(n), \beta(n))$ (see \Cref{def:hamiltonian-problem}) with $\beta(n) - \alpha(n) = 1 - \mathsf{negl}(n)$, where $H$ is an $n$-qubit operator ($n = \poly(|x|)$) with the following form:
\begin{align*}
H =  \1 - \underset{{w \in D}}{\E}\sum_{u \in Q(w)} \pi^w_u
\end{align*}
where
\begin{itemize}
    \item $w \in \{\1, X, Z\}^n$ is a Pauli string,
    \item $\{\pi^w_u\}_{u}$ is the projective measurement corresponding to measuring $n$ qubits in the natural way in the Pauli bases specified by $w$ (see \Cref{eq:pauli-proj} for a formal definition),
    \item $Q(w) \subseteq \{0,1\}^n$ is a set for which membership can be decided in polynomial time in $n$ given $w$, and
    \item $D$ is a distribution over $\{\1,X,Z\}^{n}$ which can be efficiently sampled from using sampling randomness of length $2^{\poly\log(n)}$. 
\end{itemize}
In words, if $H$ is of this form, then the quantity $1 - \Tr[ H \rho ]$ for any $n$-qubit state $\rho$ can be estimated by a QPT (in $n$) verifier who samples a $w$ from $D$, measures the $n$ qubits of $\rho$ in the Pauli bases specified by $w$, obtains outcomes $a$, and accepts iff $a \in Q(w)$.

\emph{Assume that we are considering a yes-instance $x \in A_{\mathrm{yes}}$. Honest Alice receives a witness that $x \in A_{\mathrm{yes}}$ as input, and converts it in polynomial time to a $n$-qubit witness $\rho$ that the ground energy of $H$ is $\leq \alpha(n)$. Honest Alice and Bob then share $n+1$ EPR pairs between themselves; most of the tests in the following protocol use only the last $n$ EPR pairs, but the first EPR pair will be used in \Cref{prot:acomgame}.}

The verifier executes each of the following tests with Alice and Bob with equal probability.
\begin{enumerate}
    \item \textbf{Pauli braiding test.} Described in \Cref{prot:pauli-braiding}; $n$ will be the number of qubits on which $H$ acts.
    \item \textbf{Mixed-versus-pure basis test.} Described in \Cref{prot:mixed-vs-pure}; the distribution $D$ will be the $D$ such that $H = \1 -  \underset{{w \sim D}}{\E}\sum_{u \in Q(w)} \pi^w_u$.
    \item \textbf{Hamiltonian test.} Described in \Cref{prot:hamiltonian}; the distribution $D$ and the sets $\{Q(w)\}_w$ will be those such that $H = \1 - \underset{{w \sim D}}{\E}\sum_{u \in Q(w)} \pi^w_u$.
\end{enumerate}
\end{longfbox}

\begin{longfbox}[breakable=false, padding=1em, padding-right=1.8em, padding-top=1.2em, margin-top=1em, margin-bottom=1em]
\begin{protocol} Pauli braiding test with succinct questions \label{prot:pauli-braiding} \end{protocol}
\textbf{Input:} an integer $n$.

Let $S \subseteq \bits^{n}$ be a $\mu$-biased set.
The verifier picks $a,b$ uniformly at random from $S \subseteq \{0,1\}^n$. Assuming that $|S| = \poly(n)$, this process takes $O(\log n)$ bits of randomness. Let the randomness which the verifier uses in this process be denoted $r = (r_a, r_b)$, where $r_a$ determines $a$ and $r_b$ determines $b$.
\begin{enumerate}
\item \textbf{(Commutation):} If $a \cdot b = 0$, the verifier executes the commutation test (\Cref{prot:comgame}) with questions $r_a, r_b$.
\item \textbf{(Anticommutation):} If $a \cdot b = 1$, the verifier executes the anticommutation test (\Cref{prot:acomgame}) with questions $r_a, r_b$.
\end{enumerate}
\end{longfbox}

\begin{longfbox}[breakable=false, padding=1em, padding-right=1.8em, padding-top=1.2em, margin-top=1em, margin-bottom=1em]
\begin{protocol} Commutation test
\label{prot:comgame} \end{protocol}
\textbf{Input:} questions $r_a$ and $r_b$ which index two elements $a$ and $b$ in $S \subseteq \{0,1\}^n$.
\begin{enumerate}[label=\arabic*.]
\item The verifier sends $(\textsf{com}, r_a, r_b)$ to Alice, and recieves responses $(u_a, u_b)$ with $u_a, u_b \in \{0,1\}$. \emph{Honest Alice computes $a$ and $b$ from $r_a$ and $r_b$, measures $\sigma_Z(a)$ and $\sigma_X(b)$ on her last $n$ qubits, and returns the results as $u_a$ and $u_b$ respectively.}
\item The verifier picks $W \in \{X,Z\}$ uniformly at random and sends $W$ to Bob. Bob responds with $v \in \{0,1\}^n$. Note that Bob is \emph{not} told that he is playing the commutation test: he receives only the question label $W$. \emph{Honest Bob measures his last $n$ qubits in the $W$ basis.}
\item If $W=Z$, the verifier accepts iff $\prod_{i: a_i = 1} v_{i} = u_a$. If $W=X$, the verifier accepts iff $\prod_{i: b_i = 1} v_{i} = u_b$.
\end{enumerate}
\end{longfbox}

\begin{longfbox}[breakable=false, padding=1em, padding-right=1.8em, padding-top=1.2em, margin-top=1em, margin-bottom=1em]
\begin{protocol} Anticommutation test 
\label{prot:acomgame} \end{protocol}
\textbf{Input:} questions $r_a$ and $r_b$ which index two elements $a$ and $b$ in $S \subseteq \{0,1\}^n$.

In this protocol, the verifier plays a version of the Mermin-Peres Magic Square game~\cite{aravind2002simple,mermin1990simple,peres1990incompatible} with the provers, in which Alice is asked to measure three observables forming a row or column of the square, and Bob is asked to measure an observable from a single cell of the square. Both provers are instructed to use the observables labeled by $r_a$ and $r_b$ for two specific cells in the square (the top centre and centre left cells), as indicated below.

\begin{enumerate}[label=\arabic*.]
\item The verifier chooses a cell index $j \in [9]$ uniformly at random; it then chooses uniformly at random a row or a column on a $3 \times 3$ grid which contains cell $j$. Suppose that the cell indices of the 3 cells in this row or column are $(i_1, i_2, i_3)$ (one of these will be equal to $j$).
\item The verifier sends $(\mathsf{MS}, r_a, r_b, i_1, i_2, i_3)$ to Alice, and receives responses $u_1, u_2, u_3 \in \{0,1\}$. \emph{Honest Alice computes $a$ and $b$ from $r_a, r_b$. Then she measures the three $(n+1)$-qubit observables associated with cells $i_1, i_2, i_3$ in the following grid, and returns all three 1-bit outcomes to the verifier:}

\begin{tabular}{c|c|c}
$\sigma_Z \otimes \1$ & $\1 \otimes \sigma_Z(a)$ & $\sigma_Z \otimes \sigma_Z(a)$ \\
\hline
$\1 \otimes \sigma_X(b)$ & $\sigma_X \otimes \1$ & $\sigma_X \otimes \sigma_X(b)$ \\
\hline
$- \sigma_Z \otimes \sigma_X(b)$ & $- \sigma_X \otimes \sigma_Z(a)$ & $- \sigma_Z \sigma_X \otimes \sigma_Z(a) \sigma_X(b)$
\end{tabular}

\emph{Operators before the tensor product are understood always to act on the first qubit of Alice's halves of the shared EPR pairs, and operators after the tensor product on the last $n$ qubits.}
\item If $j = 2$ (i.e.~Bob's question indicates the top centre cell), the verifier sends $Z$ to Bob, receives $v \in \{0,1\}^n$ as an answer, and accepts iff $\prod_{i: a_i = 1} v_{i}$ is equal to $u_2$ (if Alice was asked a row question) or $u_1$ (if Alice was asked a column question). If $j = 4$ (i.e.~Bob's question indicates the centre left cell), the verifier sends $X$ to Bob, receives $v \in \{0,1\}^n$ as an answer, and accepts iff $\prod_{i: b_i = 1} v_{i}$ is equal to $u_1$ (if Alice was asked a row question) or $u_2$ (if Alice was asked a column question. In all other cases, the verifier sends $(\mathsf{MS}, r_a, r_b, j)$ to Bob, receives a single-bit answer $v \in \{0,1\}$, and accepts iff $v = u_k$ for the $k \in [3]$ such that $u_k = j$. \emph{Honest Bob measures all his qubits in the $W$ basis when he receives a single bit question $W$, and in all other cases uses the same strategy as honest Alice, except that he only measures a single cell instead of 3.}

\end{enumerate}
\end{longfbox}

\begin{longfbox}[breakable=false, padding=1em, padding-right=1.8em, padding-top=1.2em, margin-top=1em, margin-bottom=1em]
\begin{protocol} Mixed-versus-pure basis test 
\label{prot:mixed-vs-pure} \end{protocol}
\textbf{Input:} a distribution $D$ over $\{\1, X,Z\}^n$ which can be sampled from using no more than $2^{\poly\log n}$ random bits.

\begin{enumerate}[label=\arabic*.]
    \item The verifier selects a basis $W$ from $\{X,Z\}$ uniformly at random and sends $W$ to Alice. It receives answer $u \in \{0,1\}^n$. \emph{Honest Alice measures all of her qubits in the $W$ basis.}
    \item The verifier selects a uniformly random $b \gets \{0,1\}$.
    \begin{enumerate}
        \item If $b=0$, the verifier sends $W$ to Bob, and receives answer $v \in \{0,1\}^n$. It accepts iff $u=v$. \emph{Honest Bob measures all of his qubits in the $W$ basis.}
        \item If $b=1$, the verifier samples a Pauli string $w$ from $D$ and sends its sampling randomness (which is $\poly\log n$ bits in length) to Bob. It accepts iff for all $i$ where $w_i = W$, it is the case that $u_i = v_i$. \emph{Honest Bob measures his last $n$ qubits in the Pauli bases designated by $w$ and reports all the measurement results. For qubits where $w_i = \1$, he always reports the measurement result 0.}
    \end{enumerate}
\end{enumerate}
\end{longfbox}

\begin{longfbox}[breakable=false, padding=1em, padding-right=1.8em, padding-top=1.2em, margin-top=1em, margin-bottom=1em]
\begin{protocol} Hamiltonian test
\label{prot:hamiltonian} \end{protocol}
\textbf{Input:} a distribution $D$ over $\{\1, X,Z\}^n$ which can be sampled from using no more than $2^{\poly\log n}$ random bits, along with a collection of sets $\{Q(w) : w \in \{\1,X,Z\}^n\}$ such that $Q(w) \subseteq \{0,1\}^n$ (described in a way such that membership in $Q(w)$ can be efficiently decided given $w$).
\begin{enumerate}[label=\arabic*.]
    \item The verifier sends question $\mathsf{tele}$ to Alice, and receives in response two strings, $u_x, u_z \in \{0,1\}^n$. \emph{Honest Alice teleports her ground state to Bob through their last $n$ shared EPR pairs and reports the teleportation corrections.}
    \item The verifier samples a Pauli string $w$ from $D$ and sends its sampling randomness (which is $\poly\log n$ bits in length) to Bob. The verifier then receives measurement outcomes from Bob, corrects Bob's outcomes using Alice's reported teleportation corrections, and does the appropriate energy test. More specifically, the verifier receives answer $v$ from Bob, and computes for every $i$ such that $w_i \neq \1$
\begin{equation*}
s_i = \underbrace{v_i}_{\text{Bob's measurement}} \oplus \underbrace{[(u_x)_i]^{\1[w_i = Z]} \oplus [(u_z)_i]^{\1[w_i = X]}}_{\text{correction from Alice}}
\end{equation*}
Then the verifier sets $s_i = 0$ for every $i$ such that $w_i = \1$, and accepts iff $s \in Q(w)$. \emph{Honest Bob measures his last $n$ qubits in the Pauli bases designated by $w$ and reports all the measurement results. For qubits where $w_i = \1$, he always reports the measurement result 0.}
\end{enumerate}
\end{longfbox}

\subsection{Question types}

For convenience, we summarise the different question types that Alice and Bob may each see in \Cref{prot:main-2-prover} here. (These are not necessarily in one-to-one correspondence with the subgames of the protocol, since some subgames are indistinguishable from Bob's point of view.)

We denote the set of Alice questions by $\cQ_A$ and the set of Bob questions by $\cQ_B$.

\paragraph{Alice questions:}
\begin{itemize}
\item $(\comm, r_a, r_b)$: compute $a$ and $b$ from $r_a$ and $r_b$ and measure $\sigma_Z(a)$ and $\sigma_X(b)$.
\item $(\mathsf{MS}, r_a, r_b, i_1, i_2, i_3)$: compute $a$ and $b$ and play magic square for the cells indicated by $i_1, i_2, i_3$ with the operators for cells 1 and 5 coinciding with the $\sigma_X(a)$ and $\sigma_Z(b)$ operators.
\item $X$ or $Z$: measure all qubits in $\sigma_X$ or $\sigma_Z$ basis.
\item $w$ for Pauli string $w \in \{\1,X,Z\}^n$: do mixed basis measurement in Pauli bases given by $w$.
\item $\tele$: do teleportation measurement.
\end{itemize}

\paragraph{Bob questions:}
\begin{itemize}
\item $X$ or $Z$: measure all qubits in $\sigma_X$ or $\sigma_Z$ basis.
\item $(\mathsf{MS}, r_a, r_b, j)$: compute $a$ and $b$ and play magic square for the cell indicated by $j$ with the operators for cells 1 and 5 coinciding with the $\sigma_X(a)$ and $\sigma_Z(b)$ operators.
\item $w$ for Pauli string $w \in \{\1,X,Z\}^n$: do mixed basis measurement in Pauli bases given by $w$.
\end{itemize}

\section{Modeling and state-dependent norms for the compiled game}
\label{sec:modelling}

\subsection{Modelling the prover in any compiled game}

We recall the formalism used by \cite{NZ23} to model the prover's strategy in a compiled nonlocal game. In general, the prover's actions in any compiled game can be modeled as follows. The prover
starts with some initial (pure) state $\ket{\psi}$. In the first
round, it performs a projective measurement depending on the ciphertext
question $c$ to obtain an outcome $\alpha$, \emph{followed by a unitary} depending on $c$ and $\alpha$, to obtain a post-measurement state. In the second round, it performs a projective measurement on its residual state depending on the plaintext question $y$, to obtain an outcome $b$.

We note that both measurements can be assumed to be projective without loss of generality, by the Naimark dilation theorem. We also note that the post-measurement unitary is necessary in the compiled case, since both ``provers" act sequentially on the same register; it is usually ignored in the case of nonlocal games, since the provers act on separate subsystems.

Mathematically, we model the strategy of the prover as follows:
\begin{enumerate}
    \item The initial pure state is denoted $\ket{\psi}$. We often write $\psi$ for $\ket{\psi}\bra{\psi}$.
    \item The first measurement (the ``Alice measurement") is modeled by a collection of non-Hermitian operators $A^c_\alpha$. These satisfy the condition that for each $c$, the collection of operators $\{(A^c_\alpha)^\dagger (A^c_\alpha)\}$ forms a projective measurement.  Thus, the probability that Alice returns outcome $\alpha$ in response to question $c$ is
    \[ \Pr[\alpha] = \bra{\psi} (A^c_\alpha)^\dagger (A^c_\alpha) \ket{\psi}. \]
    \item The \emph{un-normalized post-measurement state} after receiving question $c$ and responding with answer $\alpha$ is
    \begin{equation} \ket{\psi^{c}_{\alpha}} = A^c_\alpha \ket{\psi}. \label{eq:def-post-meas-states-chsh} \end{equation}
    Note that $\| \ket{\psi^c_\alpha}\|^2 = \Pr[\alpha]$.
    The post-measurement state marginalizing over outcomes for question $c$ is the mixed state
    \[ \psi^c = \sum_\alpha \psi^c_\alpha = \sum_\alpha \ket{\psi^c_\alpha}\bra{\psi^c_\alpha}. \]
    \item The second measurement (the ``Bob measurement") is modeled by a POVM measurement $\{B^y_b\}$ for each question $y$. While this measurement can always be taken to be projective without loss of generality, in the analysis it will sometimes be convenient to construct strategies where this measurement is a POVM.
\end{enumerate}

Altogether, then, we can specify a strategy for a compiled nonlocal game by a triple $(\ket{\psi}, \{A^{c}_{\alpha}\},
\{B^y_b\})$, where each element of the triple is implicitly a function
of the security parameter $\lambda$ and the encryption key $sk$.

For notational convenience, given any Alice question $q \in \cQ_A$, we sometimes use the shorthand
    \begin{align*}
    \psi^{\Enc(q)}_\alpha = \E_{sk} \E_{c \leftarrow \Enc_{sk}(q)} \psi^{c}_\alpha \,, \qquad \psi^{\Enc(q)} = \E_{sk} \E_{c \leftarrow \Enc_{sk}(q)} \psi^{c}
    \end{align*}

\begin{remark}[Implicit expectation over keys]
The security of an encryption scheme $(\Gen, \Enc, \Dec)$ requires the encryption keys to be chosen randomly by $\Gen$. In particular, this means that certain indistinguishability statements we will want to make hold only in expectation over keys $sk$ output by $\Gen$: for example, when we claim that two Alice post-measurement (mixed) states $\psi^{\Enc(q)}$ and $\psi^{\Enc(q')}$ are indistinguishable to Bob, we really mean that, for any two-outcome measurement $\{M, I-M\}$ that can be implemented by a circuit with size $\poly(\lambda)$, there exists a negligible function $\eta(\lambda)$ such that 
\begin{equation}
\E_{sk \leftarrow \Gen(\lambda)} \left| \Tr[M \psi^{\Enc(q)}] - \Tr[M \psi^{\Enc(q')}] \right| \leq \eta(\lambda)\,.
\label{eq:crypto-binary}
\end{equation}
We will frequently leave the expectation over these keys implicit: in other words, in the above example, we may simply write
\begin{equation*}
\left| \Tr[M \psi^{\Enc(q)}] - \Tr[M \psi^{\Enc(q')}] \right| \leq \eta(\lambda)\,.
\end{equation*}
Whenever we have expressions that include both $\Enc$ and $\Dec$, it is understood that both functions use the same keys, i.e.~that there is one global implicit expectation over keys.
For example, when we write $\norm{O(\Dec(\alpha))}^2_{\psi^{\Enc(q)}}$for some observables $O(\cdot)$ that depend on the decryption of a ciphertext $\alpha$, this is understood to mean 
\begin{align*}
\E_{sk \leftarrow \Gen(\lambda)} \norm{O(\Dec_{sk}(\alpha))}^2_{\psi^{\Enc_{sk}(q)}} \,,
\end{align*}
where, analogously to before, $\psi^{\Enc_{sk}(q)} \coloneqq \E_{c \leftarrow \Enc_{sk}(q)} \psi^{c}$. We usually use this convention for the squared norm, where it is justified by the linearity property (\Cref{lem:state_dep_norm_props} \Cref{item:linearity}). 

\end{remark}

We record the following fact, which follows directly from the security of the QHE scheme used in the KLVY compiler in \cref{prot:main}, as defined in \Cref{def:QHE-aux}.
\begin{lemma}\label{lem:alice_states_indist}
For any two Alice questions $q_1, q_2 \in \cQ_A$,
\begin{align*}
\psi^{\Enc(q_1)} \capprox \psi^{\Enc(q_2)} \,.
\end{align*}
\end{lemma}

\paragraph{Measuring closeness of strategies.}

In the analysis of a compiled nonlocal game, it often occurs that we wish to show if we replace  Bob's measurements $\{B^y_b\}$ with new measurements $\{C^y_b\}$ that are close in the appropriate distance metric, then the winning probability of the strategy is approximately preserved. Specifically, we will often measure closeness in terms of the state-dependent norm on the post-measurement state, after Alice's measurement has been applied:
\[ \E_{x, y \sim D_G}  \sum_{\alpha, y} \| B^y_b - C^y_b \|_{\psi^{\Enc(x)}_\alpha}^2 \leq \eps, \]
where $D_G$ is the distribution over questions sampled in the game.

In some special cases, \Cref{lem:replace_on_state} can be applied to relate this state-dependent distance bound to the winning probability. Specifically, it applies for the case of \emph{unique games} where every answer for Bob determines a unique winning answer for Alice and vice versa; this lets us define $\psi_i$ in the lemma to be Alice's post-measurement state for this unique outcome, which satisfies the condition $\sum_i \Tr \psi_i \leq 1$. However, for general games, the situation is more complicated, and there is a nice relation between state-dependent distance and success probability only if one of the two strategies being compared consists of \emph{projective} measurements. The following lemma, based on Fact~4.31 from~\cite{NW19}, gives us the statement we will need.
\begin{lemma}
  \label{lem:close_strats}
    Let $G$ be a nonlocal game, $\ket{\psi}$ a state, $\{A^c_\alpha\}$ be a collection of compiled Alice measurements, and $\{B^y_b\}$ and $\{C^y_b\}$  be two collections of Bob measurements such that $B^y_b$ is projective, and 
    \[ \E_{x,y \sim D_G}  \sum_\alpha \sum_b  \| B^y_b - C^y_b \|_{\psi^{\Enc(x)}_\alpha}^2 \leq \eps,\]
    where $D_G$ is the distribution over question pairs in the game.
    Then the success probabilities of the strategy using $(\ket{\psi}, \{A^c_\alpha\}, \{B^y_b\})$ and the strategy using $(\ket{\psi}, \{A^c_\alpha\}, \{C^y_b\})$ are $O(\sqrt{\eps})$-close. 
\end{lemma}
\begin{proof}
  Define
  \begin{equation*}
    \Delta^y_b = C^y_b - B^y_b.
  \end{equation*}

  For any $c, \alpha, y$, let $S$ be a subset of Bob's outcome space. Then we claim that
  \begin{equation}
    \sum_{b \in S} \Tr[ C^y_b \psi^c_\alpha] \geq \sum_{b \in S} \Tr[B^y_b \psi^c_\alpha] - 2\sqrt{\sum_{b \in S} \| \Delta^y_b\|_{\psi^c_\alpha}^2 }. \label{eq:switch-strat-lower-bound}
  \end{equation}
  Before we prove this claim, let us see why it implies the conclusion of the Lemma. For each $c, \alpha, y$, let $S_{c, \alpha, y}$ be the set of answers $b$ such that $(c, y, \alpha, b)$ is an accepting question and answer tuple for the compiled game verifier. Then the success probability of the Bob strategy $\{C^y_b\}$ is given by
  \begin{align*}
    \omega^*(\{C^y_b\}) &=  \E_{x,y \sim D_G} \E_{c \leftarrow \Enc(x)} \sum_\alpha \sum_{b \in S_{b, \alpha, y}} \Tr[C^y_b \psi^c_\alpha] \\
                        &\geq  \E_{x,y \sim D_G} \E_{c \leftarrow \Enc(x)} \sum_\alpha\left(   \sum_{b \in S_{b, \alpha, y}} \Tr[B^y_b \psi^c_\alpha] -2\sqrt{\sum_{b \in S(b, \alpha, y)} \| \Delta^y_b\|_{\psi^c_\alpha}^2 } \right) \\
                        &\geq \omega^*(\{B^y_b\}) - 2 \E_{x,y \sim D_G} \E_{c \leftarrow \Enc(x)} \sum_\alpha \sqrt{\sum_{b \in S(b, \alpha, y)} \| \Delta^y_b\|_{\psi^c_\alpha}^2 } \\
                        &\geq \omega^*(\{B^y_b\}) - 2 \sqrt{\E_{x,y \sim D_G} \E_{c \leftarrow \Enc(x)} \sum_\alpha \sum_b \|\Delta^y_b\|_{\psi^c_\alpha}^2} \\
    &= \omega^*(\{B^y_b\}) - 2\sqrt{\eps}.
  \end{align*}
  Applying the same logic with $S_{c,y, \alpha}$ taken to be the set of \emph{rejecting} Bob answers $b$ yields the statement
  \begin{align*}
    1 - \omega^*(\{C^y_b\}) &\geq ( 1- \omega^*(\{B^y_b\})) - 2\sqrt{\eps} \\
    \omega^*(\{C^y_b\}) &\leq \omega^*(\{B^y_b\}) + 2\sqrt{\eps}.
  \end{align*}
  Thus, the conclusion of the Lemma follows.

  Now it remains to prove \Cref{eq:switch-strat-lower-bound}. First, let us observe that
  \begin{align*}
    C^y_b &\succeq  (C^y_b)^2 \\
          &= B^y_b + (\Delta^y_b)^2 + B^y_b \Delta^y_b + \Delta^y_b B^y_b.
  \end{align*}
  Here, the first line follows from the fact that, as a valid POVM element, $C^Y_b \preceq I$, and the second line uses the projectivity of $B^y_b$ to replace $(B^y_b)^2$ by $B^y_b$. We will now prove \Cref{eq:switch-strat-lower-bound} by lower-bounding the success probability of a ``subnormalized'' strategy using the \emph{squared} operators $(C^y_b)^2$ as POVM elements. 
\begin{align*}
     \sum_{b \in S} \Tr[C^y_b  \psi^c_\alpha ] &\geq \sum_{b \in S} \Tr[(C^y_b)^2 \psi^c_\alpha ] \\
     &=\sum_{b \in S} \Tr[B^y_b \psi^c_\alpha] + \sum_{b \in S} \| \Delta^y_b \|_{\psi^c_\alpha}^2 + \sum_{b \in S} \Tr[(B^y_b \Delta^y_b + \Delta^y_b B^y_b)\psi^c_\alpha] \\
     &\geq \sum_{b \in S} \Tr[B^y_b \psi^c_\alpha] - 2 \sum_{b \in S} \| B^y_b (\psi^c_\alpha)^{1/2} \|_2 \| \Delta^y_b (\psi^c_\alpha)^{1/2} \|_2 \\
     &\geq \sum_{b \in S} \Tr[B^y_b \psi^c_\alpha] - 2 \sqrt{\sum_{b \in S} \| B^y_b (\psi^c_\alpha)^{1/2} \|^2_2}\sqrt{\sum_{b \in S}  \| \Delta^y_b (\psi^{c}_\alpha)^{1/2} \|^2_2} \\
     &\geq \sum_{b \in S} \Tr[B^y_b \psi^c_\alpha] - 2\sqrt{\sum_{b \in S} \|\Delta^y_b\|_{\psi^c_\alpha}^2}.
\end{align*}
Here, we used H\"{o}lder for the first inequality, Cauchy-Schwarz for the second, and the fact that $\{B^y_b\}$ is a normalized measurement for the third.
\end{proof}

\subsection{Modelling the prover in \Cref{prot:main}}

We will also need to establish some notation specific to the prover's `Bob' operators in \Cref{prot:main}.

\begin{definition}[Projective measurements]
~
\begin{enumerate}
\item \emph{Pure basis measurements.} In the Pauli braiding test (\Cref{prot:pauli-braiding}), and also in the $b=0$ case of the mixed versus pure basis test (\Cref{prot:mixed-vs-pure}), Bob is asked a single bit question corresponding to a basis setting ($X$ or $Z$). For each basis setting $W \in \{X,Z\}$, we notate Bob's projective measurement after receiving question $W$ as a set of projectors $\{P^{W}_u\}_{u}$ with outcomes $u \in \{0,1\}^n$.
\item \emph{Mixed basis measurements.} In the mixed versus pure basis test (\Cref{prot:mixed-vs-pure}), Bob receives a question $w \in \{\1, X, Z\}^n$ corresponding to $n$ Pauli basis settings (one for each of $n$ qubits). For each Pauli string $w \in \{\1, X,Z\}^n$, we notate Bob's projective measurement after receiving question $w$ as a set of projectors $\{M^{w}_u\}_{u \in \bits^n}$.
We can assume without loss of generality that Bob always answers ``0'' on indices where he was asked to measure the identity, i.e.~formally we can assume that $M^w_u = 0$ if there exists an index $i \in [n]$ for which $w_i = \1$ but $u_i = 1$.
The reason that this assumption is without loss of generality is that we can always replace Bob's measurements by a post-processed version that has this property; since the verifier ignores all indices for which $w_i = \1$, this post-processing affects neither Bob's success probability nor any of the rigidity statements we show below.
\end{enumerate}
\end{definition}

\begin{definition}[Pure basis observables]
    For $W \in \{X, Z\}$ and $a \in \bits^n$, define the binary observable
    \[ W(a) \deq \sum_{u} (-1)^{ u \cdot a} P^W_u. \]
\end{definition}

\begin{definition}[Mixed basis observables]
For $w \in \{\1, X,Z\}^n$ and $a \in \bits^n$, define the binary observables 
    \begin{align*}
        O^w(a) &= \sum_{u \in \bits^n} (-1)^{a \cdot u} M^w_u \,,\\
        O^w_W(a) &= \sum_{u \in \bits^n} (-1)^{\vec 1_{w = W \wedge a = 1} \cdot u} M^w_u \,.
    \end{align*}
We also write $O^w = O^w(\vec 1)$ and $O^w_W = O^w_W(\vec 1)$.
\end{definition}

\begin{remark} \label{rem:observable_redundancy}
Recall that $M^w_u = 0$ if there exists an index $i$ for which $w_i = \1$ but $u_i = 1$.
This implies that if two bitstrings $a, a' \in \bits^n$ only differ on indices $i$ for which $w_i = \1$, then $O^w(a) = O^w(a')$.
By the same reasoning, if $a, a' \in \bits^n$ only differ on indices $i$ for which $w_i \neq W$, then $O^w_W(a) = O^w_W(a')$.

For the same reason, we can also restrict the sum over $u$ only to those indices that matter, i.e.~we can drop the sum over those $u$ for which $M^w_u = 0$.
In particular, this yields another expression for $O^w$ and $O^w_W$, which will occasionally be more convenient to use: 
\begin{align*}
O^w &= \sum_{u \in \bits^{|\{u \neq \1\}|}} (-1)^{\vec 1 \cdot u} (M^w\setr{w \neq \1})_{u} \,, \\
O^w_W &= \sum_{u \in \bits^{|\{w = W \}|}} (-1)^{\vec 1 \cdot u} (M^w\setr{w = W})_{u} \,.
\end{align*}
\end{remark}

\begin{lemma} \label{lem:O_factorises}
For all $w \in \{\1, X, Z\}^n$,
\begin{align*}
O^w = O^w_X O^w_Z \,.
\end{align*}
\end{lemma}
\begin{proof}
First observe that for any $u \in \bits^{|\{w \neq \1\}|}$, 
\begin{align*}
(M^w\setr{w = X})_{u\setr{w = X}} (M^w\setr{w = Z})_{u\setr{w = Z}}
&= \left( \sum_{a: a\setr{w = X} = u\setr{w = X}} M^w_a \right)
\left( \sum_{b: b\setr{w = Z} = u\setr{w = Z}} M^w_b \right) \\
&= \sum_{\substack{a \sth \\
a\setr{w = X} = u\setr{w = X}\\
\wedge \; a\setr{w = Z} = u\setr{w = Z}}} M^w_a \\
&= \sum_{a: a\setr{w \neq \1} = u} M^w_a \\
&= (M^w\setr{w \neq \1})_{u} \,.
\end{align*}
Here, the second equality uses orthogonality of the projectors $\{M^w_a\}_{a}$.
With the expressions from \cref{rem:observable_redundancy}, we get that 
\begin{align*}
O^w_X O^w_Z 
&= \left( \sum_{u_X \in \bits^{|\{w = X \}|}} (-1)^{\vec 1 \cdot u_X} (M^w\setr{w = X})_{u_X}  \right) \left( \sum_{u_Z \in \bits^{|\{w = Z \}|}} (-1)^{\vec 1 \cdot u_Z} (M^w\setr{w = Z})_{u_Z}  \right) \\
&= \sum_{u \in \bits^{|\{u \neq \1\}|}} (-1)^{\vec 1 \cdot u} (M^w\setr{w = X})_{u\setr{w = X}} (M^w\setr{w = Z})_{u\setr{w = Z}} \\
&= \sum_{u \in \bits^{|\{u \neq \1\}|}} (-1)^{\vec 1 \cdot u} (M^w\setr{w \neq \1})_{u} = O^w \,. \\
\end{align*}
\end{proof}

\section{Analysis of the question-succinct protocol}

\subsection{Consistency test and ``prover switching''}
\label{sec:prover-switching}

The following lemma shows that the consistency test (the $b=0$ case in \cref{prot:mixed-vs-pure} that just checks whether Alice's and Bob's answers are consistent) implies that on Alice's post-measurement state, Bob's measurement in the same basis has a definite outcome.
This is formalised by saying that Bob's measurement operator is close to identity (with a sign indicating the outcome) on Alice's post-measurement state.
\begin{lemma} \label{lem:pure_basis_sign}
Suppose a computationally efficient prover $P$ modelled as in \Cref{sec:modelling} wins with probability $1 - \eps$ in the compiled version of \cref{prot:mixed-vs-pure}.
Then for all $b \in \bits^n$ and $W \in \{X,Z\}$,
\begin{align*}
\sum_{\alpha} \norm{W(b) - (-1)^{\Dec(\alpha) \cdot b} \1}_{\psi^{\Enc(W)}_\alpha}^2 \leq O(\eps) \,.
\end{align*}
\end{lemma}
\begin{proof}
Expanding out the definition of the state-dependent distance, we see that it suffices to show that 
\begin{align*}
\E_{c \leftarrow \Enc(W)} \sum_{\alpha} (-1)^{\Dec(\alpha) \cdot b}\Tr[W(b) A^c_\alpha \psi A^c_\alpha] \geq 1 - O(\eps) \,.
\end{align*}
The winning condition of the pure-basis case in \cref{prot:mixed-vs-pure} (i.e.~the $b=0$ case in the notation of \cref{prot:mixed-vs-pure}) implies 
\begin{align}
\E_{c \leftarrow \Enc(W)} \sum_{\alpha} \Tr[P^W_{\Dec(\alpha)} A^c_{\alpha} \psi A^c_{\alpha}] \geq 1 - O(\eps) \,. \label{eqn:win_cond1}
\end{align}
Since the prover's measurements are normalised, this is equivalent to 
\begin{align}
\E_{c \leftarrow \Enc(W)} \sum_{\alpha} \sum_{v \neq \Dec(\alpha)} \Tr[P^W_{v} A^c_{\alpha} \psi A^c_{\alpha}] \leq O(\eps) \,. \label{eqn:win_cond2}
\end{align}
Inserting the definition of $W(b)$: 
\begin{align*}
&\E_{c \leftarrow \Enc(W)} \sum_{\alpha} (-1)^{\Dec(\alpha) \cdot b}\Tr[W(b) A^c_\alpha \psi A^c_\alpha]\\
&= \E_{c \leftarrow \Enc(W)} \sum_{\alpha} \sum_{v} (-1)^{(\Dec(\alpha) + v) \cdot b}\Tr[P^W_v A^c_\alpha \psi A^c_\alpha] \\
&= \E_{c \leftarrow \Enc(W)} \sum_{\alpha} \Tr[P^W_{\Dec(\alpha)} A^c_\alpha \psi A^c_\alpha] + \E_{c \leftarrow \Enc(W)} \sum_{\alpha} \sum_{v \neq \Dec(\alpha)} (-1)^{(\Dec(\alpha) + v) \cdot b}\Tr[P^W_v A^c_\alpha \psi A^c_\alpha] \\
\intertext{Since $\Tr[P^W_v A^c_\alpha \psi A^c_\alpha] \geq 0$:}
&\geq \E_{c \leftarrow \Enc(W)} \sum_{\alpha} \Tr[P^W_{\Dec(\alpha)} A^c_\alpha \psi A^c_\alpha] - \E_{c \leftarrow \Enc(W)} \sum_{\alpha} \sum_{v \neq \Dec(\alpha)} \Tr[P^W_v A^c_\alpha \psi A^c_\alpha] \geq 1 - O(\eps) \,,
\end{align*}
where in the last line we used \cref{eqn:win_cond1} to bound the first term and \cref{eqn:win_cond2} to bound the second term.
\end{proof}

We can use \cref{lem:pure_basis_sign} to show that Bob's observables $W(b)$ commute with Alice's post-measurement state $\psi^{\Enc(W)}$.
This lemma can be seen as a substitute for the ``prover switching'' technique from the non-local games literature.
\begin{lemma} \label{lem:conj_invariant}
Suppose a computationally efficient prover $P$ modelled as in \Cref{sec:modelling} wins with probability $1 - \eps$ in the compiled version of \cref{prot:mixed-vs-pure}.
Then for all $b \in \bits^n$ and $W \in \{X,Z\}$, 
\begin{align*}
\norm{W(b) \psi^{\Enc(W)} W(b) - \psi^{\Enc(W)}}_1 \leq O(\sqrt{\eps}) \,.
\end{align*}
\end{lemma}
\begin{proof}
We first use the triangle inequality:
\begin{align*}
&\norm{W(b) \psi^{\Enc(W)} W(b) - \psi^{\Enc(W)}}_1 \\
&\leq \sum_{\alpha} \norm{W(b) \psi^{\Enc(W)}_\alpha W(b) - \psi^{\Enc(W)}_\alpha}_1 \\
&\leq \sum_{\alpha} \norm{W(b) \psi^{\Enc(W)}_\alpha W(b) - (-1)^{\Dec(\alpha) \cdot b}\psi^{\Enc(W)}_\alpha W(b)}_1 + \norm{(-1)^{\Dec(\alpha) \cdot b} \psi^{\Enc(W)}_\alpha W(b) - \psi^{\Enc(W)}_\alpha}_1 \\
\intertext{For the first term, we use the unitary invariance of the 1-norm, and for the second term we use that the 1-norm is invariant under taking the dagger (and $W(b)$ and $\psi^c_\alpha$ are Hermitian):}
&= \sum_{\alpha} \norm{W(b) \psi^{\Enc(W)}_\alpha - (-1)^{\Dec(\alpha) \cdot b}\psi^{\Enc(W)}_\alpha}_1 + \norm{ W(b) \psi^{\Enc(W)}_\alpha - (-1)^{\Dec(\alpha) \cdot b} \psi^{\Enc(W)}_\alpha}_1 \\
&= 2 \sum_{\alpha} \norm{(W(b) - (-1)^{\Dec(\alpha) \cdot b}) \psi^{\Enc(W)}_\alpha}_1 \\
\intertext{We can now apply \cref{lem:pure_basis_sign} and \cref{lem:replace_on_state}:}
&\leq O(\sqrt \eps) \,.
\end{align*}
\end{proof}

\subsection{Analysis of compiled Pauli braiding test: obtaining the group relations}
\begin{lemma}\label{thm:pauli-braiding-approx-phase}
Suppose a computationally efficient prover $P$ modelled as in \Cref{sec:modelling} wins with probability $1 - \eps$ in the compiled version of \Cref{prot:pauli-braiding}.
Then the prover's observables satisfy the following properties for any Alice question $q \in \cQ_A$:
\begin{align}
\E_{a, b \in \bits^n} \norm{Z(a) X(b) - (-1)^{a \cdot b} X(b) Z(a)}_{\psi^{\Enc(q)}}^2 \leq O(\eps) + \negl(\lambda) \,. \label{eqn:unif_commutation}
\end{align}
\end{lemma}

To prove this lemma, we first examine what each subtest of the Pauli braiding test lets us conclude.

\begin{lemma}
\label{lem:anticom-test}
Suppose a computationally efficient prover $P$ succeeds with probability $1 - \eps$ in the compiled anticommutation test (\Cref{prot:acomgame}) with inputs $r_a, r_b$, and let $B^{r_a}$ and $B^{r_b}$ be the Bob observables corresponding to these questions. Then for any Alice question $x \in \cQ_A$, there exists a negligible function $\eta(\lambda)$ (depending on $P$ and on $x, r_a, r_b$) such that
\[ \| \{B^{r_a}, B^{r_b}\} \|_{\psi^{\Enc(x)}}^2 \leq O(\eps) + \eta(\lambda)\,. \]
\end{lemma}
\begin{proof}
    In Theorem~5.7 of~\cite{CMMNPSWZ24}, essentially the same statement is shown for the Magic Square game. In that theorem, it is shown that for the Bob's observables $B^2, B^4$ used for cells $2$ and $4$ in the square, the expected anticommutator is bounded by
    \[ \clef \E_{c \leftarrow \Enc(x)} \sum_\alpha \| \{B^2, B^4\} \ket{\psi^c_\alpha} \|_2^2 \leq 17280 \eps + \eta(\lambda).\]
    In our setting, $B^2 = B^{r_a}$ and $B^{4} = B^{r_b}$, and the quantity on the left-hand side is exactly the average squared state-dependent norm, by~\Cref{def:state_dep_inner_product}.
\end{proof}

\begin{lemma}[Lemma~23 of~\cite{NZ23}]
\label{lem:com-test}
    Suppose a computationally efficient prover $P$ modelled as in \Cref{sec:modelling} succeeds with probability $1-\eps$ in the compiled commutation test (\Cref{prot:comgame}) with inputs $r_a, r_b$, and let $B^{r_a}, B^{r_b}$ be Bob's observables corresponding to the questions $r_a, r_b$. Then for any Alice question $x \in \cQ_A$, there exists a negligible function $\eta(\lambda)$ (depending on $P$ and on $x, r_a, r_b$) such that
    \[ \| [ B^{r_a}, B^{r_b}] \|_{\psi^{\Enc(x)}}^2 \leq O(\eps) + \eta(\lambda)\,.\]
\end{lemma}
\begin{proof}
    Lemma~23 of~\cite{NZ23} shows that, for all $(r_a, r_b)$,
    \[ \| [B^{r_a}, B^{r_b}] \|_{\psi^{\Enc((\comm, r_a, r_b))}}^2 \leq O(\eps).\]
    To obtain the conclusion of the lemma, it suffices to use computational indistinguishability together with the fact that the commutator is efficiently measurable, which follows from \Cref{lem:computational_state_switching}.
\end{proof}

\begin{proof}[Proof of \Cref{thm:pauli-braiding-approx-phase}]
We are given that $P$ succeeds with probability $1 - \eps$ in \Cref{prot:pauli-braiding}. This means that it holds that
\begin{align*}
    &\sum_{a,b \in S : a \cdot b = 0} \frac{1}{|S|^2} \cdot (1 - \Pr[P\text{ passes commutation on }Z(a), X(b)])\nonumber \\
    &\qquad +  \sum_{a,b \in S : a \cdot b = 1} \frac{1}{|S|^2} \cdot (1-\Pr[P\text{ passes anticommutation on }Z(a), X(b)]) \nonumber \\
    &\qquad \leq \eps.
\end{align*}
For any given $a,b$, suppose $a \cdot b = 0$, and let $1-\eps_{a,b} =\Pr[P\text{ passes commutation on }Z(a), X(b)]$. Then by \Cref{lem:com-test}, it holds that
\[ \| [Z(a), X(b) ] \|_{\psi^{\Enc(q)}}^2 \leq O(\eps_{a,b}) + \eta_{a,b}(\lambda), \]
and $O(\cdot)$ is a convex function.

Likewise, suppose $a \cdot b = 1$, and let $1-\eps_{a,b} =\Pr[P\text{ passes anticommutation on }Z(a), X(b)]$. Then by \Cref{lem:anticom-test}, it holds that
\[ \| \{Z(a), X(b) \} \|_{\psi^{\Enc(q)}}^2 \leq O(\eps_{a,b}) + \eta_{a,b}(\lambda), \]
and $O(\cdot)$ is a convex function.

Putting these relations together, we have
\[ \sum_{a,b \in S} \frac{1}{|S|^2} \| Z(a) X(b) - (-1)^{a \cdot b} X(b)Z(a) \|_{\psi^{\Enc(q)}}^2 \leq O(\eps) + \eta(\lambda),\]
where $\eta(\lambda) = \max_{a,b} \eta_{a,b}(\lambda)$. 
Finally, applying \Cref{lem:dls-anticom}, we get
\[ \E_{a,b \in \{0,1\}^n} \| Z(a) X(b) - (-1)^{a \cdot b} X(b) Z(a) \|^2_{\psi^{\Enc(q)}} \leq \frac{1}{(1 - \lambda)^2} O(\eps) + \eta'(\lambda),\]
where $\eta'$ is some negligible function. Here we have used that $\psi^{\Enc(q)} \capprox_{\negl} \E_{a \in S} W(a) \psi^{\Enc(q)} W(a)$ for $W \in \{X, Z\}$, which follows from \Cref{lem:conj_invariant}.
\end{proof}

\subsection{Analysis of compiled mixed-vs-pure basis test}

\begin{theorem} \label{thm:mixed_basis_final}
Suppose a computationally efficient prover $P$ modelled as in \Cref{sec:modelling} wins with probability $1 - \eps$ in both~\Cref{prot:pauli-braiding} and \cref{prot:mixed-vs-pure}. Then there exists a Hilbert space $\cH' = \C^{2^n} \otimes \cH_{\rm aux}$ (where $n$ is the number of qubits an honest prover would use) and an isometry $V: \cH \to \cH'$ such that for every Alice question $q \in \cQ_A$
\begin{align*}
\E_{w \sim D} \E_{a \in \bits^n} \| O^w(a) - V^\dagger (\sigma_w(a) \otimes \1_{\rm aux}) V \|_{\psi^{\Enc(q)}}^2 \leq O(\sqrt\eps) + \negl(\lambda) \,.
\end{align*}
\end{theorem}

\cref{thm:mixed_basis_final} also implies that the prover's projective measurements $\{M^w_u\}$ in the mixed vs pure basis test must be close to the corresponding Pauli projectors.
More formally, we show the following statement about the prover's measurements.

\begin{corollary} \label{thm:mixed-vs-pure-projectors}
Suppose a computationally efficient prover $P$ modelled as in \Cref{sec:modelling} wins with probability $1 - \eps$ in both~\Cref{prot:pauli-braiding} and \cref{prot:mixed-vs-pure}. Then there exists a Hilbert space $\cH' = \C^{2^n} \otimes \cH_{\rm aux}$ (where $n$ is the number of qubits an honest prover would use) and an isometry $V: \cH \to \cH'$ such that for every Alice question $q \in \cQ_A$, and for any subset $S \subseteq \{0,1\}^n$,
\begin{align}
\E_{w \sim D} \sum_{u \in S} \| M^w_u - V^\dagger (\pi^w_u \otimes \1_{\rm aux}) V \|_{\psi^{\Enc(q)}}^2 \leq O(\sqrt\eps)  + \negl(\lambda) \,. \label{eqn:proj_close}
\end{align}
Here, $\{M^w_u\}$ are the prover's measurements in the mixed vs pure basis test, and $\pi^w_a$ are the Pauli projectors, i.e.
\begin{align*}
\pi^w_u = \bigotimes_i \left( \frac{\1 + (-1)^{u_i} \sigma_{w_i}}{2} \right) = \E_{a \in \bits^n} (-1)^{u \cdot a} \sigma_w(a)\,.
\end{align*}
\end{corollary}
\begin{proof}
Since each of the terms in \cref{eqn:proj_close} is non-negative, it suffices to show this for $S = \bits^n$.
Recall that for any $a \in \bits^n$,
\begin{align*}
O^w(a) &= \sum_{u \in \bits^{n}} (-1)^{a \cdot u} M^w_u \,.
\end{align*}
In other words, $O^w(a)$ and $\{M^w_u\}_u$ are related by a Fourier transform.
We can invert this Fourier transform to get that 
\begin{align*}
M^w_u = \E_{a \in \bits^n} (-1)^{u \cdot a} O^w(a) \,.
\end{align*}
The proof now is essentially that of Parseval's identity, but we spell out the details for completeness: inserting this and the expansion of $\pi^w_u$ in terms of Pauli observables, we get that for any $w \in \{\1, X, Z\}^n$,
\begin{align*}
\sum_{u \in \bits^n} \| M^w_u - V^\dagger (\pi^w_u \otimes \1_{\rm aux}) V \|_{\psi^{\Enc(q)}}^2
&= \sum_{u \in \bits^n} \| \E_{a \in \bits^n} (-1)^{u \cdot a} \underbrace{\left( O^w(a) - V^\dagger (\sigma_w(a) \otimes \1_{\rm aux}) V \right)}_{\deq \Gamma(a)} \|_{\psi^{\Enc(q)}}^2 \\
&= \sum_{u \in \bits^n} \left\langle \E_{a \in \bits^n} (-1)^{u \cdot a} \Gamma(a), \E_{a' \in \bits^n} (-1)^{u \cdot a} \Gamma(a')\right\rangle_{\psi^{\Enc(q)}} \\
&= \E_{a, a' \in \bits^n} \left( \sum_{u \in \bits^n} (-1)^{u \cdot (a + a')} \right) \left\langle \Gamma(a), \Gamma(a')\right\rangle_{\psi^{\Enc(q)}} \\
&= \E_{a, a' \in \bits^n} 2^n \delta_{a, 
 a'} \left\langle \Gamma(a), \Gamma(a')\right\rangle_{\psi^{\Enc(q)}} \\
 &= \E_{a \in \bits^n} \| \Gamma(a) \|_{\psi^{\Enc(q)}}^2 \\
&= \E_{a \in \bits^n} \| O^w(a) - V^\dagger (\sigma_w(a) \otimes \1_{\rm aux}) V \|_{\psi^{\Enc(q)}}^2
\end{align*}
Here, $\langle\cdot , \cdot \rangle_\psi$ is the state-dependent inner product defined in \cref{def:state_dep_inner_product}. 
Taking $V$ to be the isometry from \cref{thm:mixed_basis_final} and inserting the expectation over $w \sim D$, the result now follows directly from \cref{thm:mixed_basis_final}.
\end{proof}

Before proving \cref{thm:mixed_basis_final}, we introduce a piece of notation that will be useful throughout the rest of the section.
\begin{definition} \label{def:dprime}
Let $D$ be the distribution over Pauli strings in $\{\1,X,Z\}^n$ from \cref{prot:mixed-vs-pure}.
Then we define $D'$ to be distribution over Pauli strings in $\{\1,X,Z\}^n$ induced by the following sampling procedure: to sample $w' \sim D'$, first sample $w \sim D$ and $a \sim \bits^n$, then set 
\begin{align*}
w'_i = \begin{cases}
w_i & \text{ if $a_i = 1$,} \\
\1 & \text{otherwise.}
\end{cases}
\end{align*}
\end{definition}

The utility of this notation is that we can now collapse the expectations $\E_{w \sim D} \E_{a \sim \bits^n}$ into one expectation.
Concretely, this allows us to restate the conclusion of \cref{thm:mixed_basis_final} as 
\begin{align}
\E_{w \sim D'} \| O^w - V^\dagger (\sigma_w \otimes \1_{\rm aux}) V \|_{\psi^{\Enc(q)}}^2 \leq O(\sqrt\eps) \,. \label{eqn:conclusion_restate_D}
\end{align}

\begin{proof}[Proof of \cref{thm:mixed_basis_final}]
We will show that \cref{eqn:conclusion_restate_D} holds.
By the triangle inequality for the squared norm (\cref{lem:state_dep_norm_props} \cref{item:triangle_ineq_sq}) we have that 
\begin{align*}
& \E_{w \sim D'} \| O^w - V^\dagger (\sigma_w \otimes I_{d}) V \|_{\psi^{\Enc(q)}}^2 \\
&\leq 2 \E_{w \sim D'} \norm{O^w  - X(1_{\{w = X\}}) Z(1_{\{w = Z\}})}_{\psi^{\Enc(q)}}^2 + 2 \E_{w \sim D'} \norm{X(1_{\{w = X\}}) Z(1_{\{w = Z\}}) - V^\dagger (\sigma_w \otimes I_{d}) V}_{\psi^{\Enc(q)}}^2 \,.
\end{align*}
The first term is at most $O(\sqrt \eps)$ by \cref{lem:mixed_final}, which we prove below.
The second term is at most $O(\eps)$ by combining \cref{thm:pauli-braiding-approx-phase} and \cref{lem:pauli_rounding_all_dist}.
\end{proof} 

It remains to prove \cref{lem:mixed_final}.
For this, we first establish a few auxiliary lemmas.

\begin{lemma} \label{lem:mixed_basis_sign}
Suppose a computationally efficient prover $P$ modelled as in \Cref{sec:modelling} wins with probability $1 - \eps$ in \cref{prot:mixed-vs-pure}.
Then for all $W \in \{X,Z\}$,
\begin{align*}
\E_{w \sim D'} \sum_{\alpha} \norm{O^w_W - (-1)^{\Dec(\alpha) \cdot 1_{\{w = W\}} }\1}_{\psi^{\Enc(W)}_\alpha}^2 \leq O(\eps) \,.
\end{align*}
\end{lemma}
\begin{proof}
The proof is similar to \cref{lem:pure_basis_sign}.
Expanding out the definition of the state dependent distance, we see that it suffices to show 
\begin{align*}
\E_{w \sim D'} \E_{c \leftarrow \Enc(W)} \sum_{\alpha} (-1)^{\Dec(\alpha) \cdot 1_{\{w = W\}}} \Tr[O^w_Z A^c_\alpha \psi A^c_\alpha] \geq 1 - O(\eps) \,.
\end{align*}
Recalling the definition of $D'$ from \cref{def:dprime}, this is equivalent to
\begin{align}
\E_{w \sim D} \E_{a \sim \bits^n} \E_{c \leftarrow \Enc(W)} \sum_{\alpha} (-1)^{\Dec(\alpha) \cdot 1_{\{w = W \wedge a=1\}}} \Tr[O^w_Z(a) A^c_\alpha \psi A^c_\alpha] \geq 1 - O(\eps) \,. \label{eqn:modcond}
\end{align}
The winning condition of \cref{prot:mixed-vs-pure} implies that 
\begin{align}
\E_{w \sim D} \E_{c \leftarrow \Enc(W)} \sum_{\alpha} \Tr[(M^w\setrw])_{\Dec(\alpha)\setrw} A^c_\alpha \psi A^c_\alpha] \geq 1 - O(\eps) \,. \label{eqn:wincond}
\end{align}
We now need to use \cref{eqn:wincond} to show \cref{eqn:modcond}.
For this, first recall that 
\begin{align*}
O^w_W(a) &= \sum_{u \in \bits^n} (-1)^{\vec 1_{w = W \wedge a = 1} \cdot u} M^w_u \,.
\end{align*}
Inserting this definition into the l.h.s.~of \cref{eqn:modcond}, we get that 
\begin{align*}
&\E_{w \sim D} \E_{a \sim \bits^n} \E_{c \leftarrow \Enc(W)} \sum_{\alpha} (-1)^{\Dec(\alpha) \cdot \vec 1_{\{w = W \wedge a=1\}}} \Tr[O^w_Z(a) A^c_\alpha \psi A^c_\alpha]\\
&=\E_{w \sim D} \E_{a \sim \bits^n} \E_{c \leftarrow \Enc(W)} \sum_{\alpha} \sum_{u \in \bits^n} (-1)^{(\Dec(\alpha) + u) \cdot \vec 1_{\{w = W \wedge a=1\}}} \Tr[ M^w_u A^c_\alpha \psi A^c_\alpha] \\
\intertext{We now split the sum over $u$ into the terms for which $u\setr{w = W \wedge a = 1} = \Dec(\alpha)\setr{w = W \wedge a = 1}$ and the remaining terms, and note that the phase $(-1)^{(\Dec(\alpha) + u) \cdot \vec 1_{\{w = W \wedge a=1\}}}$ disappears when $u\setr{w = W \wedge a = 1} = \Dec(\alpha)\setr{w = W \wedge a = 1}$:}
&=\E_{w \sim D} \E_{a \sim \bits^n} \E_{c \leftarrow \Enc(W)} \sum_{\alpha} \sum_{\substack{u \in \bits^n \sth \\ u\setr{w = W \wedge a = 1} = \Dec(\alpha)\setr{w = W \wedge a = 1}}}  \Tr[ M^w_u A^c_\alpha \psi A^c_\alpha] \\
&\qquad+ \E_{w \sim D} \E_{a \sim \bits^n} \E_{c \leftarrow \Enc(W)} \sum_{\alpha} \sum_{\substack{u \in \bits^n \sth \\ u\setr{w = W \wedge a = 1} \neq \Dec(\alpha)\setr{w = W \wedge a = 1}}} (-1)^{(\Dec(\alpha) + u) \cdot \vec 1_{\{w = W \wedge a=1\}}} \Tr[ M^w_u A^c_\alpha \psi A^c_\alpha] \\
\intertext{For the second term we note that the trace expression is always non-negative, so we can bound}
&\geq \E_{w \sim D} \E_{a \sim \bits^n} \E_{c \leftarrow \Enc(W)} \sum_{\alpha} \sum_{\substack{u \in \bits^n \sth \\ u\setr{w = W \wedge a = 1} = \Dec(\alpha)\setr{w = W \wedge a = 1}}}  \Tr[ M^w_u A^c_\alpha \psi A^c_\alpha] \\
&\qquad- \E_{w \sim D} \E_{a \sim \bits^n} \E_{c \leftarrow \Enc(W)} \sum_{\alpha} \sum_{\substack{u \in \bits^n \sth \\ u\setr{w = W \wedge a = 1} \neq \Dec(\alpha)\setr{w = W \wedge a = 1}}} \Tr[ M^w_u A^c_\alpha \psi A^c_\alpha] \\
\intertext{Using the normalisation of the measurements, we can see that (second term) = 1 - (first term).
Using this and changing the order of sums and expectations, we can simplify this to}
&= \Big( 2 \E_{c \leftarrow \Enc(W)} \sum_{\alpha} \E_{w \sim D} \E_{a \sim \bits^n}  \sum_{\substack{u \in \bits^n \sth \\ u\setr{w = W \wedge a = 1} = \Dec(\alpha)\setr{w = W \wedge a = 1}}}  \Tr[ M^w_u A^c_\alpha \psi A^c_\alpha]  \Big) - 1 \\
\intertext{Restricting the sum over $u$ to a subset of the summands and using that each summand $\Tr[ M^w_u A^c_\alpha \psi A^c_\alpha]$ is non-negative, we can bound}
&\geq \Big( 2 \E_{c \leftarrow \Enc(W)} \sum_{\alpha} \E_{w \sim D} \E_{a \sim \bits^n}  \sum_{\substack{u \in \bits^n \sth \\ u\setr{w = W} = \Dec(\alpha)\setr{w = W}}}  \Tr[ M^w_u A^c_\alpha \psi A^c_\alpha] \Big) - 1
\intertext{Since there is no remaining dependence on $a$, we can simply remove the expectation over $a$. 
We can also rewrite this more compactly using the notation from \cref{def:reduced_meas}:}
&= \Big( 2 \E_{c \leftarrow \Enc(W)} \sum_{\alpha} \E_{w \sim D} \Tr[ (M^w\setr{w = W})_{\Dec(\alpha)\setr{w = W}} A^c_\alpha \psi A^c_\alpha] \Big) - 1
\intertext{Finally, we note that the term in parentheses is bounded in \cref{eqn:wincond}, so we get}
&\geq 1 - O(\eps) \,.
\end{align*}
This proves \cref{eqn:modcond}.
\end{proof}

\begin{lemma} \label{lem:O_indiv_close}
Suppose a computationally efficient prover $P$ modelled as in \Cref{sec:modelling} wins with probability $1 - \eps$ in \cref{prot:mixed-vs-pure}.
Then for all $W \in \{X,Z\}$,
\begin{align*}
\E_{w \sim D'} \norm{W(1_{\{w = W\}}) - O^w_W}_{\psi^{\Enc(W)}}^2 &= O(\eps) \,.
\end{align*}
\end{lemma}
\begin{proof}
By \cref{lem:state_dep_norm_props} \cref{item:linearity}:
\begin{align*}
&\E_{w \sim D'} \norm{W(1_{\{w = W\}}) - O^w_W}_{\psi^{\Enc(W)}}^2 \\
&= \E_{w \sim D'} \sum_{\alpha} \norm{W(1_{\{w = W\}}) - O^w_W}_{\psi^{\Enc(W)}_\alpha}^2 \\
\intertext{By \cref{lem:state_dep_norm_props} \cref{item:triangle_ineq_sq}:}
&\leq 2 \E_{w \sim D'} \sum_{\alpha} \norm{W(1_{\{w = W\}}) - (-1)^{\Dec(\alpha) \cdot 1_{\{w = W\}}} \1}_{\psi^{\Enc(W)}_\alpha}^2 + 2 \E_{w \sim D'} \sum_{\alpha} \norm{O^w_W - (-1)^{\Dec(\alpha) \cdot 1_{\{w = W\}} }\1}_{\psi^{\Enc(W)}_\alpha}^2 \\ 
\intertext{These are exactly the terms we bounded in \cref{lem:pure_basis_sign} (which holds for every choice of $w$ and consequently also in expectation over $w \sim D'$) and \cref{lem:mixed_basis_sign}, so we get:}
&= O(\eps) \,.
\end{align*}
\end{proof}

We are now in a position to prove \cref{lem:mixed_final}, the missing ingredient in the proof of \cref{thm:mixed_basis_final}.

\begin{lemma} \label{lem:mixed_final}
Suppose a computationally efficient prover $P$ modelled as in \Cref{sec:modelling} wins with probability $1 - \eps$ in both~\Cref{prot:pauli-braiding} and \cref{prot:mixed-vs-pure}.
Then for any Alice question $q \in \cQ_A$,
\begin{align*}
\E_{w \sim D'} \norm{O^w  - X(1_{\{w = X\}}) Z(1_{\{w = Z\}})}_{\psi^{\Enc(q)}}^2 = O(\sqrt \eps) + \negl(\lambda) \,.
\end{align*}
\end{lemma}

\begin{proof}
By \cref{lem:alice_states_indist}, for all $q \in \cQ_A$ we have that $\psi^{\Enc(q)} \capprox_{\negl(\lambda)} \psi^{\Enc(Z)}$.
We can therefore use \cref{lem:computational_state_switching} to get that
\begin{align*}
&\E_{w \sim D'} \norm{O^w  - X(1_{\{w = X\}}) Z(1_{\{w = Z\}})}_{\psi^{\Enc(q)}}^2 \\
&\leq \E_{w \sim D'} \norm{O^w  - X(1_{\{w = X\}}) Z(1_{\{w = Z\}})}_{\psi^{\Enc(Z)}}^2 + \negl(\lambda)
\intertext{By \cref{lem:conj_invariant} and \cref{lem:state_replace}:} 
&\leq \E_{w \sim D'} \norm{O^w  - X(1_{\{w = X\}}) Z(1_{\{w = Z\}})}_{\E_{a \in \bits^n} Z(a) \psi^{\Enc(Z)} Z(a)}^2 + \negl(\lambda) + O(\sqrt \eps)\\
\intertext{By \cref{lem:state_dep_norm_props} \cref{item:linearity} and \cref{lem:state_dep_norm_props} \cref{item:conj_mult}:}
&= \E_{w \sim D'} \E_{a \in \bits^n} \norm{O^w Z(a) - X(1_{\{w = X\}}) Z(1_{\{w = Z\}} + a)}_{ \psi^{\Enc(Z)}}^2 + \negl(\lambda) + O(\sqrt \eps) \\
\intertext{Since the expectation over $a$ is uniform, we can shift it by $1_{\{w = Z\}}$:}
&= \E_{w \sim D'} \E_{a \in \bits^n} \norm{O^w Z(1_{\{w = Z\}} + a) - X(1_{\{w = X\}}) Z(a)}_{ \psi^{\Enc(Z)}}^2 + \negl(\lambda) + O(\sqrt \eps) \\
&= \E_{w \sim D'} \E_{a \in \bits^n} \norm{O^w Z(1_{\{w = Z\}}) Z(a) - X(1_{\{w = X\}}) Z(a)}_{ \psi^{\Enc(Z)}}^2 + \negl(\lambda) + O(\sqrt \eps) \\
\intertext{Now performing the same steps in reverse:}
&= \E_{w \sim D'} \E_{a \in \bits^n} \norm{O^w Z(1_{\{w = Z\}}) - X(1_{\{w = X\}})}_{Z(a) \psi^{\Enc(Z)} Z(a)}^2 + \negl(\lambda) + O(\sqrt \eps) \\
&\leq \E_{w \sim D'} \norm{O^w Z(1_{\{w = Z\}}) - X(1_{\{w = X\}})}_{\psi^{\Enc(Z)}}^2 + \negl(\lambda) + O(\sqrt \eps)
\intertext{From \cref{lem:O_factorises} we have that $O^w = O^w_X O^w_Z$. 
Inserting this and using \cref{lem:state_dep_norm_props} \cref{item:left_unitary_inv}, we get that:}
&= \E_{w \sim D'} \norm{O^w_Z Z(1_{\{w = Z\}}) - O^w_X X(1_{\{w = X\}})}_{\psi^{\Enc(Z)}}^2 + \negl(\lambda) + O(\sqrt \eps)
\intertext{By \cref{lem:state_dep_norm_props} \cref{item:triangle_ineq_sq}:}
&\leq 2 \E_{w \sim D'} \norm{O^w_Z Z(1_{\{w = Z\}}) - \1}_{\psi^{\Enc(Z)}}^2 + \norm{O^w_X X(1_{\{w = X\}}) - \1}_{\psi^{\Enc(Z)}}^2 + \negl(\lambda) + O(\sqrt \eps) \\
\intertext{By \cref{lem:state_dep_norm_props} \cref{item:left_unitary_inv}:}
&\leq 2 \E_{w \sim D'} \norm{Z(1_{\{w = Z\}}) - O^w_Z}_{\psi^{\Enc(Z)}}^2 + \norm{X(1_{\{w = X\}}) - O^w_X}_{\psi^{\Enc(Z)}}^2 + \negl(\lambda) + O(\sqrt \eps) \\
\intertext{By \cref{lem:O_indiv_close}:}
&\leq \negl(\lambda) + O(\sqrt \eps) \,.
\end{align*}
\end{proof}

\subsection{Subsampled Hamiltonian}
\label{sec:hamiltonian-subsamp}

\begin{definition}[Hamiltonian problem]
\label{def:hamiltonian-problem}
We refer to a tuple of the form $(H, \alpha, \beta)$, where $H$ is a Hermitian operator and $\alpha$ and $\beta$ are both real numbers, as a \emph{Hamiltonian problem}. We may refer to such a tuple where $H$ acts on $n$ qubits as $n$-qubit Hamiltonian problems.
\end{definition}

\begin{definition}[deciding a Hamiltonian problem]
Given a tuple of the form $(H, \alpha, \beta)$, where $H$ is a Hermitian operator and $\alpha$ and $\beta$ are both real numbers, we refer to the problem of deciding whether the ground energy of $H$ is $\leq \alpha$ (yes case) or $\geq \beta$ (no case) as the problem of \emph{deciding} $(H, \alpha, \beta)$.
\end{definition}

\begin{definition}[family of Hamiltonians]
We will use the notation $\cH = \{\cH(n)\}_{n \in \mathbb{N}}$ to refer to a family of Hamiltonians, i.e., a collection of sets $\cH(n)$ (one for every value of $n \in \mathbb{N}$) such that the $n$th set $\cH(n)$ contains only $n$-qubit Hamiltonian problems. We assume that the length of the binary description of $(H, \alpha, \beta)$ for any $(H, \alpha, \beta) \in \cH(n)$ is $\poly(n)$.
\end{definition}

\begin{definition}[QMA-completeness of a family of Hamiltonians]
We say a family of Hamiltonians $\cH = \{\cH(n)\}_{n \in \mathbb{N}}$ is QMA-complete if, for every promise problem $A = (A_\mathrm{yes}, A_\mathrm{no})$ in $\mathsf{QMA}$, and every instance $x \in \{0,1\}^*$, the problem of deciding whether $x \in A_\mathrm{yes}$ or $x \in A_\mathrm{no}$ can be Karp reduced in polynomial time (in $|x|$, given as input $x$ and also the algorithm $C$ which characterises the verifier for $A$) to the problem of deciding some element of $\cH(n)$ for $n = \poly(|x|)$.
\end{definition}

\begin{definition}[2-local X/Z Hamiltonian family with inverse polynomial promise gap]
A 2-local X/Z Hamiltonian family is a family of Hamiltonians where $\cH(n)$ contains only $(H, \alpha, \beta)$ such that $H$ is a sum of at most $m(n)$ terms for $m(n) = \poly(n)$, each one of which is an $n$-qubit tensor product of $\sigma_Z$, $\sigma_X$ and $\1$ such that at most 2 of the $n$ factors in the tensor product are not $\1$. If, in addition, the tuples in $\cH(n)$ all have the same values of $\alpha$ and $\beta$ (for every $n$), and it is the case that $\beta(n) - \alpha(n)$ (where $\alpha$ and $\beta$ are considered as functions of $n$) is lower bounded by $\frac{1}{\poly(n)}$, then we say $\cH = \{\cH(n)\}_{n \in \mathbb{N}}$ is a 2-local X/Z Hamiltonian family with inverse-polynomial promise gap.
\end{definition}

\begin{lemma}
\label{lem:xz-complete}
There is a family of 2-local X/Z Hamiltonians with inverse-polynomial promise gap which is $\mathsf{QMA}$-complete.
\end{lemma}
\begin{proof}
This is Theorem~2 of~\cite{biamonte2008realizable}.
\end{proof}

\begin{lemma}
\label{lem:ksv-ampl}
Let $\cH = \{\cH(n)\}_{n \in \mathbb{N}}$ be a $\mathsf{QMA}$-complete 2-local X/Z Hamiltonian family. Then there exists a polynomial $t(n)$ (the amplification parameter) such that deciding any $(H, \alpha(n), \beta(n)) \in \cH(n)$ can be efficiently reduced to deciding a problem of the form $(H', \alpha'(n), \beta'(n))$, where there exists a negligible function $\mathsf{negl}(\cdot)$ such that $\beta'(n) - \alpha'(n) \geq 1 - \mathsf{negl}(n)$, and $H'$ (which acts on $t(n) \cdot n$ qubits) has the following form:
\begin{align*}
H' = \1 - \underset{w \sim D}{\E}\sum_{a \in T(w)} \pi^w_a
\end{align*}
where $D$ is a distribution over $\{\1,X,Z\}^{t(n) \cdot n}$ which can be efficiently sampled from, and membership in $T(w) \subseteq \{0,1\}^{t(n) \cdot n}$ can be decided in time polynomial in $n$ given $w$.
\end{lemma}
\begin{proof}
We do this reduction in two stages. Firstly, we note that instead of taking a 2-local X/Z Hamiltonian problem $(H, \alpha(n), \beta(n))$ as our starting point, we can take a tuple $(H'', \alpha''(n), \beta''(n))$ where $H''$ is in the form
\begin{align*}
H'' = \1 - \underset{{w \sim D''}}{\E}\sum_{a \in S(w)} \pi^{w}_{a}
\end{align*}
where $D''$ is an efficiently sampleable distribution over $\{\1,X,Z\}^n$, and $S(w) \subseteq \{0,1\}^n$ is a set for which membership can be decided efficiently given $w$. (The `$\1 -$' is merely to make sure that the yes case is the low-energy case.) The promise gap will decrease at most by a factor of $\frac{1}{\poly(n)}$ as a result of this conversion (so that $\beta''(n) - \alpha''(n)$ is still at least $\frac{1}{\poly(n)}$). To get $H''$, we will use the Morimae-Fitzsimons energy test~\cite{morimae2016post}. This procedure will have an inverse-polynomial promise gap if the original Hamiltonian had an inverse-polynomial promise gap and constant locality, as Morimae and Fitzsimons calculate. It is straightforward to see that the success probability of the procedure on any given pure state $\ket{\psi}$ can be expressed as $1 - \bra{\psi} H'' \ket{\psi}$ for some $H''$ of the form above.

Then, to amplify the gap, we apply the technique from Lemma 14.1 in Kitaev, Shen and Vyalyi~\cite{KSV02} to $H''$, to obtain the desired Hamiltonian $H'$. The amplification technique is to measure a independent randomly chosen term from $H''$ on $t(n)$ many $n$-qubit registers, and to accept if the number of accepting outcomes is above a certain threshold. Concretely,
\[ H' = \1 - \E_{w \sim D} \sum_{a \in T(w)} \pi^w_a,\]
where $D = (D'')^{\otimes t(n)}$, and $T(w)$ is defined implicitly by the following algorithm to decide membership in it:
\begin{enumerate}
\item Divide up the answer string $a$ into $t(n)$ blocks $B_i$ of $n$ bits each. Also divide up $w$ into $t(n)$ blocks $w_i$ such that each $w_i$ is in $\{\1,X,Z\}^n$.
\item For each block $B_i, i \in [t(n)]$, compute the quantity (in $\{-1,+1\}$, for notational convenience) $b_i = (-1)^{\1[a_{|B_i} \in S(w_i)]}$. This results in $t(n)$ quantities $b_1, \cdots, b_{t(n)} \in \{-1,+1\}$.
\item $a$ is in $T(w)$ iff
\begin{align*}
\sum_{i=1}^{t(n)} b_i < t(n) \cdot \frac{\alpha''(n) + \beta''(n)}{2}.
\end{align*}
\end{enumerate}
It is clear that this algorithm runs in time $\poly(n)$. As for the condition on eigenvalues of $H'$, it is shown in \cite{KSV02} that the bound $\beta'(n) - \alpha'(n) \geq 1 - \negl(n)$ holds if $t(n)$ is chosen to be polynomial in $n$.
\end{proof}

\begin{lemma}
\label{lem:prg-expansion}
Fix a PRG family $\cF = \{ f_\lambda : \{0,1\}^\lambda \rightarrow \{0,1\}^{\lambda} \}_{\lambda \in \mathbb{N}}$. Assuming that $\cF$ is secure against (non-uniform) QPT adversaries, and for any function $\ell(\cdot)$ such that $\ell(x) = \omega( \log(x) )$, there exists a PRG family $\cG = \{ g_\lambda : \{0,1\}^{\ell(\lambda)} \rightarrow \{0,1\}^{\lambda} \}_{\lambda \in \mathbb{N}}$ which is secure against (non-uniform) QPT adversaries.
\end{lemma}
\begin{proof}
Standard arguments to show that PRGs can be used to construct PRFs~\cite{GGM} suffice.
\end{proof}

\begin{theorem}
\label{thm:hamiltonian-subsamp}
Let $\cH = \{\cH(n)\}_{n \in \mathbb{N}}$ be a $\mathsf{QMA}$-complete 2-local X/Z Hamiltonian family. Then there exists a polynomial $t(n)$ (the amplification parameter) such that deciding any $(H, \alpha(n), \beta(n)) \in \cH(n)$ can be efficiently reduced to deciding a problem of the form $(H', \alpha'(n), \beta'(n))$, where there exists a negligible function $\mathsf{negl}(\cdot)$ such that $\beta'(n) - \alpha'(n) \geq 1 - \mathsf{negl}(n)$, and $H'$ (which acts on $t(n) \cdot n$ qubits) has the following form:
\begin{align*}
H' = \1 - \underset{w \sim D}{\E}\sum_{a \in T(w)} \pi^w_a
\end{align*}
where membership in $T(w) \subseteq \{0,1\}^{t(n) \cdot n}$ can be decided in time polynomial in $n$ given $w$,
and $D$ is a distribution over $\{\1,X,Z\}^{t(n) \cdot n}$ which can be efficiently sampled from using sampling randomness of length $2^{\poly\log(n)}$.
\end{theorem}

\begin{proof}
Fix a PRG family $\cG = \{ g_\lambda : \{0,1\}^{\ell(\lambda)} \rightarrow \{0,1\}^{\lambda} \}_{\lambda \in \mathbb{N}}$ with $\ell(\lambda) = \poly \log (\lambda)$ which is secure against non-uniform QPT adversaries. (Such a family exists by \Cref{lem:prg-expansion}.) Take an X/Z Hamiltonian family $\cH = \{\cH(n)\}_{n \in \mathbb{N}}$ with gap $\beta(n) - \alpha(n) = 1 - \mathsf{negl}(n)$ which is QMA-complete, such that any $H$ such that $(H, \alpha(n), \beta(n)) \in \cH(n)$ is a Hamiltonian on $t(n) \cdot n$ qubits with the following form:
\begin{align}
H = \1 - \underset{w \sim \Delta}{\E}\sum_{a \in T(w)} \pi^w_a \label{eq:original-hamiltonian}
\end{align}
where $\Delta$ is a distribution over $\{\1,X,Z\}^{t(n) \cdot n}$ which can be efficiently sampled from, and membership in $T(w) \subseteq \{0,1\}^{t(n) \cdot n}$ can be decided in time polynomial in $n$ given $w$. (Such a family exists by \Cref{lem:ksv-ampl}.)

The following algorithm accepts any given state $\rho$ with probability exactly $1 - \Tr[H \rho]$:
\begin{process}
\label{proc:sample-original}
\end{process}
\begin{enumerate}
\item Sample $w$ from $D$. We will assume this takes $R(n)$ bits of randomness for some polynomial $R$.
\item Measure every qubit of $\rho$ in the Pauli bases determined by $w$. Let the outcome be $a \in \{0,1\}^{t(n) \cdot n}$.
\item Accept iff $a \in T(w)$.
\end{enumerate}

Now consider the following \emph{derandomised} version of this algorithm:
\begin{process}[Derandomised version of \Cref{proc:sample-original}]
\label{proc:sample-derandomised}
\end{process}
\begin{enumerate}
\item Use the PRG $g_{\ell(R(n))}$ (i.e.~the $\ell(R(n))$th PRG in the family $\cG$) in order to generate a pseudorandom seed $s$ of length $R(n)$ using $\ell(R(n))$ bits of true randomness. Use $s$ to sample $w$ from $D$.
\item Measure every qubit of $\rho$ in the Pauli bases determined by $w$. Let the outcome be $a \in \{0,1\}^{t(n) \cdot n}$.
\item Accept iff $a \in T(w)$.
\end{enumerate}

It is clear that this algorithm accepts any given state $\rho$ with probability exactly $1 - \Tr[H' \rho]$, where $H'$ is the following Hamiltonian:
\begin{align}
H' = \1 - \underset{w \sim D}{\E}\sum_{a \in T(w)} \pi^w_a \label{eq:derandomised-hamiltonian}
\end{align}
where we define $D$ to be the distribution which step (i) in \Cref{proc:sample-derandomised} samples from. Since $R$ is a polynomial and $\ell(n) = \poly\log(n)$, this setting of parameters yields sampling randomness for $D$ of $2^{\poly\log(n)}$, as desired.

Now we show that there exist $\alpha'(\cdot), \beta'(\cdot)$ with $\beta'(n) - \alpha'(n) \geq 1 - \mu(n)$ for some negligible function $\mu(\cdot)$ such that the problem of deciding $\cH'$ is QMA-complete. Consider a sequence of Hamiltonians $\{H(n) : n \in \mathbb{N}\}$ such that, for all $n$, $(H(n), \alpha(n), \beta(n)) \in \cH(n)$. (We fix a single Hamiltonian per value of $n$ for notational simplicity; the same argument can be applied to all the elements of $\cH(n)$.) For each $n$, let $H'(n)$ be the Hamiltonian defined by the `subsampling' procedure above applied to $H(n)$, i.e. the Hamiltonian defined in \Cref{eq:derandomised-hamiltonian}. We will show that there exists a negligible function $\negl(\cdot)$ such that the lowest eigenvalue of $H'(n)$ differs from the lowest eigenvalue of $H(n)$ by at most a $\negl(n)$ amount. 

Suppose, for contradiction's sake, that this is not the case; then, defining $v(n)$ to be the lowest eigenvalue of $H(n)$ and $v'(n)$ to be the lowest eigenvalue of $H'(n)$, (wlog) it is the case that $v'(n) - v(n) \geq p(n)$ for some non-negligible function $p(n)$. Then we construct a QPT adversary $\cA$ taking non-uniform quantum advice which can break the security of $\cG$. Let $r(\cdot)$ be a sufficiently large polynomial, to be set later. We assume (wlog by standard reductions) that, on input security parameter $\lambda$, $\cA$ either receives a challenge consisting of $r(\lambda)$ samples from the range of the PRG $g_\lambda$ or $r(\lambda)$ uniformly random strings of length $\lambda$. On input security parameter $\lambda$, $\cA$ gets as advice the following. Let $n$ be such that $R(n) = \lambda$. $\cA$ gets $r(\lambda)$ many copies of a ground state for the Hamiltonian $H(n)$; the value of $v(n)$, the lowest eigenvalue of $H(n)$; and the value of $p(n) = v'(n) - v(n)$. It then receives its challenge strings $s_1, \dots, s_r$, and does the following for each string $s_i$: use $s_i$ to sample a term from $H(n)$; measure this term of $H(n)$ on the $i$th copy of the ground state it was given as advice; compute the average of all its measurement outcomes; output `PRG' if the average is greater than $v(n) + \frac{p(n)}{2}$, and `uniform' otherwise.

This procedure is clearly polynomial time in $\lambda$. We now analyse $\cA$'s success probability. In the case where $s$ is a uniformly random string, the average energy of a term in $H(n)$ measured against the ground state of $H(n)$ is at most $v(n)$. Then, by a Chernoff bound, the probability that $\cA$ guesses wrongly in this case is upper bounded by the following calculation (letting $X$ be the random variable associated with the sum of $\cA$'s measurement outcomes):
\begin{align*}
\Pr[X \geq (1+\delta)\mu] &\leq \exp\Big(-\frac{2 \delta^2 \mu^2}{r(\lambda)}\Big) \\
\Pr[X \geq (v(n) + \frac{1}{2}p(n))r(\lambda)] &\leq \exp\bigg( 
-2\Big(\frac{1}{2}\frac{p(n)}{v(n)}\Big)^2 \cdot (v(n) \cdot r(\lambda))^2 / r(\lambda) \bigg) \\
&\leq \exp\Big(\frac{1}{2} p(n)^2 r(\lambda)\Big).
\end{align*}
Therefore, setting $r(\lambda)$ to be a sufficiently large polynomial (i.e.~sufficiently larger than $\frac{1}{p(n)}^2$), we can ensure that the probability that $\cA$ guesses wrongly in the uniform case is negligible.

The calculation for the PRG case is similar. By a union bound, then, we can conclude that $\cA$ distinguishes between the PRG and the uniform cases with $1 - \negl(\lambda)$ probability, which contradicts the non-uniform security of the PRG family $\cG$. It follows that the lowest eigenvalue of $H'(n)$ differs from the lowest eigenvalue of $H(n)$ by at most a $\negl(n)$ amount.

\end{proof}

\begin{theorem}\label{thm:qma-to-hamiltonian}
Let $C$ be a $\QMA$ verification algorithm for a $\QMA$ promise problem $A = (A_{yes}, A_{no})$, and let $x$ be an input of length $n$. Then there is a Hamiltonian problem $(H, \alpha(n), \beta(n))$ that can be efficiently computed from the input $(x, C)$, such that if $x \in A_{yes}$, then $H$ is a YES instance of the Hamiltonian problem, and if $x \in A_{no}$, then $H$ is a NO instance, and $\beta(n) - \alpha(n) = 1 - \negl(n)$.

$H$ has the form
\begin{align*}
H' = \1 - \underset{w \sim D}{\E}\sum_{a \in T(w)} \pi^w_a
\end{align*}
\end{theorem}
\begin{proof}
The statement of $\QMA$ completeness in \Cref{lem:xz-complete} implies that there is an efficient mapping $(x, C)$ to a 2-local X/Z Hamiltonian problem $(H_{XZ}, \alpha_{XZ}(n), \beta_{XZ}(n))$. Applying \Cref{thm:hamiltonian-subsamp} to this family of Hamiltonians yields the result.
\end{proof}

\subsection{Analysis of the compiled Hamiltonian test}

\begin{lemma}\label{lem:ham-test-analysis}
    Given an instance $x$, let $H(n)$ be the corresponding Hamiltonian generated in \Cref{prot:main}, and let $P$ be a prover that succeeds in the compiled Pauli braiding and mixed-versus-pure tests (\Cref{prot:pauli-braiding} and \Cref{prot:mixed-vs-pure}) with probability $1-\delta$, and in the Hamiltonian test (\Cref{prot:hamiltonian}) with probability $p_{H}$. Then there exists a state $\rho$ such that
    \[ \Tr[\rho H(n)] \leq (1 - p_{H}) + O(\delta^{1/4}) + \negl(\lambda).\]
\end{lemma}
\begin{proof}
First, since $P$ succeeds with probability $1 - \delta$ in the Pauli braiding and mixed-versus-pure tests, by \Cref{thm:mixed-vs-pure-projectors} it holds that there exists a Hilbert space $\cH' = \C^{2^n} \otimes \cH_{\rm aux}$ (where $n$ is the number of qubits an honest prover would use) and an isometry $V: \cH \to \cH'$ such that for every Alice question $q \in \cQ_A$, and for any subset $S \subseteq \{0,1\}^n$,
\begin{align}
\E_{w \sim D} \sum_{u \in S} \| M^w_u - V^\dagger (\pi^w_u \otimes \1_{\rm aux}) V \|_{\psi^{\Enc(q)}}^2 \leq O(\sqrt\delta) + \negl(\lambda) \,. \label{eqn:proj_close_repeat}
\end{align}
where $D$ is the set of Pauli strings $D(n)$ corresponding to the Hamiltonian $H(n)$.

In particular, setting $q = \tele$, we have that
\begin{equation}
    \E_{w\sim D} \sum_{u \in S} \E_{c \leftarrow \Enc(\tele)} \sum_{\alpha} \| M^w_u - V^\dagger (\pi^w_u \otimes \1_{\aux}) V\|^2_{\psi^c_\alpha} \leq O(\sqrt{\delta}) + \negl(\lambda). \label{eqn:M-close-to-pauli-on-H}
\end{equation}
where $\psi^c_\alpha = A^c_\alpha \psi (A^c_\alpha)^\dagger$.

Now, recall that the Hamiltonian $H(n)$ has the form
\[ H(n) = \1 - \E_{w \sim D} \sum_{u \in Q(w)} \pi^w_u,\]
where $\pi^w_u$ is a Pauli basis projector. From the success in the Hamiltonian test, we know that
\begin{align}
    \E_{w \sim D} \sum_{u}  \E_{c = \Enc(\tele)} \sum_{\alpha} \Tr[M^w_u A^{c}_\alpha \psi (A^{c}_\alpha)^\dagger] \cdot \mathbb{1}[\mathrm{correct}(w, u, \Dec(\alpha)) \in Q(w)] = p_H, \label{eq:ham-test-success}
\end{align}
where the function $\mathrm{correct}$ implements the Pauli corrections from the teleportation step, and is defined by
\begin{equation*}
\mathrm{correct}(w,u, (u_x, u_z))_i = \begin{cases} u_i \oplus [(u_x)_i]^{\1[w_i = Z]} \oplus [(u_z)_i]^{\1[w_i = X]} & w_i \neq \1\\
0 & w_i = \1
\end{cases}
\end{equation*}
Also note that the right hand side of \Cref{eq:ham-test-success} is $1- p_H$, because of the left we are summing over corrected outcomes that are contained in $Q(w)$; the test was defined to accept if the outcome is \emph{not} in $Q(w)$. (This is so that a high success probability corresponds to a low energy according to the Hamiltonian.)

By \Cref{eqn:M-close-to-pauli-on-H} together with \Cref{lem:close_strats}, with $B^y_b$ taken to be $M^w_u$ and $C^y_b$ taken to be $V^\dagger (\pi_u^w \otimes \1_\aux) V$, it thus follows that 
\begin{equation*} \E_{w \sim D} \sum_u \E_{c \leftarrow \Enc(\tele)} \sum_\alpha \Tr[ V^\dagger (\pi_u^w \otimes \1_\aux) V A^c_\alpha \psi (A^c_\alpha)^\dagger] \cdot \1[\mathrm{correct}(w, u, \Dec(\alpha)) \in Q(w)] \geq p_H - O(\delta^{1/4}) - \negl(\lambda).
\end{equation*}

We will now show how to construct the desired state $\rho$. We will do this by simulating the effect of the Pauli corrections by actual Pauli operators applied to Alice's post-measurement state.

First, we need to deal with a technical inconvenience: the function $\mathrm{correct}$ is defined differently for $w_i = \1$ and for $w_i \neq \1$, and in the $w_i = \1$ case, it acts by replacing the measurement outcome with $0$. Fortunately, we recall that by the definition of 
\[ \pi^w_u = \bigotimes_{i=1}^{n} \left(\frac{I + (-1)^{u_i} \sigma_{w_i}}{2}\right),\]
it follows that whenever $w_i = \1$ and $u_i \neq 0$ for any $i$, the corresponding projector $\pi^w_u = 0$. Thus, we may freely assume that whenever $w_i = \1$, $u_i = 0$. This enables us to replace $\mathrm{correct}$ with the modified function $\mathrm{correct}'$ defined by
\begin{equation*}
    \mathrm{correct}'(w,u,(u_x, u_z))_i = u_i \oplus [(u_x)_i]^{\1[w_i = Z]} \oplus [(u_z)_i]^{\1[w_i = X]}. 
\end{equation*}
With respect to this modified function, we have
\begin{align*} 
&\E_{w \sim D} \sum_u \E_{c \leftarrow \Enc(\tele)} \sum_\alpha \Tr[ V^\dagger (\pi_u^w \otimes \1_\aux) V A^c_\alpha \psi (A^c_\alpha)^\dagger] \cdot \1[\mathrm{correct}'(w, u, \Dec(\alpha)) \in Q(w)] \nonumber \\
&\qquad \geq p_H - O(\delta^{1/4}) - \negl(\lambda).
\end{align*}
By shifting the sum over $u$, this may be equivalently rewritten as 
\begin{equation*} \E_{w \sim D} \sum_{u \in Q(w)} \E_{c \leftarrow \Enc(\tele)} \sum_\alpha \Tr[ V^\dagger (\pi_{\mathrm{correct}'(w, u, (u_x, u_z))}^w \otimes \1_\aux) V A^c_\alpha \psi (A^c_\alpha)^\dagger]  \geq p_H - O(\delta^{1/4}) - \negl(\lambda).\end{equation*}

Now, we observe that the Pauli corrections have the following simple form:
\begin{equation*}
    \pi^w_{\mathrm{correct}'(w,u,(u_x, u_z))} = \sigma_Z(u_z) \sigma_X(u_x) \pi^w_u \sigma_X(u_x) \sigma_Z(u_z).
\end{equation*}
Applying this observation, plus the cyclicity of the trace, we get
\begin{align*}
    &\E_{w \sim D} \sum_{u \in Q(w)} \E_{c \leftarrow \Enc(\tele)} \sum_\alpha \Tr[  (\pi_{u}^w \otimes \1_\aux) \sigma_X(u_x) \sigma_Z(u_z) V A^c_\alpha \psi (A^c_\alpha)^\dagger V^\dagger \sigma_Z(u_z) \sigma_X(u_x)]  \nonumber \\
    &\qquad \geq p_H - O(\delta^{1/4}) - \negl(\lambda),
\end{align*}
where $(u_x, u_z) = \Dec(\alpha)$. Now, let us define
\begin{equation*}
    \rho = \Tr_{\aux} \E_{c \leftarrow \tele} \sum_\alpha \sigma_X(u_x) \sigma_Z(u_z) V A^c_\alpha \psi (A^c_\alpha)^\dagger V^\dagger \sigma_Z(u_z) \sigma_X(u_x).
\end{equation*}
It is clear that $\rho$ is a normalized quantum state. Moreover, we have that
\begin{equation*}
    \Tr[H(n) \rho] = 1 - \E_{w \sim D} \sum_{u \in Q(w)} \Tr[ \pi^w_u \rho ] \leq (1-p_h) + O(\delta^{1/4}) + \negl(\lambda),
\end{equation*}
which was the desired conclusion.
\end{proof}

\subsection{Analysis of full compiled protocol}

\begin{theorem}\label{thm:main}
    The protocol \Cref{prot:main} is a question-succinct argument system for $\QMA$ assuming a $\mathsf{QHE}$ scheme satisfying the definition given in \Cref{def:QHE-aux}.
    
More precisely, let $V$ be the verifier, and $P$ be the honest prover described in the protocol. Then for any promise problem $A = (A_{yes}, A_{no})$ in $\QMA$ with verification algorithm $C$, the following hold:
    \begin{itemize}
        \item \textbf{Completeness:} Let $x \in A_{yes}$ , and let $\ket{\psi}$ be an accepting $\QMA$ witness for $x$. Then the verifier $V$ on input $(x, C,1^\lambda)$, in interaction with the honest prover $P$ on input $(x,C, \ket{\psi}^{\poly n}, 1^\lambda)$, accepts with probability $\geq 1 - \negl(n)$.  
        \item \textbf{Soundness:} Let $x \in A_{no}$, and let $\ket{\phi}$ be \emph{any} state on $\poly(n + \lambda)$ qubits. Then the verifier $V$ on input $(x,C, 1^\lambda)$, in interaction with any QPT prover $P^*$ on input $(x,C, \ket{\phi}, 1^\lambda)$, accepts with probability at most $s$ for some universal constant $s < 1$.
        \item \textbf{Question-succinctness:} On any input $x$ and for security parameter $\lambda$, the number of bits sent by $V$ to $P$ is $O(\poly\log n + \poly(\lambda))$.
    \end{itemize}
\end{theorem}
\begin{proof}
We recall that the protocol \Cref{prot:main} is obtained by applying the KLVY compilation to \Cref{prot:main-2-prover}. This in turn consists of three subtests: the Pauli braiding test (\Cref{prot:pauli-braiding}), the mixed-versus-pure basis test (\Cref{prot:mixed-vs-pure}), and the Hamiltonian test (\Cref{prot:hamiltonian}).

\paragraph{Completeness:}
For the completeness, we observe that the honest prover passes the Pauli braiding test (\Cref{prot:pauli-braiding}) and the mixed-versus-pure basis test (\Cref{prot:mixed-vs-pure}) with certainty, and, using a state $\ket{\phi}$, passes the Hamiltonian test (\Cref{prot:hamiltonian}) with probability equal to $1 - \bra{\phi} H(n) \ket{\phi}$. Taking $H(n)$ to be the Hamiltonian computed in \Cref{prot:main}, and $\ket{\phi}$ to be the appropriate polynomial number of copies of $\ket{\psi}$, we have $\bra{\phi} H \ket{\phi} = 1 - \negl(n)$. 

\paragraph{Soundness:}
We now establish soundness of this protocol. Suppose $x$ is a NO instance of the QMA language. Then the Hamiltonian $H$ from \Cref{prot:main-2-prover} has minimum eigenvalue at least
\[ \lambda_{\min}(H) \geq 1 - \negl(n). \]
Moreover, we can assume by a padding argument that without loss of generality, that $\lambda_{\min}(H) \geq 2/3$ for all $n$.

Now, suppose that $P'$ is a prover that succeeds in \Cref{prot:main} with probability $1 - \delta$. This means that $P'$ succeeds in the KLVY compilations of the Pauli braiding test (\Cref{prot:pauli-braiding}), the mixed-versus-pure basis test (\Cref{prot:mixed-vs-pure}), and the Hamiltonian test (\Cref{prot:hamiltonian}) each with probability at least $1 - 3\delta$. By \Cref{lem:ham-test-analysis}, this means that there exists a state $\rho$ such that
\[ \tr[\rho H] \leq 3\delta + O(\delta^{1/4}).\]
For any $\delta$ below some universal $\delta_0$, the RHS of the expression above will be at most $1/2$, and thus in contradiction with the fact that $\lambda_{\min}(H) \geq 2/3$. 

\paragraph{Succinctness:} The question-succinctness property is evident from the description of the protocol. We note that the longest messages are sent in the mixed-versus-pure basis test (\Cref{prot:mixed-vs-pure}), and the Hamiltonian test \Cref{prot:hamiltonian}, where to describe a term in $H$ the verifier must send $\poly \log (n)$ bits. In all other tests, the verifier sends at most $O(\log (n))$ bits to the prover.
\end{proof}

\section{Compiling from a question-succinct protocol into a fully succinct protocol using succinct arguments of knowledge}
\label{sec:killian}

\Cref{prot:main} is a \emph{question-succinct} cryptographic single-prover protocol, in the sense that the messages the verifier sends to the prover are $\poly \log n \cdot \poly \lambda$ bits long, where $n$ is the size of the instance and $\lambda$ is the security parameter. In \cite[Section 9]{BKLMMVVY22}, Bartusek et al.~present two compilers which map any question-succinct single-prover argument system for $\mathsf{QMA}$ satisfying a certain \emph{obliviousness} property into a fully succinct single-prover argument system for $\mathsf{QMA}$. More specifically, for their compilers to work, Bartusek et al.~require that the verifier's questions in the question-succinct protocol can be computed independently of the prover's answers and also the $\mathsf{QMA}$ instance (except for its length). They also require implicitly that the verifier's questions can be generated extremely efficiently, namely, in time $\tilde{O}(n) + \poly\log(n) \cdot \poly(\lambda)$.

It can be easily verified that \Cref{prot:main} satisfies the obliviousness property. However, the efficiency property is not quite satisfied because of the commutation and anticommutation tests. Specifically, the choice of which test to execute is a function of $a \cdot b$, where $a$ and $b$ are strings of length $\poly(n)$, and so na\"{i}vely, computing this inner product takes time $\poly(n)$. However, there is an easy fix for this: after sampling $r_a,r_b$, the verifier decides uniformly at random whether to run the commutation or anticommutation test, and then afterwards, if it chose wrongly, it automatically accepts. This affects the soundness gap by at most a constant factor.

With this modification made, the verifier of \Cref{prot:main} is indeed oblivious and efficient. In fact, all the information the verifier sends to the prover in this modified version of \Cref{prot:main} takes one of the following forms:
\begin{itemize}
\item Encryptions of uniform randomness of some predetermined length.
\item Non-encrypted uniform randomness of some predetermined length.
\end{itemize}
Its decision process about which messages to send can also be made completely independent of the prover's answers. As such, the compilers presented in \cite[Section 9]{BKLMMVVY22} apply in a black-box fashion to the modified version of \Cref{prot:main}. Nonetheless, because the presentation of the analysis in \cite[Section 9]{BKLMMVVY22} is fairly terse, we present for the reader's benefit some additional intuition about how Bartusek et al.'s first compiler works when applied to \Cref{prot:main} in particular.
We emphasise that this is only for intuition and not meant as a full (re-)proof of the \cite{BKLMMVVY22} compiler.

\subsection{Post-quantum succinct arguments of knowledge}
The central building block for Bartusek et al.'s first compiler is a \emph{post-quantum succinct argument of knowledge} (based on Killian's work \cite{kilian1992note}) with the following commit-and-open structure. Fix some instance $x$, some $\mathsf{NP}$ language $L$, and some associated predicate $R_L(\cdot, \cdot)$ (that is, $R_L(x,w)$ checks whether $w$ is a valid witness that $x \in L$, and can be computed in polynomial time). Suppose also that the verifier has already sent the prover the hash key $hk$ for some \emph{collapsing}\footnote{Collapsing is a post-quantum strengthening of collision resistance; see \cite{unruh2016computationally} for a definition.} hash function family. Then the following protocol allows the prover to succinctly prove knowledge of some $w$ such that $R_L(x, w) = 1$:

\begin{longfbox}[breakable=false, padding=1em, padding-right=1.8em, padding-top=1.2em, margin-top=1em, margin-bottom=1em]
\begin{protocol}[Three-message succinct argument of knowledge]
\label{prot:saok}
\end{protocol}
\begin{enumerate}[label=\arabic*.]
\item The prover encodes the witness $w$ under an error correcting code $E$ (with corresponding decoding $D$) to obtain a string $\tilde w = E(w)$. The prover also constructs a PCPP proof $\pi$ that $R(x,D(\tilde w)) = 1$. The prover then constructs a Merkle tree \cite{merkle1987digital} on $m = (\tilde w, \pi)$ using $hk$, computes the root of this tree $rt_m$, and sends the verifier $rt_m$. The prover also constructs a Merkle tree on $w$ itself, computes the root of this tree $rt_w$, and sends the verifier $rt_w$. For notational convenience, we will call the algorithm that constructs $rt$ using $hk$ by the name $\mathsf{Merkle}_{hk}(\cdot)$. We will also use $rt$ to denote the combination $rt = (rt_m, rt_w)$.
\item The verifier sends the prover a challenge string $j = (j_1, \dots, j_k)$ which indicates a set of indices for which it wants the prover to reveal $m_{j_1}, \dots, m_{j_k}$.
\item The prover reveals $m_{j_1}, \dots, m_{j_k}$ and also some auxiliary information (namely, the path of hashes in the Merkle tree leading from the root to $m_{j_i}$), which allows the verifier to verify whether or not $(m_{j_1}, \dots, m_{j_k})$ were valid openings and also to compute the verification predicate for the PCPP proof that $R(x,D(\tilde w)) = 1$.
\end{enumerate}
\end{longfbox}

Clarifying information about Merkle trees and their uses in commit-and-open protocols like this can be found in \cite[Section 2.1]{CMSZ}, but for us the important part is the message structure of this protocol. In particular, note that all the messages in this protocol are $\poly \log |x| \cdot \poly \lambda$ in length, where $\lambda$ is the security parameter for $hk$.

\Cref{prot:saok} is similar to Killian's classic protocol \cite{kilian1992note} (instantiated in a form that allows for extraction of the witness $w$): the only change which must be made to the protocol to make it post-quantum (apart from using a collapsing instead of a collision-resistant hash function) is the `extra' Merkle commitment to $w$ in the first message in addition to the commitment to $m = (\tilde w, \pi)$.\footnote{It is not clear whether this additional commitment is necessary; however, the authors of \cite{LMS} were not able to make the state-preserving extraction analysis work without it, although for the analysis in the setting of \cite{CMSZ} the original Killian protocol is sufficient.} However, the \emph{analysis} of this protocol in the post-quantum setting is significantly more subtle than its analysis in the classical setting. In particular, the soundness statement for \Cref{prot:saok}---which says that, given a prover who wins with high probability, there is an efficient extractor that extracts $w$ given black-box access to the prover---is usually proven in the classical setting via a \emph{rewinding} argument, in which the extractor reconstructs some significant fraction of $m$ by choosing an index $j$, obtaining $m_j$, rewinding the prover, choosing another index $j'$, obtaining $m_{j'}$, etc., and finally using the `error robustness' properties of both the PCPP and of the encoding $E$ in order to extract $w$ and to be sure that it satisfies $R(x, \cdot)$ despite the missing indices. Rewinding arguments that are secure against quantum adversaries with quantum auxiliary input tend to be much more difficult than their classical counterparts, because of the possibility that executing the protocol even once will destroy the auxiliary input and prevent rewinding from succeeding.

In \cite{CMSZ}, it was shown how to analyse \Cref{prot:saok} in the post-quantum setting using a clever technique involving alternating projections. However, the analysis of \Cref{prot:saok} presented in \cite{CMSZ} was not particularly composable. If the succinct argument of knowledge in \Cref{prot:saok} is used as a subprotocol in some longer protocol, and there are other tests in the full protocol which follow after the succinct argument of knowledge---this is the case for us---then it may be necessary to ensure that running extraction does not hinder the prover's ability to pass in the remainder of the protocol. To motivate this requirement, consider a situation in which we are trying to design a reduction $\cR$ which reduces the security of a succinct protocol called \textsf{SuccinctProtocol}, in which \Cref{prot:saok} is used as a subprotocol, to the security of a non-succinct protocol \textsf{OriginalProtocol}. \textsf{OriginalProtocol} requires the prover to output a full witness $w$ instead of merely passing in a succinct argument of knowledge for $w$; therefore, given some prover $P$ who is successful in \textsf{SuccinctProtocol}, the reduction has to run extraction on $P$ in order to recover $w$, so that it can succeed with the challenger/verifier for \textsf{OriginalProtocol}. However, if \textsf{OriginalProtocol} contains tests (mirrored in \textsf{SuccinctProtocol}) that happen \emph{after} its prover is supposed to output $w$, the reduction $\cR$ will not necessarily succeed if extracting $w$ destroys $P$'s ability to succeed in the remainder of \textsf{SuccinctProtocol}, because then $\cR$ may be unable to answer the remaining questions in \textsf{OriginalProtocol}.

In order to ensure that, even after extraction has been performed, the prover \emph{continues} to pass with high probability in the rest of the protocol, the prescribed extractor from \cite{CMSZ} was required to measure a projector that corresponded to coherently computing whether or not the verifier would accept in the remainder of the protocol and conditioning on the accept outcome. The issue is that the verifier's final decision predicate might depend on secret information (that is: the \emph{entire} protocol might not be public-coin, even though \Cref{prot:saok} is public-coin), and so an efficient extractor might be unable to do this.

In followup work \cite{LMS}, it was shown how to analyse \Cref{prot:saok} in a more composable way, so that the extractor only needs to measure the projection which corresponds to the verifier of \Cref{prot:saok} accepting (note that this verifier's decision predicate is public, so the extractor will always be able to do this), but even so the extractor's activity is essentially \emph{undetectable} to the prover, meaning that the prover (who uses the extractor's `leftover state' instead of its original state) will continue to succeed in the rest of the protocol (if it succeeded with high probability to begin with) even after extraction has been performed. This condition on the extractor is true by default in the classical setting, but it is nontrivial in the quantum setting. Because this guarantee may be somewhat surprising, we sketch in the next section how this guarantee is shown.

\subsection{\cite{LMS} extraction}

The following (taken largely from~\cite[Definition 9.1]{BKLMMVVY22}) is the formal statement of succinctness and security for~\Cref{prot:saok} for which we will sketch a proof in this section.
\begin{lemma}
\label{thm:lms}
\Cref{prot:saok} satisfies the following properties:
\begin{itemize}
    \item \emph{Succinctness}. When invoked on security parameter $\lambda$ for the hash function family, instance size $|x| = n$, and a relation $R$ decidable in time $T$, the communication complexity of the protocol is $\mathsf{poly}(\lambda, \log T)$. The verifier's computational complexity is $\mathsf{poly}(\lambda, \log T) + \tilde O(n)$.
    \item \emph{$\eps$-state-preserving extraction.} There exists an extractor $E^{(\cdot)}(x, \epsilon)$ with the following properties.
    \begin{itemize}
          \item Efficiency: $E^{(\cdot)}(x, \epsilon)$ runs in time $\poly(n, \secp, 1/\epsilon)$ as a quantum oracle algorithm (with the ability to apply controlled $U$-gates given an oracle $U(\cdot)$), outputting a classical transcript $\tilde \tau$ and a classical string $w$.
          \item State-preserving: Let $\ket{\psi} \in \cA \otimes \cI$ be any $\poly(\secp)$-qubit pure state and let $\rho = \Tr_{\cA}(\ket{\psi}) \in \mathrm{D}(\cI)$.\footnote{In general, the prover's input state on $\cI$ may be entangled with some external register $\cA$, and we ask that computational indistinguishability holds even given $\cA$. Our definition is stated this way for maximal generality, though we remark that the applications in this section do not require indistinguishability in the presence of an entangled external register.} Consider the following two games:
          \begin{itemize}
              \item Game 0 (real): Generate a transcript $\tau$ by running $P^*(\rho_{\cI},x)$ with the honest verifier $V$. Output $\tau$ along with the residual state on $\cA \otimes \cI$.
              \item Game 1 (simulated): Generate a transcript-witness pair $(\tilde \tau,w) \gets E^{P^*(\rho_{\cI},x)}$. Output $\tilde \tau$ and the residual state on $\cA \otimes \cI$.
          \end{itemize}
          Then, we have that the output distributions of Game 0 and Game 1 are computationally $\varepsilon$-indistinguishable to any quantum distinguisher.
          
          \item Extraction correctness: for any $P^*$ as above, the probability that $\tilde \tau$ is an accepting transcript but $w$ is \emph{not} in $R_x$ is at most $\epsilon + \negl(\secp)$. 
          
      \end{itemize}
\end{itemize}
\end{lemma}

In order to describe the extractor which is guaranteed by \Cref{thm:lms}, we firstly fix some notation related to the prover's state and actions in \Cref{prot:saok}. We can model any prover $P^*$ in \Cref{prot:saok} as a process which does the following:
\begin{process}
\label{proc:original}
\end{process}
\begin{enumerate}
\item Send some message $rt$ to the verifier. Let the state that $P^*$ has left over after sending $rt$ be $\ket{\psi_{P^*}}$ (held in a private register $\mathsf{P}$).
\item Receive a challenge $j$ from the verifier. We assume that $j$ is provided as a state $\ket{j}$ in a \emph{message register} $\mathsf{M}$ which is accessible both to the verifier and to the prover.
\item Apply some unitary $U_{P^*}$ which acts on both $\mathsf{P}$ and $\mathsf{M}$.
\item Measure some part of the resulting state (wlog in the standard basis) to get a response $z$, and send this to the verifier.
\end{enumerate}
The verifier will then check its decision predicate $V_{rt}(j,z)$, and accept iff $V_{rt}(j,z)$ evaluates to 1. If \Cref{prot:saok} is a subprotocol in some longer protocol, then the verifier rejects immediately if $V_{rt}(j,z) = 0$.

Now we will `purify' the prover in order to make the \cite{LMS} extraction procedure easier to state. It is easy to see that the state left over in all registers at the end of the process below is exactly equivalent to the state which is left over at the end of \Cref{proc:original}:
\begin{process}
\label{proc:purified}
\end{process}
\begin{enumerate}
\item Send some message $rt$ to the verifier. Let the state that $P^*$ has left over after sending $rt$ be $\ket{\psi_{P^*}}$ (held in a private register $\mathsf{P}$).
\item Prepare a state $\sum_j \ket{j}$ (we will ignore normalisation) that is a \emph{uniform superposition} over challenges in a message register $\mathsf{M}$. Also create a new register $\mathsf{Z}$ initialised to the all zero state which will be used in step (iv).
\item Apply the prover's unitary $U_{P^*}$ jointly to $\mathsf{P}$ and $\mathsf{M}$.
\item Coherently copy (in the standard basis) the part of the state that would have been measured in step (iv) of \Cref{proc:original} above to obtain a superposition over responses $z$ into register $\mathsf{Z}$. For short we will call the unitary that does this copying $\mathsf{CNOT}_z$.
\item Coherently compute the verifier's predicate $V_{rt}(j,z)$ into yet another new register $\mathsf{A}$.
\item Measure registers $\mathsf{M}$ and $\mathsf{Z}$; obtain outcomes $j$ and $z$.
\item Measure register $\mathsf{A}$; obtain $V_{rt}(j,z)$, and accept iff it is 1.
\end{enumerate}
As before, if \Cref{prot:saok} is a subprotocol in some longer protocol, then the verifier rejects immediately if $V_{rt}(j,z)$ measures to 0. In other words, conditioning on continuing in the protocol essentially \emph{projects} into the subspace where $V_{rt}(j,z) = 1$.

Note that steps (vi) and (vii) can be switched, because they act on different registers and therefore commute. As such, we could also have stated \Cref{proc:purified} in the following way:
\begin{process}
\label{proc:switched}
\end{process}
\begin{enumerate}
\item Send some message $rt$ to the verifier. Let the state that $P^*$ has left over after sending $rt$ be $\ket{\psi_{P^*}}$ (held in a private register $\mathsf{P}$).
\item Prepare a state $\sum_j \ket{j}$ (we will ignore normalisation) that is a \emph{uniform superposition} over challenges in a message register $\mathsf{M}$. Also create a new register $\mathsf{Z}$ initialised to the all zero state which will be used in step (iii).
\item This step can be stated in words as a series of substeps:
\begin{enumerate}
\item Apply the prover's unitary $U_{P^*}$ jointly to $\mathsf{P}$ and $\mathsf{M}$.
\item Coherently copy (in the standard basis) the part of the state that would have been measured in step (iv) of \Cref{proc:original} above to obtain a superposition over responses $z$ into register $\mathsf{Z}$. For short we will call the unitary that does this copying $\mathsf{CNOT}_z$.
\item \emph{Project} the state in registers $\mathsf{P}$, $\mathsf{M}$ and $\mathsf{Z}$ into the subspace where $V_{rt}(j,z) = 1$. If the projective measurement results in $V_{rt}(j,z) = 0$, reject.
\item Undo $\mathsf{CNOT}_z$ and $U_{P^*}$ in that order.
\end{enumerate}
In other words, in this step, apply the projective measurement where one of the projectors in the measurement is
\begin{align*}
\Pi_c \coloneqq U_{P^*}^\dagger \mathsf{CNOT}_z^\dagger \Big( \sum_{V_{rt}(j,z) = 1} \1_{\mathsf{P}} \otimes \ket{j,z}_{\mathsf{MZ}} \bra{j,z}_{\mathsf{MZ}} \Big) \mathsf{CNOT}_z U_{P^*}
\end{align*}
and the other is $I - \Pi_c$; reject if the outcome is $I - \Pi_c$.

\item Redo $\mathsf{CNOT}_z$ and $U_{P^*}$ (to counteract substep (iv) of the previous step, step (iii)). Measure registers $\mathsf{M}$ and $\mathsf{Z}$; obtain outcomes $j$ and $z$.
\end{enumerate}
The \cite{LMS} extractor will set up the state which exists after step (iii) in \Cref{proc:switched} (it can do this given an appropriate notion of black-box access to $P^*$: see e.g.~\cite{unruh2016computationally}), and then it will insert an extra step (E) in between steps (iii) and (iv) in \Cref{proc:switched}, in which it extracts $w$ in a way that is \emph{undetectable}, in the sense that the state left over after the whole of \Cref{proc:switched} with (E) inserted between (iii) and (iv) is computationally indistinguishable from the state left over after the whole of \Cref{proc:switched} without (E). We will now describe in two stages how it accomplishes this.

\paragraph{\cite{CMSZ} extraction}
In \cite{CMSZ}, an algorithm was written down that, for any $P^*$ in \Cref{prot:saok}, takes in the state in registers $\mathsf{P}$, $\mathsf{M}$ and $\mathsf{Z}$ after step (ii) in \Cref{proc:switched} and alternates the following two projectors $\Pi_u$ and $\Pi_c$ ($u$ for `uniform' and $c$ for `correct') in order to attempt to \emph{extract} a witness $w$ such that $\mathsf{Merkle}_{hk}(w) = rt_w$, where $rt_w$ is the second part of the message $rt = (rt_m, rt_w)$ that $P^*$ sent in step (i).
\begin{align*}
\Pi_u &= \1_{\mathsf{P}} \otimes \ket{+_m}_{\mathsf{M}}\bra{+_m}_{\mathsf{M}} \otimes \ket{0}\bra{0}_{\mathsf{Z}}, \quad \text{where } \ket{+_m} \coloneqq \sum_{j} \ket{j} \text{ with the appropriate normalisation} \\
\Pi_c &= U_{P^*}^\dagger \mathsf{CNOT}_z^\dagger \Big( \sum_{V_{rt}(j,z) = 1} \1_{\mathsf{P}} \otimes \ket{j,z}_{\mathsf{MZ}} \bra{j,z}_{\mathsf{MZ}} \Big) \mathsf{CNOT}_z U_{P^*}
\end{align*}
Note that $\Pi_c$ coincides with the $\Pi_c$ we wrote down in \Cref{proc:switched}. The analysis involves decomposing the joint space of $\mathsf{P}$, $\mathsf{M}$ and $\mathsf{Z}$ into the Jordan subspaces (for more information about the Jordan decomposition and how it is usually used in quantum computing, see \cite[Section 1.2.3]{Vid20-course})  of $\Pi_u$ and $\Pi_c$. We establish some notation in order to state the extraction guarantee which \cite{CMSZ} prove.

Let the set of Jordan subspaces of $\Pi_u$ and $\Pi_c$ be $\{S_i\}_i$. Let $\ket{u_i}$ be the rank-1 projector associated with $\Pi_u$ in subspace $i$, and let $\ket{c_i}$ be the rank-1 projector associated with $\Pi_c$ in subspace $i$: that is, for any $\ket{\phi} \in S_i$, $\Pi_u \ket{\phi} = \ket{u_i}\bra{u_i} \cdot \ket{\phi}$, and similarly $\Pi_c \ket{\phi} = \ket{c_i}\bra{c_i} \cdot \ket{\phi}$. Since the state in registers $\mathsf{P}$, $\mathsf{M}$ and $\mathsf{Z}$ after step (ii) in \Cref{proc:switched}, which we will name $\ket{\phi_\mathrm{start}}$ for notational convenience, lies \emph{inside} $\Pi_u$, it can be decomposed as
\begin{align*}
\ket{\phi_\mathrm{start}} = \sum_i \alpha_i \ket{u_i}.
\end{align*}

The extraction guarantee which \cite{CMSZ} prove is the following.

\begin{lemma}[\cite{CMSZ}; informal]
\label{lem:cmsz-guarantee}
For any state $\ket{\phi} \in \Pi_u$ with a decomposition $\ket{\phi} = \sum_i (\alpha_i \ket{c_i} + \beta_i \ket{u_i})$, and given the ability to implement the two projectors $\Pi_u$ and $\Pi_c$ (note that $\Pi_c$ implicitly depends on $rt$ through $V_{rt}$), the \cite{CMSZ} rewinding algorithm takes as input $\ket{\phi}$ and outputs $w$ such that $\mathsf{Merkle}_{hk}(w) = rt_w$ with probability $1-\eps$ in time $\poly(1/\eps)$, if the following condition holds:
\begin{align}
\label{eq:cmsz-succ-condition}
\sum_i (| \alpha_i |^2 + | \beta_i |^2) \cdot \1\big[ \:\: | \langle c_i | u_i \rangle |^2 \text{ is non-negligible} \:\: \big] \geq 1 - \mathsf{negl}(\lambda).
\end{align}
\end{lemma}
In other words, if the weight in the superposition $\sum_i (\alpha_i \ket{c_i} + \beta_i \ket{u_i})$ is overwhelmingly in Jordan subspaces where $| \langle c_i | u_i \rangle |^2$ is non-negligible, then \cite{CMSZ} extraction will succeed with probability $1 - \frac{1}{\poly(\lambda)}$.

\paragraph{Coherent \cite{CMSZ} extraction.}
In \cite{LMS}, it is shown how to accomplish `undetectable extraction', i.e.~how to insert an extra step (E) in between steps (iii) and (iv) in \Cref{proc:switched}, corresponding with the execution of the extractor, such that the state left over after the whole of \Cref{proc:switched} with (E) inserted between (iii) and (iv) is computationally indistinguishable from the state left over after the whole of \Cref{proc:switched} without (E).

Lombardi, Ma and Spooner begin by considering the entire \cite{CMSZ} process as a black-box unitary $U_\mathrm{CMSZ}$ which takes as input a state $\ket{\phi} \in \Pi_u$ and outputs some state $\sum_{w'} \ket{w'} \ket{aux_{w'}}$ (ignoring normalisation) which is a superposition over candidate witnesses $w'$ and associated auxiliary states. It can be shown that, if the superposition $\sum_{w'} \ket{w'} \ket{aux_{w'}}$ has its weight only on terms with candidates $w'$ such that $\mathsf{Merkle}_{hk}(w') = rt_w$, then measuring $w'$ at this point is \emph{computationally undetectable}: this is because $\mathsf{Merkle}_{hk}(\cdot)$ is a \emph{collapse-binding commitment}, and the statement that no efficient algorithm can tell the difference between $\sum_{w'} \ket{w'}_\mathsf{W} \ket{aux_{w'}}_\mathsf{aux}$ and $\mathsf{Meas}_\mathsf{W} \big( \sum_{w'} \ket{w'}_\mathsf{W} \ket{aux_{w'}}_\mathsf{aux} \big)$, when the superposition is only over $w'$ such that $\mathsf{Merkle}_{hk}(w') = rt_w$, is precisely the definition of collapse-binding. (This is the only step in the analysis where the `extra' commitment to $w$ which the prover sends in step (i) of \Cref{prot:saok} in addition to the commitment to the PCPP is used.)

In other words, if the success condition for \cite{CMSZ} rewinding stated in \Cref{lem:cmsz-guarantee} is true, namely, the starting state $\ket{\phi} \in \Pi_u, \ket{\phi} = \sum_i \alpha_i \ket{u_i}$ on which $U_\mathrm{CMSZ}$ is run is such that
\begin{align*}
\sum_i | \alpha_i |^2 \cdot \1\big[ \:\: | \langle c_i | u_i \rangle |^2 \text{ is non-negligible} \:\: \big] \geq 1 - \mathsf{negl}(\lambda),
\end{align*}
then (except with at most $\frac{1}{\poly(\lambda)}$ probability for any $\poly$ of our choice) the extractor can run $U_\mathrm{CMSZ}$ and then measure the $\mathsf{W}$ register in the resulting superposition, and this will output a successful witness candidate $w'$ such that $\mathsf{Merkle}_{hk}(w') = rt_w$ and the measurement will be undetectable to the prover.

Our first hope might be that, if $P^*$ succeeds in \Cref{prot:saok} with high probability, then the CMSZ success condition \Cref{eq:cmsz-succ-condition} is true for the total state lying in registers $\mathsf{P}$, $\mathsf{M}$ and $\mathsf{Z}$ after step (ii) in \Cref{proc:switched}. Unfortunately, this turns out \emph{not} to be (necessarily) true; we will not go into why here.

However, one can guarantee the condition for the success of the CMSZ procedure for the total state in registers $\mathsf{P}$, $\mathsf{M}$ and $\mathsf{Z}$ after step (iii) of \Cref{proc:switched}. Note that the state after step (iii) is precisely the state after step (ii) after $\Pi_c$ has been applied to it once. The effect of applying $\Pi_c$ to a superposition $\sum_i \alpha_i \ket{u_i}$ is as follows:
\begin{align}
\Pi_c \big( \sum_i \alpha_i \ket{u_i} \big) &\propto \sum_i \alpha_i \big( \ket{c_i} \bra{c_i} \cdot \ket{u_i} \big) \label{eq:unnormalised-pi-c} \\
&= \sum_i \big(\alpha_i \langle c_i | u_i \rangle \big) \: \ket{c_i}
\end{align}
Note that this (subnormalised) state now has a decomposition of the form $\sum_i \beta_i \ket{c_i}$; and, moreover, if $|\langle c_i | u_i \rangle|^2$ is negligible, then the weight on $\ket{c_i}$ (i.e.~the squared norm of the coefficient of $\ket{c_i}$) in the superposition will be negligibly small. For brevity's sake we will ignore the issue of the normalisation, but the intuition that the coefficients where $|\langle c_i | u_i \rangle|^2$ is small get suppressed holds even in the presence of the correct renormalisation.

Therefore, the entire `undetectable extraction' procedure, given black-box access to a prover $P^*$ for \Cref{prot:saok}, is as follows:

\begin{process}[extraction]
\label{proc:extractor}
\end{process}
\begin{enumerate}
\item Run $P^*$ normally in order to generate its first message $rt$ to the verifier. Let the state that $P^*$ has left over after sending $rt$ be $\ket{\psi_{P^*}}$ (held in a private register $\mathsf{P}$).
\item Prepare a state $\sum_j \ket{j}$ (we will ignore normalisation) that is a \emph{uniform superposition} over challenges in a message register $\mathsf{M}$. Also create a new register $\mathsf{Z}$ initialised to the all zero state which will be used in step (iii).
\item This step can be stated in words as a series of substeps:
\begin{enumerate}
\item Apply the prover's unitary $U_{P^*}$ jointly to $\mathsf{P}$ and $\mathsf{M}$.
\item Coherently copy (in the standard basis) the part of the state that would have been measured in step (iv) of \Cref{proc:original} above to obtain a superposition over responses $z$ into register $\mathsf{Z}$. For short we will call the unitary that does this copying $\mathsf{CNOT}_z$.
\item \emph{Project} the state in registers $\mathsf{P}$, $\mathsf{M}$ and $\mathsf{Z}$ into the subspace where $V_{rt}(j,z) = 1$. If the projective measurement results in $V_{rt}(j,z) = 0$, output $\mathsf{fail}$.
\item Undo $\mathsf{CNOT}_z$ and $U_{P^*}$ in that order.
\end{enumerate}
In other words, in this step, apply the projective measurement where one of the projectors in the measurement is
\begin{align*}
\Pi_c \coloneqq U_{P^*}^\dagger \mathsf{CNOT}_z^\dagger \Big( \sum_{V_{rt}(j,z) = 1} \1_{\mathsf{P}} \otimes \ket{j,z}_{\mathsf{MZ}} \bra{j,z}_{\mathsf{MZ}} \Big) \mathsf{CNOT}_z U_{P^*}
\end{align*}
and the other is $I - \Pi_c$; output $\mathsf{fail}$ if the outcome is $I - \Pi_c$.
\item[(E)] Perform $U_\mathrm{CMSZ}$, the unitary that does the CMSZ rewinding procedure coherently, which results in a state $\sum_{w'} \ket{w'}_\mathsf{W} \ket{aux_{w'}}_\mathsf{aux}$. Measure the $\mathsf{W}$ register and check if the outcome is a string $w$ such that $\mathsf{Merkle}_{hk}(w) = rt_w$. If no, output $\mathsf{fail}$. If yes, apply $U_\mathrm{CMSZ}^\dagger$.
\item Redo $\mathsf{CNOT}_z$ and $U_{P^*}$ (to counteract substep (iv) of step (iii)). Measure registers $\mathsf{M}$ and $\mathsf{Z}$; obtain outcomes $j$ and $z$.
\end{enumerate}
Conditioned on the extractor not outputting $\mathsf{fail}$, the state that remains in registers $\mathsf{P}$ and $\mathsf{M}$ (note that $\mathsf{Z}$ is a work register for the extractor) after \Cref{proc:extractor} is computationally indistinguishable from the state which would remain in those same registers after a real (successful) execution of \Cref{prot:saok}. Note that the probability that the extractor outputs $\mathsf{fail}$ in step (iii) is the same as the probability that the prover fails the real execution of \Cref{prot:saok}, and that the probability the extractor outputs $\mathsf{fail}$ in step (E) is at most $\eps$ if the extractor runs for time $\poly(1/\eps)$. As such, if \Cref{prot:saok} is a sub-protocol in a longer protocol, then extraction can be performed while affecting the prover's probability of passing in the rest of the longer protocol by only $\eps + \mathsf{negl}(\lambda)$.

\subsection{A fully succinct version of \Cref{prot:main}}

The intuition for compiling \Cref{prot:main} (which is already question-succinct) into a fully succinct protocol is as follows. Every time the prover $P$ is supposed to send a message to the verifier $V$ in \Cref{prot:main}, we ask the prover $\tilde P$ in the succinct version of \Cref{prot:main} to \emph{commit} (using some succinct computationally binding commitment) to the answer that $P$ would have provided, and then execute \Cref{prot:saok} in order to prove succinctly that it `knows' a valid opening to that commitment. At the end of the protocol, $V$ in \Cref{prot:main} takes the answers that it receives and evaluates a decision predicate. Since $\tilde V$, the succinct protocol's verifier, does not have the prover's answers (instead it only experienced a short interactive proof that the prover `knew' answers of some description), it cannot evaluate the decision predicate for itself. Instead, the verifier $\tilde V$ reveals all of its secret randomness at the end of the protocol, and asks $\tilde P$ to execute \Cref{prot:saok} one more time to prove that it knows full-length answers which are `consistent' with its ($\tilde P$'s) earlier commitments, and moreover that these full-length answers satisfy $V$'s decision predicate.

More specifically, the fully succinct version of \Cref{prot:main} is as follows:

\begin{longfbox}[breakable=true, padding=1em, padding-right=1.8em, padding-top=1.2em, margin-top=1em, margin-bottom=1em]
\begin{protocol}[Fully succinct version of \Cref{prot:main}]
\label{prot:succinct}	
\end{protocol}
Inputs: an instance $x$, a description of the verification algorithm $C$ of a promise problem $A \in \mathsf{QMA}$, a security parameter $\lambda$, and (for the honest prover $\tilde P$) polynomially many copies of a witness that $x \in A_\mathrm{yes}$. 
\begin{enumerate}
	\item \textbf{Phase 1: hash key, `Alice' question and answer}
	\begin{enumerate}
	\item The verifier $\tilde V$ of \Cref{prot:succinct} samples a hash key $hk$ from some collapsing hash function family. It also runs $V$, the verifier of \Cref{prot:main}, on the inputs $x$, $C$, $\lambda$ in order to generate the two questions of \Cref{prot:main}, one of which is a ciphertext $\hat q_1$ encrypting some `Alice' question $q_1$, and the other of which is a plaintext `Bob' question $q_2$. $\tilde V$ sends $hk$ and $\hat q_1$ to the prover $\tilde P.$
	\item Honest $\tilde P$ responds with a \emph{succinct commitment}\footnote{Succinct computationally binding commitments can be constructed from collapsing hash functions using Merkle trees.} to the answer that $P$ would have given $V$. We will call $\tilde P$'s response here $\mathsf{com}_1$.
	\item $\tilde P$ and $\tilde V$ execute \Cref{prot:saok} so that $\tilde P$ can prove that it knows a \emph{valid opening} to its commitment $\mathsf{com}_1$. The property of being a valid opening can be phrased as an $\mathsf{NP}$ relation $R_1(\mathsf{com}_1, w_1) =  (\mathsf{Merkle}_{hk}(w_1) = \mathsf{com}_1) $.
	\end{enumerate}
	\item \textbf{Phase 2: `Bob' question and answer}
	\begin{enumerate}
	\item $\tilde V$ sends $q_2$ (which it generated earlier along with $\hat q_1$) to $\tilde P$.
	\item Honest $\tilde P$ responds with a succinct commitment to the answer that $P$ would have given $V$. We will call $\tilde P$'s response here $\mathsf{com}_2$.
	\item $\tilde P$ and $\tilde V$ execute \Cref{prot:saok} so that $\tilde P$ can prove that it knows a valid opening to its commitment $\mathsf{com}_2$. As before, this can be phrased as an $\mathsf{NP}$ relation $R_2(\mathsf{com}_2, w_2) = ( \mathsf{Merkle}_{hk}(w_2) = \mathsf{com}_2)$.
	\end{enumerate}
	\item \textbf{Phase 3: Proof of knowledge that $V$ would have accepted}
	\begin{enumerate}
	\item $\tilde V$ now reveals the secret key $sk$ of the quantum homomorphic encryption scheme under which $\hat q_1$ was encrypted.
	\item $\tilde V$ and $\tilde P$ execute \Cref{prot:saok} so that $\tilde P$ can prove that it knows $\hat a_1$ and $a_2$ such that:
	\begin{enumerate}
	\item $\hat a_1$ is a valid opening of $\mathsf{com}_1$,
	\item $a_2$ is a valid opening of $\mathsf{com}_2$,
	\item $V(x,C,1^\lambda,\mathsf{Dec}_{sk}(\hat q_1), \mathsf{Dec}_{sk}(\hat a_1), q_2, a_2) = 1$.
	\end{enumerate}
    Formally, this can be expressed as an $\mathsf{NP}$ relation $R_3(\underbrace{(\mathsf{com}_1, \mathsf{com}_2, x, C, 1^\lambda, \hat q_1, q_2, sk)}_{\text{instance}}, \underbrace{(\hat{a}_1, a_2)}_{\text{witness}})$, in the standard way.
	\end{enumerate}
\end{enumerate}
\end{longfbox}

\begin{theorem}\label{thm:succinct-soundness}
There exists a negligible function $\mathsf{negl}(\cdot)$ such that, for any $(x,C,\lambda)$, and any $\eps = \frac{1}{\poly(\lambda)}$, if there exists an efficient prover $\tilde P$ which causes $\tilde V$ (the verifier of \Cref{prot:succinct}) to accept on input $(x,C,\lambda)$ with probability $\tilde p$, then there exists an efficient prover $P$ which causes $V$ (the verifier of \Cref{prot:main}) to accept $(x,C,\lambda)$ with probability $p$, such that $p \geq \tilde p - \mathsf{negl}(\lambda) - O(\eps)$.
\end{theorem}
\begin{proof}[Proof sketch]
The analysis proceeds via reduction to \Cref{prot:main}. In particular, we consider a reduction $\cR$ that plays the part of the verifier $\tilde V$ with $\tilde P$ in \Cref{prot:succinct} and the part of the prover $P$ with $V$ in \Cref{prot:main}, and show that $\cR$ passes in \Cref{prot:main} with about the same probability that $\tilde P$ passes in \Cref{prot:succinct}. We assume that $\cR$ gets the same kind of black-box access to $\tilde P$ that the extractor of \Cref{proc:extractor} does.

$\cR$ does the following:
\begin{enumerate}

\item \textbf{Phase 1: `Alice' question and answer}
\begin{enumerate}
\item Receives $\hat q_1$ from $V$; samples $hk$ for itself; inputs $hk, \hat q_1$ into $\tilde P$.
\item Gets $\mathsf{com}_1$ from $\tilde P$.
\item Gets $rt_1$, the first message of the first execution of \Cref{prot:saok} in \Cref{prot:succinct}, from $\tilde P$.
\item Runs \Cref{proc:extractor} on $\tilde P$ in order to extract a witness $w_1$ such that $w_1$ is a valid opening for $\mathsf{com}_1$. (Note that, in the case of honest $\tilde P$, $w_1$ is an encryption of an `Alice' answer $a_1$.)
\item Returns $w_1$ to $V$ as its `Alice' answer.
\end{enumerate}

\item \textbf{Phase 2: `Bob' question and answer}
\begin{enumerate}
\item Receives $q_2$ from $V$; inputs $q_2$ into $\tilde P$.
\item Gets $\mathsf{com}_2$ from $\tilde P$.
\item Gets $rt_2$, the first message of the second execution of \Cref{prot:saok} in \Cref{prot:succinct}, from $\tilde P$.
\item Runs \Cref{proc:extractor} on $\tilde P$ in order to extract a witness $w_2$ such that $w_2$ is a valid opening for $\mathsf{com}_2$. (Note that, in the case of honest $\tilde P$, $w_2$ is a `Bob' answer $a_2$.)
\item Returns $w_2$ to $V$ as its `Bob' answer.
\end{enumerate}

\item $\cR$ ignores Phase 3 of \Cref{prot:succinct} and aborts after Phase 2.
\end{enumerate}

Note that $\cR$ does not use the third execution of \Cref{prot:saok} in \Cref{prot:succinct}; in fact, $\cR$ cannot continue playing with $\tilde P$ after Phase 2 because it does not know the secret key for the homomorphic encryption which $V$ generated. However, the third execution of \Cref{prot:saok} in \Cref{prot:succinct} will be used in the analysis. In particular, in order to prove that $\cR$ succeeds in \Cref{prot:main} with about the same probability that $\tilde P$ succeeds in \Cref{prot:succinct}, we will consider the following mental experiment. Consider an $\cR'$ which behaves identically to $\cR$ in Phase 1 and Phase 2 of \Cref{prot:succinct}, and in Phase 3, instead of aborting, does an inefficient brute force search for the secret key of the encryption.
\begin{process}[A mental experiment: $\cR'$'s execution]
\label{proc:mental-experiment}	
\end{process}

\begin{enumerate}
\item \textbf{Phase 1: `Alice' question and answer}

\begin{enumerate}
\item Receives $\hat q_1$ from $V$; samples $hk$ for itself; inputs $hk, \hat q_1$ into $\tilde P$.
\item Gets $\mathsf{com}_1$ from $\tilde P$.
\item Gets $rt_1$, the first message of the first execution of \Cref{prot:saok} in \Cref{prot:succinct}, from $\tilde P$.
\item Runs \Cref{proc:extractor} on $\tilde P$ in order to extract a witness $w_1$ such that $w_1$ is a valid opening for $\mathsf{com}_1$. (Note that, in the case of honest $\tilde P$, $w_1$ is an encryption of an `Alice' answer $a_1$.)
\item Returns $w_1$ to $V$ as its `Alice' answer.
\end{enumerate}

\item \textbf{Phase 2: `Bob' question and answer}
\begin{enumerate}
\item Receives $q_2$ from $V$; inputs $q_2$ into $\tilde P$.
\item Gets $\mathsf{com}_2$ from $\tilde P$.
\item Gets $rt_2$, the first message of the second execution of \Cref{prot:saok} in \Cref{prot:succinct}, from $\tilde P$.
\item Runs \Cref{proc:extractor} on $\tilde P$ in order to extract a witness $w_2$ such that $w_2$ is a valid opening for $\mathsf{com}_2$. (Note that, in the case of honest $\tilde P$, $w_2$ is a `Bob' answer $a_2$.)
\item Returns $w_2$ to $V$ as its `Bob' answer.
\end{enumerate}

\item \textbf{Phase 3: proof of knowledge that $V$ would have accepted}
\begin{enumerate}
\item Finds the secret key $sk$ of the quantum homomorphic encryption scheme being used by $V$ (by brute force, let's say) and inputs this key into $\tilde P$.
\item Gets $rt_3$, the first message of the third and last execution of \Cref{prot:saok} in \Cref{prot:succinct}, from $\tilde P$.
\item Runs \Cref{proc:extractor} on $\tilde P$ in order to extract a witness $w_3$ such that $w_3 = (\hat a_1, a_2)$, where
	\begin{enumerate}
	\item $\hat a_1$ is a valid opening of $\mathsf{com}_1$,
	\item $a_2$ is a valid opening of $\mathsf{com}_2$,
	\item $V(\mathsf{Dec}_{sk}(\hat q_1), \mathsf{Dec}_{sk}(\hat a_1), q_2, a_2) = 1$.
	\end{enumerate}
\end{enumerate}

\end{enumerate}
$\cR'$ is, of course, inefficient; however, it is efficient \emph{given the secret key of the homomorphic encryption}. Moreover, suppose we condition on all three extractions in \Cref{proc:mental-experiment} succeeding. Then, if $w_1$ (the witness that $\cR'$ extracts from $\tilde P$ during Phase 1) is equal to $\hat a_1$ (which is our name for the first part of the witness $w_3$ extracted by $\cR'$ from $\tilde P$ during Phase 3) except with negligible probability in the mental experiment, and similarly $w_2 = a_2$ except with negligible probability, then it is the case that $\cR$ will be accepted by $V$ with probability at least $1 - \mathsf{negl}(\lambda)$. (This is because, if extraction succeeds, then $w_3$ consists of a pair of answers which, by definition, causes $V$ to accept.) The event that all three extractions succeed happens with probability at least $\tilde p - 3\eps - \mathsf{negl}(\lambda)$ by a union bound, where $\tilde p$ is the probability that $\tilde P$ passes overall. Therefore, it is sufficient to prove the following lemma:

\begin{lemma}
In \Cref{proc:mental-experiment}, except with negligible probability, $w_1 = \hat a_1$ and $w_2 = a_2$.
\end{lemma}
\begin{proof}[Proof sketch]
Consider an adversary $\cA$ for the security for the collapsing hash function family from which $hk$ was drawn. $\cA$ can generate its own homomorphic encryption keys and then simulate the entire interaction between $V$, $\cR'$ and $\tilde P$ efficiently. If $w_1 \neq \hat a_1$ or $w_2 \neq \hat a_2$, then $\cA$ has generated two valid openings of a computationally binding\footnote{We note that, in this particular situation, classical binding actually suffices; collapse-binding is not necessary.} commitment (either two valid openings to $\mathsf{com}_1$ or two valid openings to $\mathsf{com}_2$); this can only happen with negligible probability by the definition of computational binding. The claim follows.
\end{proof}

Therefore, $w_1 = \hat a_1$ and $w_2 = a_2$ except with negligible probability, and so $\cR$ will be accepted by $V$ with probability at least $\tilde p - \mathsf{negl}(\lambda) - O(\eps)$.
\end{proof}

\begin{theorem}\label{thm:main-succinct-technical}
The protocol \Cref{prot:succinct} is a succinct argument system for $\QMA$ assuming a $\mathsf{QHE}$ scheme satisfying the definition given in \Cref{def:QHE-aux}, and assuming the existence of collapsing hash functions (see \cite[Section 3.6]{BKLMMVVY22} for a definition of these objects).
    
More precisely, let $V$ be the verifier, and $P$ be the honest prover described in the protocol. Then for any promise problem $A = (A_{yes}, A_{no})$ in $\QMA$ with verification algorithm $C$, the following hold:
    \begin{itemize}
        \item \textbf{Completeness:} Let $x \in A_{yes}$ , and let $\ket{\psi}$ be an accepting $\QMA$ witness for $x$. Then the verifier $V$ on input $(x, C,1^\lambda)$, in interaction with the honest prover $P$ on input $(x,C, \ket{\psi}^{\poly n}, 1^\lambda)$, accepts with probability $\geq 1 - \negl(n)$.  
        \item \textbf{Soundness:} Let $x \in A_{no}$, and let $\ket{\phi}$ be \emph{any} state on $\poly(n + \lambda)$ qubits. Then the verifier $V$ on input $(x,C, 1^\lambda)$, in interaction with any QPT prover $P^*$ on input $(x,C, \ket{\phi}, 1^\lambda)$, accepts with probability at most $s'$ for some universal constant $s' < 1$.
        \item \textbf{Succinctness:} On any input $x$ and for security parameter $\lambda$, the runtime of the verifier $V$ is $\tilde{O}(n) + O(\poly\log n \cdot \poly(\lambda))$, and the total number of bits communicated between the prover and verifier is $O(\poly\log n \cdot \poly(\lambda))$.
    \end{itemize}
\end{theorem}
\begin{proof}
    Completeness, soundness, and the communication bound in succinctness all follow from \Cref{thm:main} and \Cref{thm:succinct-soundness}. For the runtime part of succinctness, we must be slightly more careful. Recall that the verifier $V$ of \Cref{prot:main} satisfies the obliviousness and efficiency properties that (1) the challenges it generates depend only on the algorithm $C$ and on the \emph{length} of the instance $x$, not on the instance itself, and (2) all the challenges can be generated in time $\poly\log n \cdot \poly \lambda + \tilde O(n)$. Specifically, the questions generated by $V$ were of the following form:
    \begin{itemize}
    \item Encryptions of uniform randomness of some predetermined length scaling as $\poly\log(|x|)$, or
    \item Non-encrypted uniform randomness of some predetermined length scaling as $\poly\log(|x|)$,
    \end{itemize}
    where the length depends on $|x|$ and $C$.
    For our current purposes, we would like to claim that this means that $V$ can generate its challenges in time $\tilde{O}(|x|) + \poly \log |x| \cdot \poly(\lambda)$, and in particular does \emph{not} need to run the potentially costly reduction from $x$ to a Hamiltonian problem $(H, \alpha, \beta)$. In order to ensure that this is the case, let us specify that $C$ is given as a description of the algorithm, \emph{together with} an explicit polynomial upper-bounding the runtime of $C$. Then, the length of the randomness to be generated in the challenges depends only on the number of qubits and terms in the Hamiltonian $H$, which in turn can be efficiently computed given $|x|$, and the explicit polynomial. This means it can be computed in $\tilde{O}(|x|)$ time.

    Now, let us calculate the runtime for the succinct verifier $\tilde{V}$ in Phase 1. First, the runtime to generate the questions $\hat{q}_1, q_2$ is $\tilde{O}(n) + \poly\log n \cdot \poly\lambda$ by the previous paragraph. Moreover, since these questions come from a question-succinct protocol, their length is $\poly \log n \cdot \poly \lambda$---we will need this later when we analyse Phase 3. Next, let us compute the runtime for the succinct argument of knowledge at the end of Phase 1. The relation $R_1$ has instance length equal to $\ell_1 = |\mathsf{com}_1|  = \poly(\log n) \cdot \poly(\lambda)$, and $R_1$ can be decided in time $T_1 = \poly(n, \lambda)$. Thus, by \Cref{thm:lms}, the runtime of the succinct argument is $\tilde{O}(\ell_1) + \poly(\log(T_1)) \cdot \poly(\lambda) = \poly(\log n) \cdot \poly(\lambda)$.

    Now we move on to Phases 2 and 3. Here we can see that the runtime is dominated by the runtime of the succinct arguments of knowledge. In Phase 2, the runtime is identical to that of Phase 1. For Phase 3, the relation $R_3$ has instance length 
    \begin{align*}
    \ell_3 &= |\mathsf{com}_1| + |\mathsf{com}_2| + |x| + |C| + \lambda + |\hat{q}_1| + |q_2| + |sk| \\
    &= n + \poly(\lambda) \cdot \poly \log n.
    \end{align*}
    It can be decided in time $T_3 = \poly(n, \lambda)$. Thus, again applying~\Cref{thm:lms}, the runtime of the succinct argument is $\tilde{O}(\ell_3) + \poly(\log T_3) \cdot \poly \lambda = \tilde{O}(n) + \poly(\log n) \cdot \poly(\lambda)$.
    
    All together, the total runtime of $\tilde{V}$ is  $\tilde{O}(n) + \poly\log(n) \cdot \poly(\lambda)$ as desired. 
\end{proof}

\bibliographystyle{myhalpha}
\bibliography{main}

\end{document}